\documentclass[twocolumn,natbib]{svjour3}

\smartqed

\usepackage[T1]{fontenc}% for correct hyphenation and T1 encoding
\usepackage{lmodern}% latin modern font
\usepackage[american]{babel}% for American English
\usepackage{microtype}% for character protrusion and font expansion (only with pdflatex)

\usepackage[
hyperref,% for working together with hyperref
table% for coloring matrix entries
]{xcolor}% color package; load before tocstyle

\usepackage[nouppercase]{scrpage2}% package for page style with not only uppercase letters in the head

\usepackage{mathptmx}

\usepackage{eso-pic}% for setting up background
\usepackage{rotating}% for landscape tables
\usepackage{amsmath}% sophisticated mathematical formulas with amstex (includes \text{})
\usepackage{mathtools}% fix amsmath deficiencies
\usepackage{amssymb}% sophisticated mathematical symbols with amstex
\usepackage{etoolbox}% for preto
\usepackage{bm}% for bold math symbols
\usepackage{dsfont}
\usepackage{enumitem}% sophisticated itemize, enumerate
\usepackage{graphicx}% for including figures
\usepackage{grffile}% extending the file name processing of graphicx
\usepackage{tikz}% sophisticated graphics package
\usepackage{wrapfig}% for wrapping figures and tables
\usepackage{tabularx}% for tables
\usepackage{siunitx}% for typesetting numbers with units; also for alignment in tables
\usepackage{booktabs}% for table layout
\usepackage{multirow}% for table entries spanning several rows
\usepackage{vruler}% for showing line numbers
\usepackage{fancyvrb}% for listings
\usepackage{listings}% for including source code
\usepackage{csquotes}% recommended by biblatex
\usepackage[
hypertexnames=false,% for correct links (duplicate-error solution)
setpagesize=false,% necessary in order to not change text-/paperformat for the document
pdfborder={0 0 0},% removes border around links
pdfstartview=Fit,% fit page to pdf viewer
bookmarksopen=true,% expand all bookmarks in adobe reader
bookmarksnumbered=true% number bookmarks in adobe reader
]{hyperref}% all links stay black and are thus invisible
\xdefinecolor{lightgray}{RGB}{247, 247, 247}% R's gray97; #F7F7F7
\xdefinecolor{semilightgray}{RGB}{240, 240, 240}% R's gray94; #F0F0F0
\xdefinecolor{middlegray}{RGB}{127, 127, 127}% R's gray50; #7F7F7F
\xdefinecolor{blue}{RGB}{58, 95, 205}% R's royalblue3; #3A5FCD
\xdefinecolor{deepskyblue}{RGB}{0, 154, 205}% R's deepskyblue3; #009ACD
\xdefinecolor{chocolate}{RGB}{205, 102, 29}% R's chocolate3; #CD661D

\setlist{% general list settings (enumitem's itemize, enumerate, and description)
  align=left,% left-aligned enumerate
  labelsep=*,% align all item bodies vertically
  leftmargin=*,% begin item content at a variable place depending on the item; use \parindent to set it exactly to \parindent
  topsep=1mm,% space before enumerate
  itemsep=0mm% space between enumerate items
}
\newcommand*{\mysquare}{\rule[0.18em]{0.36em}{0.36em}}
\newcommand*{\mytriangle}{\raisebox{0.12em}{\resizebox{0.48em}{0.48em}{$\blacktriangleright$}}}
\newcommand*{\mybar}{\rule[0.32em]{0.62em}{0.08em}}
\setlist[itemize,1]{label={\mysquare\ }}% itemize label on level 1
\setlist[itemize,2]{label={\mytriangle\ }}% itemize label on level 2
\setlist[itemize,3]{label={\mybar\ }}% itemize label on level 3
\setlist[enumerate,1]{label=\arabic*)}% enumerate label on level 1
\setlist[enumerate,2]{label=\arabic{enumi}.\arabic*)}% enumerate label on level 2
\setlist[enumerate,3]{label=\arabic{enumi}.\arabic{enumii}.\arabic*)}% enumerate label on level 3

\lstdefinestyle{input}{
  backgroundcolor=\color{semilightgray},% background color
  commentstyle=\itshape\color{chocolate},% comment style
  keywordstyle=\color{blue},% keyword style
  stringstyle=\color{deepskyblue},% string style
  numbers=left,% display line numbers on the left side
  numberstyle=\color{middlegray}\tiny% use small line numbers
}
\lstdefinestyle{output}{
  backgroundcolor=\color{lightgray}% background color
}

\lstdefinestyle{Lstyle}{
  language=[LaTeX]TeX,% set programming language
  texcs={},% texcs
  otherkeywords={}% undefine otherkeywords
}

\lstnewenvironment{Linput}[1][]{%
  \lstset{style=input, style=Lstyle}
  #1% content
}{\vspace{-0.25\baselineskip}}% note: -\baselineskip leads to no space

\lstnewenvironment{Loutput}[1][]{%
  \lstset{style=output, style=Lstyle}
  #1% content
}{\vspace{-0.25\baselineskip}}% note: -\baselineskip leads to no space

\lstdefinestyle{Rstyle}{
  language=R,% set programming language
  literate={<-}{{$\bm\leftarrow$}}2{<<-}{{$\bm{\mathrel{\bm\leftarrow\mkern-14mu\leftarrow}$}}}2{<=}{{$\bm\le$}}2{>=}{{$\bm\ge$}}2{!=}{{$\bm\neq$}}2,% item to replace, text, length of chars
  keywords={if, else, repeat, while, function, for, in, next, break},% keywords; see R language manual, /usr/local/texlive/2012/texmf-dist/tex/latex/listings/lstlang3.sty
  otherkeywords={}% undefine otherkeywords to remove !,!=,~,$,*,\&,\%/\%,\%*\%,\%\%,<-,<<-,_,/
}

\lstnewenvironment{Sinput}[1][]{%
  \lstset{style=input, style=Rstyle}
  #1% content
}{\vspace{-0.25\baselineskip}}% note: -\baselineskip leads to no space

\lstnewenvironment{Soutput}[1][]{%
  \lstset{style=output, style=Rstyle}
  #1% content
}{\vspace{-0.25\baselineskip}}% note: -\baselineskip leads to no space

\newcommand{\T}{^{\top}}

\newcommand*{\I}{\mathds{1}}
\newcommand*{\IN}{\mathbb{N}}

\newcommand*{\IR}{\mathbb{R}}

\newcommand*{\IP}{\mathbb{P}}
\newcommand*{\IE}{\mathbb{E}}

\newcommand*{\psii}{{\psi^{-1}}}

\newcommand*{\LS}{\mathcal{LS}}
\newcommand*{\LSi}{\LS^{-1}}

\newcommand*{\D}{\operatorname{D}}

\newcommand*{\Exp}{\operatorname{Exp}}

\newcommand*{\LN}{\operatorname{LN}}

\newcommand*{\E}{\operatorname{E}}
\newcommand*{\U}{\operatorname{U}}

\newcommand*{\VaR}{\operatorname{VaR}}
\newcommand*{\ES}{\operatorname{ES}}

\newcommand*{\R}{\textsf{R}}
, basicstyle=\ttfamily]}% for inline LaTeX code
\begin{document}
% watermark
% \AddToShipoutPicture{% set up watermark on every page
%   \begin{tikzpicture}[remember picture, overlay]
%     \node[scale=9,rotate=54.74,color=black!18] at (current page.center){\normalfont\sffamily Draft};
%   \end{tikzpicture}%
% }
% ruler
%\setvruler[10pt][1][1][4][1][0pt][0pt][-30pt][\textheight]%

\title{Quasi-random numbers for copula models}

\author{Mathieu\ Cambou \and Marius\ Hofert \and Christiane\ Lemieux}
\institute{Mathieu\ Cambou \at Institute\ of\ Mathematics,\ Station\ 8,\ EPFL,\
  1015\ Lausanne,\ Switzerland, \email{mathieu.cambou@epfl.ch} \and Marius\ Hofert \at
Department\ of\ Statistics\ and Actuarial\ Science,\ University of Waterloo,\ Canada,\ \email{marius.hofert@uwaterloo.ca} \and Christiane\ Lemieux \at
Department\ of\ Statistics\ and Actuarial\ Science,\ University of Waterloo,\ Canada,\ \email{clemieux@uwaterloo.ca}}

\maketitle

\begin{abstract}
The present work addresses the question how sampling algorithms
for commonly applied copula models can be adapted to account for quasi-random
numbers. Besides sampling methods such as
the conditional distribution method (based on a one-to-one
transformation), it is also shown that typically faster sampling
methods (based on stochastic representations) can be used to improve
upon classical Monte Carlo methods when pseudo-random number
generators are replaced by quasi-random number generators. This opens
the door to quasi-random numbers for models well beyond independent margins or the
multivariate normal distribution. Detailed examples (in the context of finance and insurance),
illustrations and simulations are given and software has been developed and
provided in the \R\ packages \texttt{copula} and \texttt{qrng}.
\end{abstract}

\section{Introduction}
In many applications, in particular in finance and insurance, the quantities of
interest can be written as $\IE[\Psi_0(\bm{X})]$, where
$\bm{X}=(X_1,\dots,X_d):\Omega\rightarrow\IR^d$ is a random vector with
distribution function $H$ on a
probability space $(\Omega,\mathcal{F},\IP)$ and
$\Psi_0:\IR^d\rightarrow\IR$ is a measurable function. The components of
$\bm{X}$ are typically dependent. To account for this dependence, the distribution
of $\bm{X}$ can be modeled by
\begin{align}
  H(\bm{x}) = C(F_1(x_1),\dots,F_d(x_d)),\quad \bm{x}\in\IR^d,\label{cop:mod}
\end{align}
where $F_j(x)=\IP(X_j\le x)$, $j\in\{1,\dots,d\}$, are the marginal distribution
functions of $H$ and $C:[0,1]^d\rightarrow[0,1]$ is a \emph{copula}, i.e., a
distribution function with standard uniform univariate margins; for an
introduction to copulas, see \cite[Chapter~5]{qrm}, \cite{nelsen} or \cite{joe2014}. A copula
model such as \eqref{cop:mod} allows one to separate the dependence structure
from the marginal distributions. This is especially useful in the context of
model building and sampling in the case where $\IE[\Psi_0(\bm{X})]$ mainly depends on the
dependence between the components of $\bm{X}$, so on $C$; for examples of this
type, see Section~\ref{sec:num}.

In terms of copula model \eqref{cop:mod}, we may then write
\begin{align*}
  \IE[\Psi_0(\bm{X})] = \IE[\Psi(\bm{U})]
\end{align*}
where $\bm{U}=(U_1,\dots,U_d):\Omega\rightarrow\IR^d$ is a random vector with
distribution function $C$, $\Psi:[0,1]^d\rightarrow \IR$ is defined as
\begin{align*}
  \Psi(u_1,\dots,u_d)=\Psi_0(F_1^-(u_1),\dots,F_d^-(u_d)),
\end{align*}
and $F_j^-(p)=\inf\{x\in\IR : F_j(x)\ge p\}$, $j\in\{1,\dots,d\}$, are the marginal
quantile functions. If $C$ and the margins $F_j$, $j\in\{1,\dots,d\}$, are known, we
can use Monte Carlo simulation to estimate $\IE[\Psi(\bm{U})]$. For a (pseudo-)random
sample $\{\bm{U}_i:i=1,\dots,n\}$ from $C$, the (classical) Monte Carlo
estimator of $\IE[\Psi(\bm{U})]$ is given by
\begin{align*}
    \frac{1}{n} \sum_{i=1}^n \Psi(\bm{U}_i) &\approx    \IE[\Psi(\bm{U})].
\end{align*}
The main challenge of a Monte Carlo simulation is thus the sampling of the
copula. This challenge also holds for quasi-Monte Carlo (QMC) methods, and is
actually amplified by the fact that these methods are more sensitive to certain
properties of the function $\Psi$. Thus the choice of the construction method of
a stochastic representation for $C$ can have complex effects on the performance
of QMC methods, a feature that does not show up when using Monte Carlo
methods. The present work includes a careful analysis of these effects, as they
must be thoroughly understood in order to successfully replace pseudorandom
numbers by quasi-random numbers into different copula sampling algorithms.

Let us briefly summarize the idea behind QMC methods and how they can be used
for copula models; more precise definitions on some of the concepts used here
will be given later. The idea is to start with a so-called \emph{low-discrepancy
  point set} $P_n =\{\bm{v}_1,\ldots,\bm{v}_n\} \subseteq [0,1)^{k}$, with
$k \ge d$, which is designed so that its empirical distribution over $[0,1)^k$
is closer (in a sense to be defined later) to the uniform distribution
$\U[0,1)^k$ than a set of independent and identically (i.i.d.) random points
is. Assuming that for $\bm{U}' \sim \U[0,1]^k$ we have a transformation $\phi_C$ such
that $\phi_C(\bm{U}')\sim C$, we can then construct the approximation
\begin{align}
 \frac{1}{n} \sum_{i=1}^n \Psi(\phi_C(\bm{v}_i))  \approx   \IE[\Psi(\bm{U})] .\label{eq:CopQmc1}
\end{align}
Figure \ref{fig:qrng:C} shows the points $\phi_C(\bm{v}_i)$ obtained using either pseudo-random or quasi-random numbers, for a transformation $\phi_C$ designed for the {\em Clayton copula}.
\begin{figure}[htbp]
  \centering
  \includegraphics[width=0.425\textwidth]{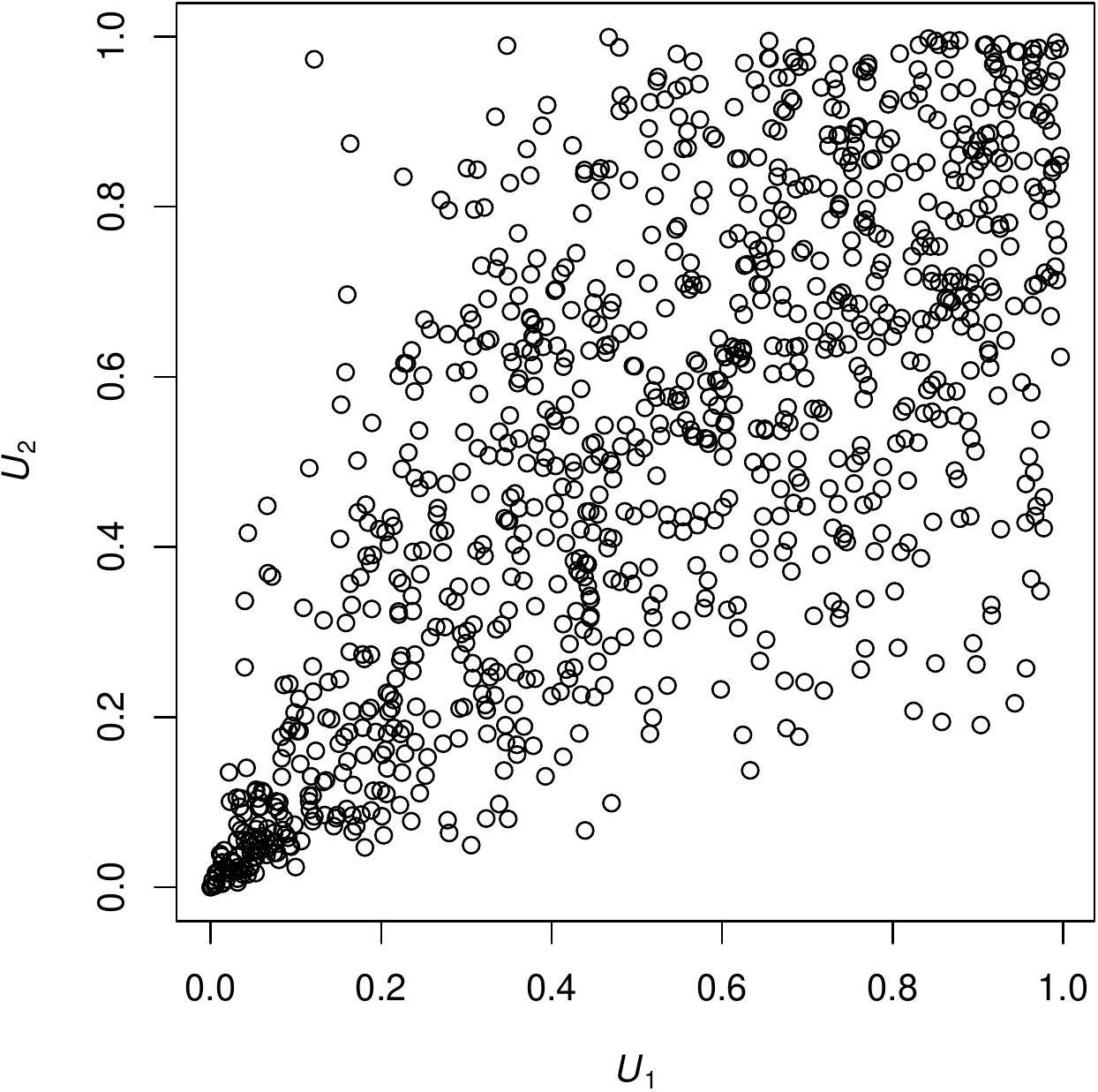}\\[2mm]
  \includegraphics[width=0.425\textwidth]{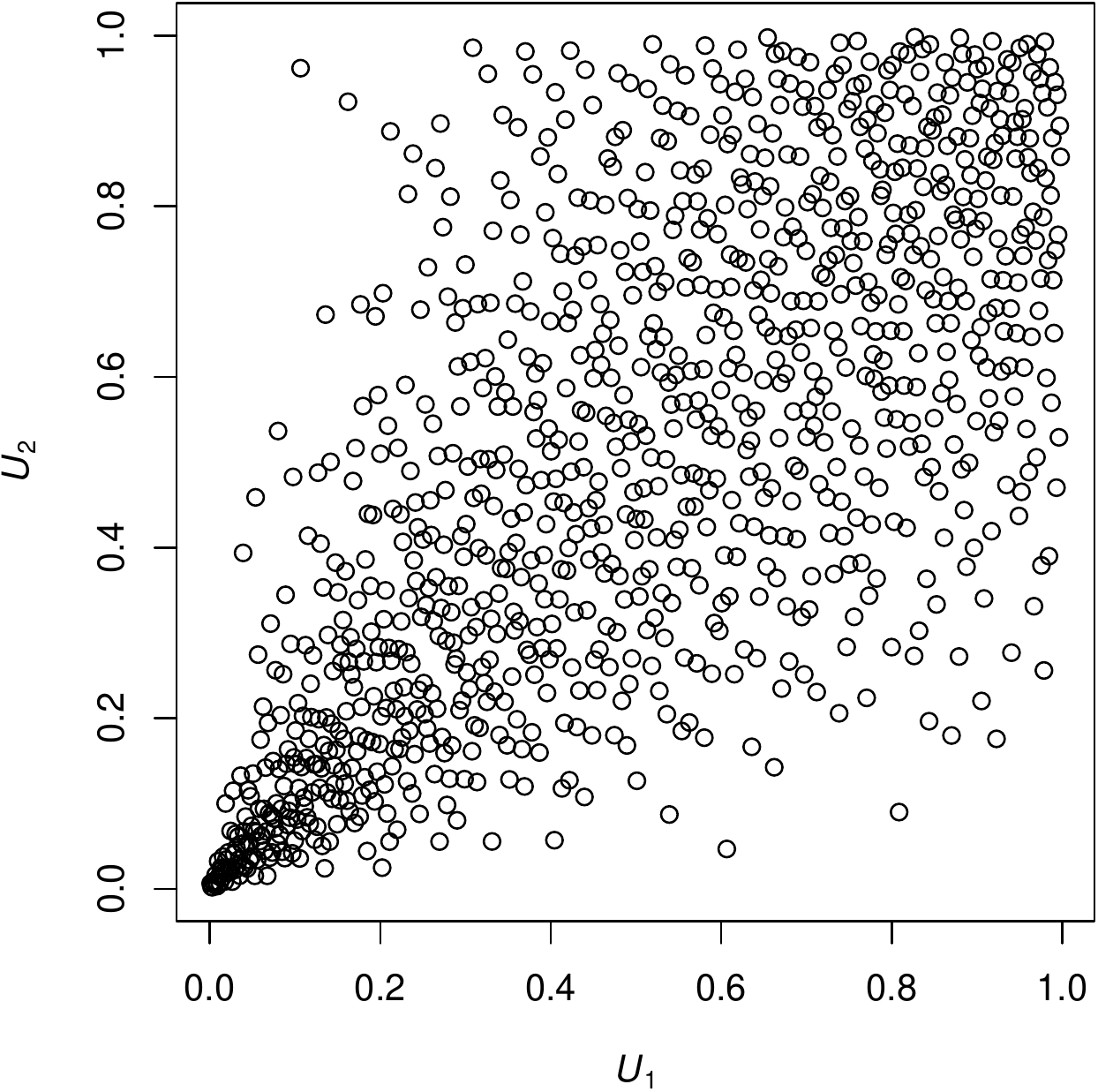}%
  \caption{1000 realizations of a bivariate Clayton copula with $\theta=2$ (Kendall's
    tau equals 0.5), generated by a
    pseudo-random number generator (top) and by a quasi-random number generator (bottom).}
  \label{fig:qrng:C}
\end{figure}

QMC methods are typically used to approximate integrals
of functions over the unit cube via
\begin{align}
\label{eq:I:Q:Intro}
Q_n=\frac{1}{n} \sum_{i=1}^n f(\bm{v}_i)\approx\int_{[0,1)^k} f(\bm{v})d\bm{v}=I(f).
\end{align}
A widely used upper bound for the integration error $|I(f)-Q_n|$ is the Koksma--Hlawka inequality
(see, e.g., \cite{Niederreiter1992}), which
is of the form $V(f)D^*(P_n)$, where $V(f)$ measures the variation of $f$ in the sense of Hardy and Krause, while $D^*(P_n)$ measures the discrepancy of $P_n$, i.e., how far $P_n$ is from $\U[0,1)^k$.

To analyze the properties of the QMC approximation (\ref{eq:CopQmc1}) for
$\IE[\Psi(\bm{U})]$, there are two possible approaches. The first one is to
define $f = \Psi \circ \phi_C$ and work within the traditional framework given
by \eqref{eq:I:Q:Intro}, the Koksma--Hlawka inequality with this composed
function and the low-discrepancy point set $P_n$. The second one is to think of
(\ref{eq:CopQmc1}) as approximating
\begin{align*}
\IE[\Psi(\bm{U})]= \int_{[0,1)^d} \Psi(\bm{u})dC(\bm{u})
\end{align*}
and work with generalizations of the Koskma--Hlawka inequality that apply to
measures other than the Lebesgue measure; see
\cite{hlawkamuck1972} or \cite{dick2014}. In the latter case, we work with the function
$\Psi$ and view the transformation $\phi_C$ as one that is applied to the
low-discrepancy point set $P_n$ rather than to $\Psi$. That is, here we work
with the transformed point set
$\tilde{P}_n = \{\phi_C(\bm{v}_1),\ldots,\phi_C(\bm{v}_n)\}$ and analyze its
quality via measures of discrepancy that quantify its distance to $C$ rather
than comparing $P_n$ to $\U[0,1)^k$.

%%%Feb. 12 2016
%\blue{
QMC methods have been used in a variety of applications, but so far most of the problems considered have been such that the stochastic models used can be formulated using independent random variables (e.g., a vector of dependent normal variates can be written as a linear transformation of independent normal variates).  In such cases, the transformation of the uniform vector $\bm{v}$ into observations from the desired stochastic model can be easily obtained by transforming each component $v_j$ of $\bm{v}$ using the inverse transform method, which is deemed to work well with QMC in part because of its monotonicity, and also because it corresponds to an overall one-to-one transformation from $[0,1)^d$ to $\mathbb{R}^d$. %}

%  \blue{
In the more general copula setting considered in this paper,
at first sight %, and whether we use the first or the second approach,
the so-called conditional distribution method (CDM) (which is the inverse of the
copula-based version of the Rosenblatt transform) appears to be a good choice to
use with quasi-random numbers, as it is the direct multivariate extension of the inverse transform
mentioned in the previous paragraph. Namely, it is a one-to-one transformation
%%% Feb. 12 2016
that maps $[0,1)^d$ to $[0,1)^d$  %(i.e., $k$ takes the minimum value $d$ in this case)
and it  is monotone in each variable.
%%%
A transformation with $k=d$ is certainly desirable (and preferable to a many-to-one transformation with $k>d$) when used in conjunction with QMC methods, since these methods do better on problems of lower dimension. Also, intuitively the monotonicity
should be helpful to preserve the smoothness of $\Psi$ (for the
first approach) and the low-discrepancy of $P_n$ (for the second approach). %, e.g., in order for convex sets to be preserved under $\phi_C$).
%%%
An additional advantage of
the CDM is that it is applicable to any copula $C$ (and the only known algorithm
such general). However, the involved (inverses of the) conditional copulas are
often challenging to evaluate which has led to other sampling algorithms being
more frequently used. An example is the Marshall--Olkin algorithm for sampling
Archimedean copulas, which we also address in this work. %}

The paper is organized as follows. Section~\ref{sec:qrng} provides a short
introduction to quasi-random numbers. Section~\ref{sec:qrng:cop} shows how
quasi-random samples from various copulas (and thus multivariate models with
these dependence structures) can be obtained using different sampling
algorithms. Detailed examples are given. In Section~\ref{sec:discrep} we discuss
the theoretical background supporting each of the two approaches mentioned
earlier to analyze the use of low-discrepancy sequences for copula
sampling. Numerical results are provided in Section \ref{sec:num}. Finally,
Section~\ref{sec:conclusion} includes concluding remarks and a discussion of
future work. Note that most results and figures presented in this paper (as well
as additional experiments conducted) can be found in the \R\ packages
\texttt{copula} (see the vignette \texttt{qrng}) and \texttt{qrng} (see
\texttt{demo(basket\_options)} and \texttt{demo(test\_functions)}).

\section{Quasi-random numbers}\label{sec:qrng}
Here we assume that a random sample $\{\bm{U}_i:i=1,\dots,n\}$ from a copula $C$
can be generated by transforming a random sample $\{\bm{U}_i':i=1,\dots,n\}$
from $\U[0,1]^k$ with $k\ge d$; several algorithms
for copula models fall under this setup. Due to the independence of the vectors
$\bm{U}_i'$, realizations of the sample $\{\bm{U}_i':i=1,\dots,n\}$ (obtained by
so-called pseudo-random number generators (PRNGs)) will inevitably show regions
of $[0,1]^k$ which are lacking points and other areas of $[0,1]^k$ which contain more
samples than expected by the uniform distribution. %; see the left-hand side of
%Figure~\ref{fig:qrng} around $(0.55,0.6)$ and $(0.6,0.5)$.
To reduce this problem of an inhomogeneous concentration of samples, quasi-random
number generators (QRNGs) do not aim at mimicking i.i.d.\ samples but instead at
producing a homogeneous coverage of $[0,1]^k$. %; see the right-hand side of
% Figure~\ref{fig:qrng}.
% \begin{figure}[htbp]
%   \centering
%   \includegraphics[width=0.48\textwidth]{./Figures/fig_prng}\\[2mm]
%   \includegraphics[width=0.48\textwidth]{./Figures/fig_qrng}%
%   \caption{1000 realizations of independent uniform random variates in $[0,1]^2$
%     generated by a PRNG (left) and a QRNG (right). Samples from the PRNG are
%     lacking, e.g., around the point $(0.55,0.6)$, whereas samples are
%     overly concentrated around $(0.6,0.5)$. In contrast, samples of the QRNG are
%     more homogeneously distributed.}
%   \label{fig:qrng}
% \end{figure}

%C -- start
The homogeneity of a sequence of points over the unit hypercube can be measured by its
discrepancy, which relates to the error incurred by representing the (Lebesgue-)measure of
subsets of the unit hypercube by the fraction of points in these
subsets. Quasi-random sequences aim at achieving smaller discrepancy than
pseudo-random number sequences and are thus also called
\emph{low-discrepancy sequences}. %Such sequences typically achieve a discrepancy
%of $O(n^{-1} \log^d n)$. %; see \cite{Niederreiter1992} and \cite{DiPi10}. More details on the notion of discrepancy are
%given in Section~\ref{sub:discrepancy} below. \blue{A review of the low-discrepancy sequences
%used in this paper is given in Section~\ref{sub:low-discrepancy_seq}}; see \cite{Niederreiter1992},
%\cite{morokoffqmc},  \cite[Chapter~5]{pG04}, and \cite{DiPi10} for more details.
%%
%\blue{
The rest of this section reviews concepts related to QRNG that are used in this paper. The reader is referred to \cite{Niederreiter1992},
\cite{morokoffqmc},  \cite[Chapter~5]{pG04}, and \cite{DiPi10} for more details. %}

\subsection{Discrepancy}\label{sub:discrepancy}
The notion of discrepancy  applies to sequences of points $X=\{\bm{v}_1,\bm{v}_2,\ldots\}$ in the
unit hypercube $[0,1)^k$.  Denote by
$P_n =\{\bm{v}_1,\ldots,\bm{v}_n\}\subseteq [0,1)^k$ the first $n$ points of the
sequence. Let $\mathcal{J^*}$ be the set of intervals of $[0,1)^k$ of the form
$[\bm{0},\bm{z})=\prod_{j=1}^k [0,z_j)$, where $0 <z_j \le 1, j=1,\dots,k$. Then the {\em
  discrepancy function} of $P_n$ on an interval $[\bm{0},\bm{z})$ is the difference
\begin{align*}
E([\bm{0},\bm{z});P_n)=\frac{A([\bm{0},\bm{z});P_n)}{n}-\lambda([\bm{0},\bm{z})),
\end{align*}
where $A([\bm{0},\bm{z});P_n)=\#\{i;\ 1\leq i \leq n,\bm{v}_i \in [\bm{0},\bm{z})\}$ is the number of points from $P_n$ that fall in $[\bm{0},\bm{z})$ and $\lambda([\bm{0},\bm{z})) = \prod_{j=1}^k z_j$ is the Lebesgue measure of $[\bm{0},\bm{z})$.

The \emph{star discrepancy} $D^*$ of $P_n$ is defined by
\begin{align*}
D^*(P_n)=\sup_{[\bm{0},\bm{z}) \in \mathcal{J^*}} \vert E([\bm{0},\bm{z});P_n) \vert.
\end{align*}
An infinite sequence $X$ satisfying $D^*(P_n)\in O(n^{-1}\log^k n)$ is
said to be a {\em low-discrepancy sequence}. %, where $P_n$ consists of the first $n$ points of $X$. %\footnote{MH: The notation introduced here does not seem to be  used elsewhere. Can we remove it?}
%If instead of restricting the intervals $[\bm{0},\bm{z})$ to be anchored at the origin, we allow them to take the form $\bm{z} = \prod_{j=1}^d [y_j,z_j)$ with $0 \le y_j < z_j \le 1$, then the corresponding discrepancy is denoted $D(P_n)$.
%\christiane{I added this because we probably will need it when discussing Hlawka and Muck.}
%

% Given $\mathcal{J}$ the set of all intervals of $[0,1)^d,$ of type $[\bm{a},\bm{b}]^d$, with $0\leq a_j\leq b_j< 1$, $j\in\{1,\dots,d\}$, the \textit{discrepancy} $D$ of $P_n$ is defined as
% \begin{align}\label{eq:Lebesgue_discrepancy}
% 	D(P_n)=\sup_{B \in \mathcal{J}} \vert E(B;P_n) \vert.
% \end{align}

For a function $\Psi:[0,1)^k \rightarrow \mathbb{R}$, we have the well-known Koksma--Hlawka error bound  given by
\begin{align}\label{eq:Koksma-Hlawka}
  \biggl|\frac{1}{n} \sum_{i=1}^n \Psi(\bm{v}_i) -\IE[\Psi(\bm{U}')]\biggr|\le V(\Psi)D^*(P_n),
\end{align}
where $\bm{U}'\sim \U[0,1]^k$ and $V(\Psi)$ denotes the variation of the function
$\Psi$ in the sense of Hardy and Krause. %; see \cite[Section 2.2]{Niederreiter1992} or \cite[Section 5.1.3]{pG04}.
See \cite{Owen04} for a
detailed account of the notion of variation and its applicability in practice. We also refer the interested reader to \cite{HartKeinTichy2004} and \cite{Sobol73} for results handling unbounded functions (and thus of unbounded variation).
\iffalse
The Koksma--Hlawka error bound is an instance of a more general error bound of
the form
\begin{align*}
 \biggl|\frac{1}{n} \sum_{i=1}^n \Psi(\bm{v}_i) -\IE[\Psi(\bm{U})]\biggr|\le V^q(\Psi)D^p(P_n),
\end{align*}
provided in \cite{Hickernell98}, where
\begin{align*}
  V^q(\Psi)&= \left[\sum_J \int\limits_{[0,1)^{|J|}}\!\!\left(
              \left|\frac{\partial^{|J|} \Psi}{\partial \bm{v}_J}
              \right|_{v_j=1\,\forall j\notin
  J}\right)^qd\bm{v}_J\right]^{1/q}\text{and}\\
  D^p(P_n) &= \left[\sum_J \int\limits_{[0,1)^{|J|}}\!\!\! E^p(\bm{z}_J;P_n(J)) d\bm{z}_J\right]^{1/p};
\end{align*}
here, $\emptyset \neq J \subseteq \{1,\ldots,d\}$, $P_n(J)$ is the projection of
$P_n$ over the subspace indexed by the variables in $J$,
$d\bm{z}_J = \prod_{j \in J} dz_j$, $p,q\ge 1$ are such that
$\frac{1}{p}+\frac{1}{q} = 1$, and it is assumed that $\Psi$ has mixed partial
derivatives that are in $L_q$, i.e., the space of all functions whose $q$th
power of their absolute value is Lebesgue integrable. In addition to the choice
$p=\infty,q=1$ that leads to the Koksma--Hlawka inequality, another common
choice is to take $p=q=2$, which has the advantage of making the corresponding
$L_2$-discrepancy $D^2(P_n)$ computationally tractable, unlike the
star-discrepancy $D^*(P_n)$. The case $p=q=2$ was originally derived in
\cite{Zaremba68}; a formula for computing $D^2(P_n)$ is provided in
\cite[Eq. (5.1c)]{Hickernell98}.
\fi

\subsection{Low-discrepancy sequences}\label{sub:low-discrepancy_seq}
There are two main approaches for constructing low-discre\-pancy sequences:
integration lattices and digital sequences. %The family of lattices is more often used to
%construct point sets of a fixed size $n$ (as opposed to sequences), but can also
%be used to construct point sets $P_n$ extensible in $n$ (see \cite{HiNi03}, for
%example).
%A brief overview of lattices is provided in the appendix. Although we do not use this construction in our numerical experiments of Section \ref{sec:num}, in Section \ref{sec:discrep} we provide some results on the variance of a lattice-based estimators in the context of copula models.
Only the latter are used in this paper, so our discussion will focus on those.

Digital sequences contain the well-known constructions of \cite{Sobol67},
\cite{Faure82}, and \cite{Niederreiter87}, and  are also closely related to the
sequence proposed in \cite{Halton60}. The basic building block for this
construction is the  {\em van der Corput sequence in base $b \ge 2$}, defined as
\begin{align}
  S_b(i) = \sum_{r=0}^\infty a_r(i) b^{-r-1},\quad i\in\IN,\label{eq:vdc}
\end{align}
where $a_r(i)$ is the $r$th digit of the $b$-adic expansion of $i-1=\sum_{r=0}^\infty a_r(i)\ b^r$.
To construct a sequence of points in $[0,1)^k$  from this one-dimensional sequence, one approach is the one proposed by \cite{Halton60}, which consists of choosing $k$ pairwise relatively prime integers $b_1,\ldots,b_k$ and defining the $i$th point of the sequence as
\begin{align*}
\bm{v}_i = (S_{b_1}(i),\ldots,S_{b_k}(i)), \quad i\in\IN.
\end{align*}

Another possibility is to fix the base $b$, and choose $k$ linear
transformations that are then applied to the digits $a_r(i)$ from the expansion
of $i-1$ before being used in (\ref{eq:vdc}) to define a real number between 0
and 1. More precisely, let $M_1,\ldots,M_k$ be (unbounded)
``$\infty \times \infty$'' matrices with entries in $\mathbb{Z}_b$ and let
\begin{align}
  S_b^{M_j}(i)=\sum_{r=0}^\infty \sum_{l=0}^{\infty}m_{r+1,l+1}a_l(i) b^{-r-1},\label{eq:genvdc}
\end{align}
where $m_{r,l}$ is the element in the $r$th row and $l$th column of $M_j$. Here
we assume for simplicity that $b$ is prime and all operations in
(\ref{eq:genvdc}) are performed in the finite field $\mathbb{F}_b$. One can then
define a sequence of points in $[0,1)^k$ by taking
\begin{align}
 \bm{v}_i =  (S_b^{M_1}(i),\ldots,S_b^{M_k}(i))\label{eq:genVDC}
\end{align}
as its $i$th point. Sobol' was the first to propose such a construction, working
in base $2$ and defining the matrices $M_1,\ldots,M_k$ so that he
was able to prove that the obtained sequence has
$D^*(P_n) \in O(n^{-1} \log^k n)$; see \cite{Sobol67}.
\iffalse
Next, \cite{Faure82} proposed a
construction with $b \ge d$ prime and matrices $M_j$ given by the successive powers
of the Pascal matrix over $\mathbb{Z}_b$, which also achieves the rate $O(n^{-1} \log^d n)$
for its star-discrepancy, in addition to having a certain optimality property related to
the equidistribution of its points in $[0,1)^d$. \cite{Niederreiter87} extended
these early constructions and, among other things, provided a more general framework
in which the base $b$ can be any positive integer. This led to the concept of
digital sequences, which is reviewed thoroughly in \cite{Niederreiter1992} and \cite{DiPi10}.
\fi

We also point out that Halton sequences can be generalized using the same idea
as in \eqref{eq:genVDC}. That is, one can choose matrices $M_1,\ldots,M_k$ with
the elements of $M_j$ in $\mathbb{Z}_{b_j}$, and ``scramble'' the digits of the
expansion of $i-1$ before reverting them via (\ref{eq:vdc}) to produce a number
between 0 and 1. A very simple way to achieve this is via diagonal matrices
$M_j$, each containing a well-chosen element (or factor) of $\mathbb{Z}_{b_j}$. In our
numerical experiments, we use such an approach, with the factors provided in
\cite{qFAU07a}; see the \R\ package \texttt{qrng} for an
implementation.

\subsection{Randomized quasi-Monte Carlo}
In contrast to the error rate $O(n^{-1/2})$ for
Monte Carlo methods based on PRNGs, approximations based on
QRNGs have the advantage of having an error in $O(n^{-1}\log^kn)$ when the
function of interest $\Psi$ is of bounded variation. However, %in order for a given method to be truly useful
in practice it is also
important to be able to estimate the corresponding error. While bounds such as
the Koksma--Hlawka inequality are useful to understand the behaviour of
approximations based on quasi-random sequences, they do not provide a practical
way to estimate the error. To circumvent this problem, an approach that is often
used is to \emph{randomize} a low-discrepancy point set in such a way that its
high uniformity (or low discrepancy) is preserved, but at the same time
unbiased estimators can be constructed (and sampled) from it. %\blue{
Another advantage of this approach is that variance expressions can be derived and compared with Monte Carlo sampling for wider classes of functions, i.e., not necessarily of bounded variation (see \cite{Owen97a}, \cite{Lemieux2009} and the references therein). %}
This approach gives rise to \emph{randomized quasi-Monte Carlo (RQMC)} methods.

To apply this approach, we need a randomization function
$r:[0,1)^s \times [0,1)^k \rightarrow [0,1)^k$ with $s \ge k$ such that for any
fixed $\bm{v} \in [0,1)^k$, we have that if $\bm{U}' \sim \U[0,1]^s$, then
$r(\bm{U}',\bm{v}) \sim \U[0,1]^k$.
%In other words, if we think of $\bm{v}$ as a
%point from a (deterministic) low-discrepancy point set $P_n$ and $\bm{V}$ as the
%source of randomness that will randomize $P_n$, then we create a new randomized
%point $\tilde{\bm{v}}$ which is $\U[0,1]^d$ by applying $r$ to $\bm{v}$ with the
%random vector $\bm{V}$. Note that this means that
%Hence for any transformation
%$\phi_C$ applied to $\tilde{\bm{v}}$, the distribution of
%$\phi_C(\tilde{\bm{v}})$ is the same as that of $\phi_C(\bm{W})$ with $\bm{W}\sim\U[0,1]^d$ %directly.
Hence  the individual RQMC samples have
the same properties as those from a random sample; the difference lies in the fact
that the RQMC samples are dependent.

An early randomization scheme originally proposed by \cite{CranleyPatterson76} is to take
\begin{align}
  r(\bm{U}',\bm{v}) = (\bm{v} + \bm{U}') \bmod{1}.\label{eq:RandomShift}
\end{align}
A randomized point set is then obtained by generating a uniform vector $\bm{U}'$ and letting
$\tilde{P}_n(\bm{U}') = \{\tilde{\bm{v}}_1,\ldots,\tilde{\bm{v}}_n\}$, where
%\begin{align*}
$\tilde{\bm{v}}_i = r(\bm{U}',\bm{v}_i),\quad i\in\{1,\ldots,n\}.$
%\end{align*}
%\blue{
Hence the same shift $\bm{U}'$ is applied to all points in $P_n$.  %}
If we let $\bm{U}'_1,\ldots,\bm{U}'_B$ be independent $\U[0,1]^s$ vectors, we can construct $B$ i.i.d.\ unbiased estimators
\begin{align*}
  \hat{\mu}^l_n = \frac{1}{n} \sum_{\tilde{\bm{v}}_i \in \tilde{P}_n(\bm{U}'_l)} \Psi(\phi_C(\tilde{\bm{v}}_i)), \quad l\in\{1,\ldots,B\}
\end{align*}
for $\IE[\Psi(\bm{U})]$, whose variances can be estimated via the sample variance of $\hat{\mu}^1_n,\ldots,\hat{\mu}^B_n$.

In addition to the simple random shift described in (\ref{eq:RandomShift}),
several other randomization schemes have been proposed and studied. A popular
randomization method for digital nets is to ``scramble'' them, an idea
originally proposed by \cite{Owen95} and subsequently studied by \cite{Owen97a},
\cite{Owen97b},  \cite{Owen03}, \cite{Matousek98} and \cite{HongHickernell2003}, among
others.
%An excellent review of various scrambling schemes and their associated
%variances can be found in \cite{Owen03}.

A simpler randomization for digital nets is to use the digital counterpart of (\ref{eq:RandomShift}), where instead of adding two real numbers modulo 1, we add (in $\mathbb{Z}_b$)  the digits of their base $b$ expansion. That is, for $u=\sum_{r=0}^{\infty} u_r b^{-r-1}$ and $v =\sum_{r=0}^{\infty} v_r b^{-r-1}$, we let
\[
u \oplus_b v = \sum_{r=0}^{\infty} ((u_r+v_r) \bmod{b}) b^{-r-1}
\]
and define
$r(\bm{u},\bm{v}) = \bm{u} \oplus_b \bm{v},$
where the $\oplus_b$ operation is applied component-wise to the $k$ coordinates of $\bm{u}$ and $\bm{v}$. The same idea can be applied to randomize Halton sequences (as shown, e.g., in \cite{Lemieux2009}), but where a different base $b$ is used in each of the $k$ coordinates. Digital shifts for the Sobol' and generalized Halton sequences are available in our \R\ package \texttt{qrng}.

%%% INSERT MARIUS PORTION HERE

\section{Quasi-random copula samples}\label{sec:qrng:cop}
Sampling procedures for a $d$-dimensional copula $C$ can be viewed as
transformations $\phi_C:[0,1]^k\to[0,1]^d$ for some $k\ge d$, such that, for
$\bm{U}'\sim \U[0,1]^k$, $\bm{U}:=\phi_C(\bm{U}')\sim C$; that is, $\phi_C$
transforms independent $\U[0,1]$ random variables to $d$ dependent random
variables with distribution function $C$.

The case $k=d$ is mostly known and applied as \emph{conditional distribution
  method (CDM)} and involves the inversion method for sampling univariate
conditional copulas (although, for example, for Archimedean copulas another
transformations $\phi_C$ with $k=d$ is known; see \cite{WVSsamplingcopulas}).
This approach thus naturally uses $d$ independent $\U[0,1]$ random variables as
input. The case $k\ge d$ (often: $k>d$) is typically known as \emph{stochastic
  representation} and is usually based on  sampling
$k$ univariate random variables from
elementary probability distributions, as we will see in Section \ref{subsec:StochRep}.
%that can be %; although not necessarily done, the latter
%could obviously also be
%sampled via the inversion method.
%and thus fit in our framework.

In what follows we consider the above two approaches and show how they can be
adapted to quasi-random number generation.

\subsection{Conditional distribution method and other one-to-one transformations ($k=d$)}\label{sec:qrng:cop_cd}
\subsubsection{Theoretical background}
The only known general sampling approach which works for any copula is the
CDM. For $j\in\{2,\dots,d\}$, let
\begin{align*}
  C(u_j\,|\,u_1,\dots,u_{j-1})=\IP(U_j\le u_j\,|\,U_1=u_1,\dots,U_{j-1}=u_{j-1})
\end{align*}
denote the \emph{conditional copula of $U_j$ at $u_j$ given
  $U_1=u_1,\dots$,\\$U_{j-1}=u_{j-1}$}. If $C^-(u_j\,|\,u_1,$ $\dots,u_{j-1})$ denotes
the corresponding quantile function, the CDM is given as follows; see
\cite{embrechtslindskogmcneil2003} or \cite[p.~45]{mariusphdthesis}.
\begin{theorem}[Conditional distribution method]
  Let $C$ be a $d$-dimensional copula, $\bm{U}'\sim\U[0,1]^d$, and $\phi_C^{\text{CDM}}$ be
  given by
  \begin{align*}
    U_1&=U_1',\\
    U_2&=C^{-}(U_2'\,\vert\,U_1),\\
    &\phantom{:}\vdots\\
    U_d&=C^{-}(U_d'\,\vert\,U_1,\dots,U_{d-1}).
  \end{align*}
  Then $\bm{U}=(U_1,\dots,U_d)=\phi_C^{\text{CDM}}(\bm{U}')\sim C$.
\end{theorem}
To find the conditional copulas $C(u_j\,\vert\,u_1,\dots,u_{j-1})$,
for $j\in\{2,\dots,d\}$, for a specific copula $C$, the following result (which
holds under mild assumptions) is often applied. A rigorous proof can be found in
\cite[p.~20]{schmitzthesis}, an implementation is provided by the function
\texttt{rtrafo()} in the \R\ package \texttt{copula}. %\blue{
  The corollary that follows is an immediate consequence of Sklar's theorem, for example. %}
\begin{theorem}[Computing conditional copulas]\label{thm:schmitz}
  Let $C$ be a $d$-dimensional copula, which, for $d\ge 3$, admits continuous partial
  derivatives with respect to the first $d-1$ arguments. For
  $j\in\{2,\dots,d\}$ and $u_l\in[0,1]$, $l\in\{1,\dots,j\}$,
  \begin{align}
    C(u_j\,|\,u_1,\dots,u_{j-1})&=\frac{\D_{j-1\dots 1}
      C(u_1,\dots,u_j)}{\D_{j-1\dots 1} C(u_1,\dots,u_{j-1})} \notag\\
      &=\frac{\D_{j-1\dots
        1} C(u_1,\dots,u_j)}{c(u_1,\dots,u_{j-1})},\label{eq:CDM}
  \end{align}
  where $\D_{j-1\dots 1}$ denotes the derivative with respect to the first $j-1$
  arguments, $C(u_1,\dots,u_j)$ denotes the marginal copula corresponding to the first $j$
  components and $c(u_1,\dots,u_{j-1})$ denotes the density of $C(u_1,\dots,u_{j-1})$. If $C$
  admits a density, then \eqref{eq:CDM} equals
  \begin{align}
    C(u_j\,|\,u_1,\dots,u_{j-1})=\frac{\int_0^{u_j}c(u_1,\dots,u_{j-1},z_j)\,dz_j}{c(u_1,\dots,u_{j-1})}.\label{eq:CDM:dens}
  \end{align}
\end{theorem}

\begin{corollary}[Conditional copulas for general multivariate distributions]\label{cor:cond_cop}
  Let $H$ be a $d$-dimensional absolutely continuous distribution function
  with margins $F_1,\dots,F_d$ and copula $C$. For $j\in\{2,\dots,d\}$ and
  $u_l\in[0,1]$, $l\in\{1,\dots,j\}$,
  \begin{align}
    C(u_j\,|\,u_1,\dots,u_{j-1})=H(F_j^-(u_j)\,|\,F_1^-(u_1),\dots,F_{j-1}^-(u_{j-1})).\label{eq:CDM:H}
    % &=\int_{-\infty}^{F_j^-(u_j)}h(z_j\,|\,F_1^-(u_1),\dots,F_{j-1}^-(u_{j-1}))\,dz_j.
  \end{align}
\end{corollary}

\subsubsection{Examples}
We now present several copula families and show how the corresponding
conditional copulas and their inverses can be computed. To the best of our
knowledge, several of these results have not appeared in the literature
before.

\subsubsection*{Elliptical copulas}
An \emph{elliptical copula} describes the dependence structure of an elliptical distribution; for the latter, see \cite{cambanis1981},
\cite{correlationanddependence}, \cite{embrechtslindskogmcneil2003}, or
\cite[Sections~3.3, 5]{qrm}. The most prominent two families in the class of
elliptical copulas are the Gauss and the $t$ copulas.

\paragraph{Gauss copulas.}
Gauss copulas are given by
\begin{align*}
  C_P^{\text{Ga}}(\bm{u})=\Phi_P(\Phi^{-1}(u_1),\dots,\Phi^{-1}(u_d)),
\end{align*}
where $\Phi_P$ denotes the $d$-variate normal
distribution function with location vector $\bm{0}$ and scale matrix $P$ (a
correlation matrix) and
$\Phi^{-1}$ is the standard normal quantile function. Consider the dimension to
be $j$ and let
$\bm{X}\sim\Phi_P$ with $\bm{X}=(\bm{X}_{1:(j-1)},X_j)$. Furthermore, assume
\begin{align*}
  P=\Bigl(\begin{smallmatrix} P_{1:(j-1),1:(j-1)} & P_{1:(j-1),j} \\ P_{j,1:(j-1)} & P_{j,j}\end{smallmatrix}\Bigr)
\end{align*}
to be positive definite. It follows from
\cite[p.~45 and 78]{fangkotzng1990} that
\begin{align*}
  X_j\,|\,\bm{X}_{1:(j-1)}=\bm{x}_{1:(j-1)}\sim\mathcal{N}(\mu_{j|1:(j-1)}(\bm{x}_{1:(j-1)}),P_{j|1:(j-1)}),
\end{align*}
where
\begin{align}
  \mu_{j|1:(j-1)}(\bm{x}_{1:(j-1)})=P_{j,1:(j-1)}\bigl(P_{1:(j-1),1:(j-1)}\bigr)^{-1}\bm{x}_{1:(j-1)},\notag\\
  P_{j|1:(j-1)}=P_{j,j}-P_{j,1:(j-1)}\bigl(P_{1:(j-1),1:(j-1)}\bigr)^{-1}P_{1:(j-1),j};\label{eq:cond:P}
\end{align}
so $H(x_j\,|\,x_1,\dots,x_{j-1})$ is again normal. With
$\Phi^{-1}(\bm{u}_{1:(j-1)})=(\Phi^{-1}(u_1),\dots,\Phi^{-1}(u_{j-1}))$, it follows from
\eqref{eq:CDM:H} that
\begin{align*}
  &\phantom{{}={}}C(u_j\,|\,u_1,\dots,u_{j-1}) =H(\Phi^{-1}(u_j)\,|\,\Phi^{-1}(\bm{u}_{1:(j-1)}))\\
  &=\Phi_{\mu_{j|1:(j-1)}(\Phi^{-1}(\bm{u}_{1:(j-1)})),P_{j|1:(j-1)}}(\Phi^{-1}(u_j))\\
  &=\Phi\biggl(\frac{\Phi^{-1}(u_j) - \mu_{j|1:(j-1)}(\Phi^{-1}(\bm{u}_{1:(j-1)}))}{\sqrt{P_{j|1:(j-1)}}}\biggr)
\end{align*}
and thus that
\begin{align*}
  &\phantom{{}={}}C^-(u_j\,|\,u_1,\dots,u_{j-1})\\
  &=\Phi\bigl(\Phi_{\mu_{j|1:(j-1)}(\Phi^{-1}(\bm{u}_{1:(j-1)})),P_{j|1:(j-1)}}^{-1}(u_j)\bigr)\\
  & =\Phi\bigl(\mu_{j|1:(j-1)}(\Phi^{-1}(\bm{u}_{1:(j-1)}))+\sqrt{P_{j|1:(j-1)}}\Phi^{-1}(u_j)\bigr).
\end{align*}
An implementation of this inverse is provided via \texttt{rtrafo(, inverse=TRUE)} in the \R\ package \texttt{copula}.

\paragraph{$t$ copulas.}
%If $\bm{X}\in\IR^d$ has the stochastic representation
%$\bm{X}=(\sqrt{\nu/S})\bm{Z},$ where $\nu>2$ and89
%$S\sim\chi_\nu^2$ and $\bm{Z}\sim\mathcal{N}(0,P)$ are independent,
%then $\bm{X}$ has a $d$-variate $t$ distribution with $\nu$ degrees of freedom
%and covariance matrix
%$(\nu/\nu-2)P.$

$t$ copulas are given by
\begin{align*}
  C_{\nu,P}^{t}(\bm{u})=t_{\nu,P}(t_\nu^{-1}(u_1),\dots,t_\nu^{-1}(u_d)),
\end{align*}
where $t_{\nu,P}$ denotes the $d$-variate $t_\nu$ distribution function
with location vector $\bm{0}$ and scale matrix $P$ (a correlation matrix) and
$t_\nu^{-1}$ is the standard $t_\nu$ quantile function. The following
proposition guarantees stability of the $t$ copula upon conditioning; see
the appendix %~\ref{sec:proof}
for its proof and \texttt{rtrafo()}
for an implementation. We are not aware that this result has appeared
before. Given the importance of $t$ copulas in practice, this is rather
remarkable.
\begin{proposition}[Conditional $t$ copulas and inverses]\label{prop:cond_t_copula}
  With the notation as in the Gauss case, the conditional $t$ copula at $u_j$,
  given $u_1,\dots,u_{j-1}$, and its inverse are given by
  \begin{align*}
    C(u_j\,|\,u_1,\dots,u_{j-1})&=t_{\nu+j-1}\Bigl(s_1\Bigl(\sqrt{P^{-1}_{j,j}} t_{\nu}^{-1}(u_j)+s_2\Bigr)\Bigr),\\
    C^-(u_j\,|\,u_1,\dots,u_{j-1})&=t_{\nu}\biggl(\frac{t_{\nu+j-1}^{-1}(u_j)/s_1-s_2}{\sqrt{P^{-1}_{j,j}}}\biggr),
  \end{align*}
  for $s_1,s_2$ as given in the proof.
\end{proposition}

% \subsection{Homogeneous sampling for margins}\label{sec:homogeneousmargins}
%
%One can apply inverse marginal cdf's to quasi copula samples. However, we propose a different method: sample margins from a equidistant grid on $[0,1]$, see Figure~\ref{fig:margins_sampling}. Then reorder according to ranks of quasicopula sample, see \cite{AggTree}.
%
%\begin{figure}
%        \begin{subfigure}[h]{0.5\textwidth}
%                \centering
%				\includegraphics[width =7.3cm]{random_margins.pdf}
%				\caption{Uniform sampling for margins.}
%        \end{subfigure}
%        \quad %add desired spacing between images, e. g. ~, \quad, \quad etc.
%          %(or a blank line to force the subfigure onto a new line)
%        \begin{subfigure}[h]{0.5\textwidth}
%                \centering
%				\includegraphics[width =7.3cm]{quasi_random_margins.pdf}
%				\caption{Quasi sampling for margins.}
%        \end{subfigure}
%        \caption{Margins sampling.}\label{fig:margins_sampling}
%\end{figure}

Figure~\ref{fig:qrng:t3} displays 1000 samples from a $t$ copula with three
degrees of freedom and correlation parameter $\rho=P_{1,2}=1/\sqrt{2}$
(Kendall's tau equals 0.5), once drawn with a PRNG (top) and once with a QRNG
(bottom). We can visually confirm in this case that the low discrepancy of the
latter is preserved. How this seemingly good feature translates into better
estimators of the form \eqref{eq:CopQmc1} will be studied further through the
theoretical results of Section~\ref{sec:discrep} and the numerical experiments
of Section~\ref{sec:num}.
\begin{figure}[htbp]
  \centering
  \includegraphics[width=0.425\textwidth]{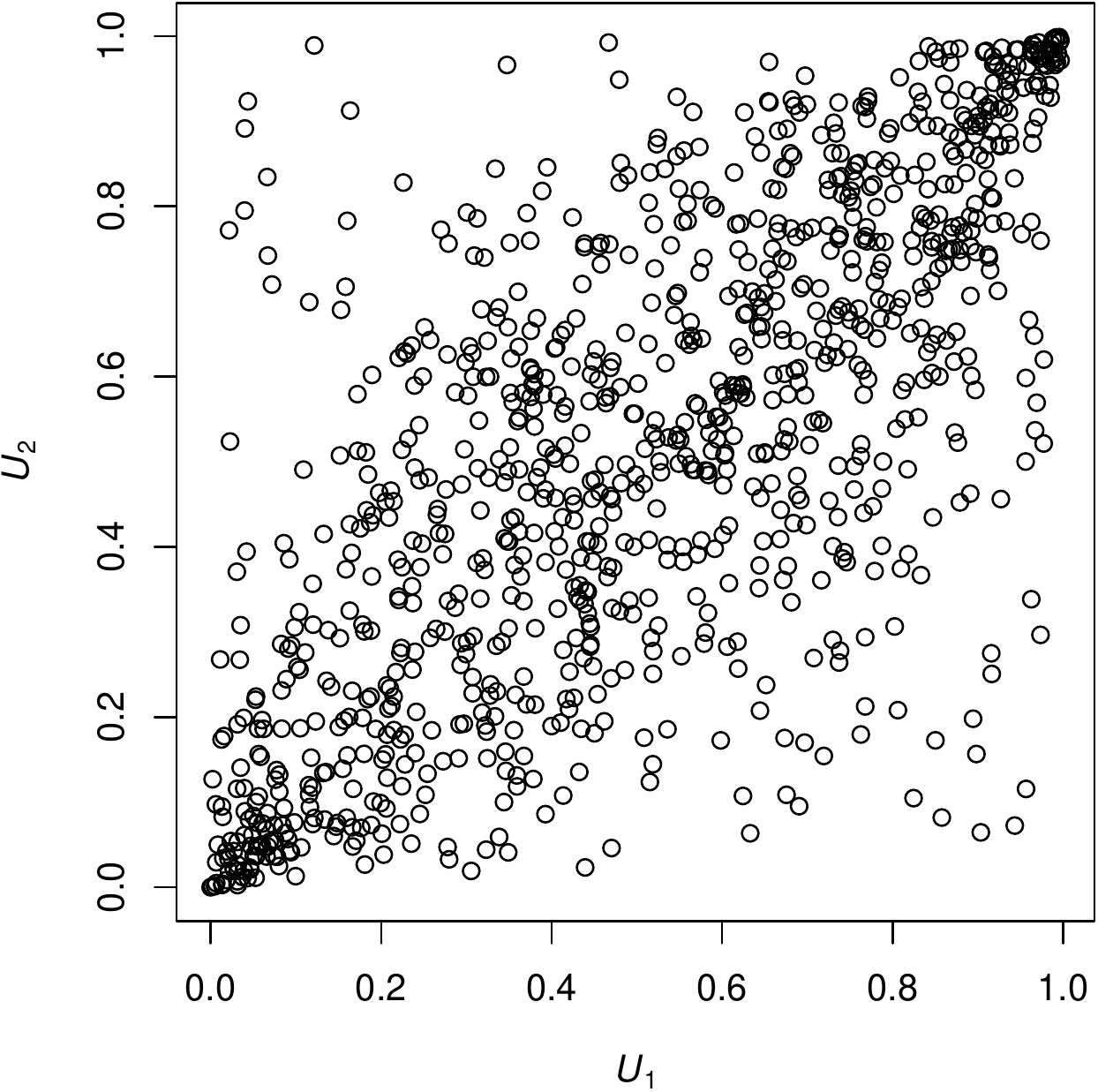}\\[2mm]
  \includegraphics[width=0.425\textwidth]{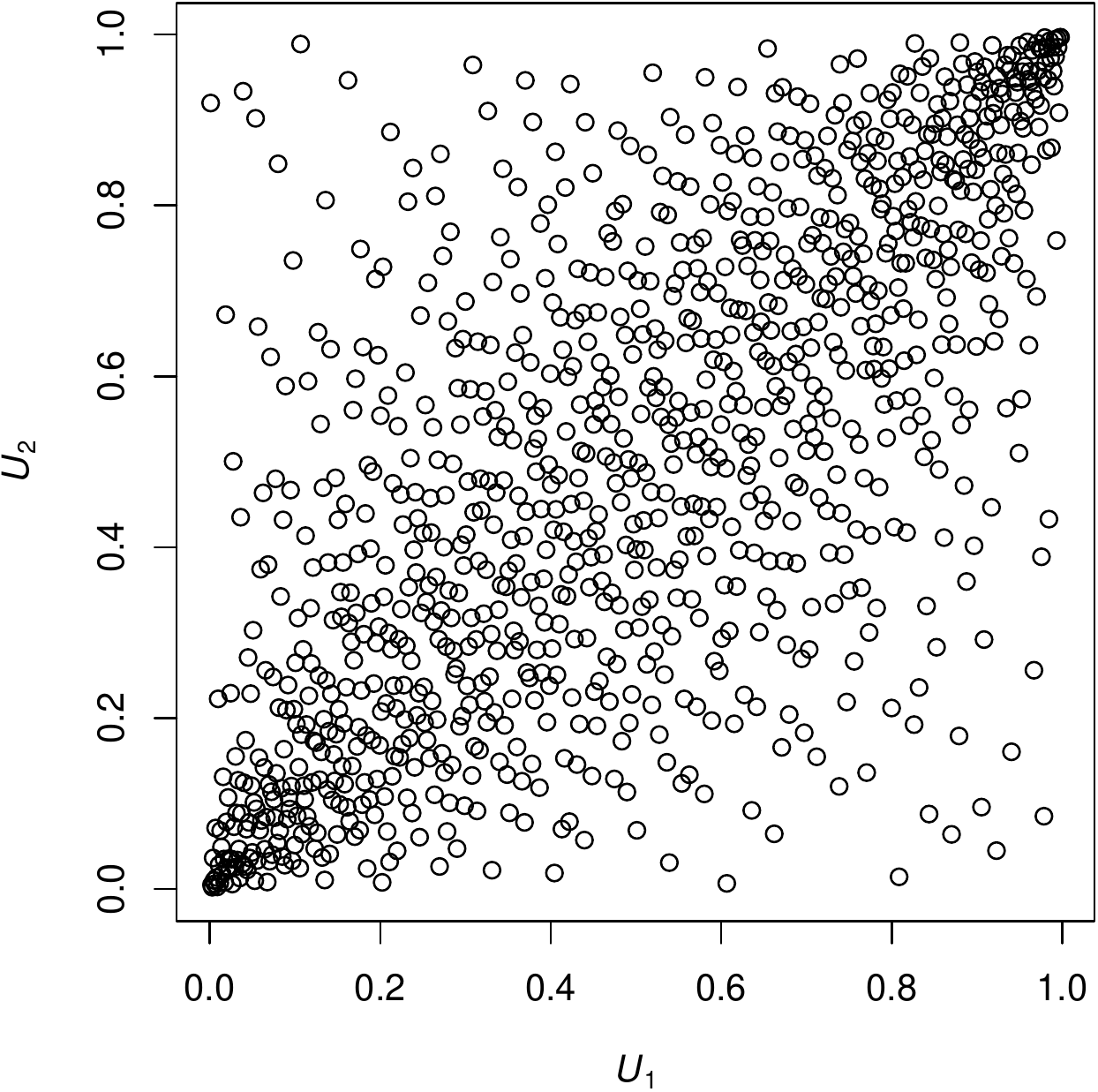}%
  \caption{1000 realizations of a $t$ copula with three degrees of freedom and
    correlation parameter $\rho=1/\sqrt{2}$ (Kendall's tau equals 0.5),
    generated by a PRNG (top) and by a QRNG (bottom).}
  \label{fig:qrng:t3}
\end{figure}

\subsubsection*{Archimedean copulas}
An \textit{(Archimedean) generator} is a continuous, decreasing function
$\psi:[0,\infty]\to[0,1]$ which satisfies $\psi(0)=1$,
$\psi(\infty)=\lim_{t\to\infty}\psi(t)=0$, and which is strictly decreasing on
$[0,\inf\{t:\psi(t)=0\}]$. A $d$-dimensional copula $C$ is called
\textit{Archimedean} if it permits the representation
  \begin{align*}
    C(\bm{u})=\psi(\psii(u_1)+\dots+\psii(u_d)),
  \end{align*}
  where $\bm{u}=(u_1,\dots,u_d)\in[0,1]^d$, and
  for some generator $\psi$ with inverse $\psii:[0,1]\to[0,\infty]$, where
  $\psii(0)=\inf\{t:\psi(t)=0\}$. For applications and the importance of
  Archimedean copulas in the realm of finance and insurance, see, e.g.,
  \cite{hofertmaechlermcneil2013}.

  \cite{mcneilneslehova2009} show that a generator defines an Archimedean copula
  if and only if $\psi$ is $d$\textit{-monotone}, meaning that $\psi$ is
  continuous on $[0,\infty]$, admits derivatives $\psi^{(l)}$ up to the order
  $l=d-2$ satisfying $(-1)^l\psi^{(l)}(t)\ge0$ for all $l\in\{0,\dots,d-2\}$,
  $t\in(0,\infty)$, and $(-1)^{d-2}\psi^{(d-2)}(t)$ is decreasing and convex on
  $(0,\infty)$.

  Assuming $\psi$ to be sufficiently often differentiable,
  conditional Archimedean copulas follow from
  Theorem~\ref{thm:schmitz} and are given by
  \begin{align}
    C(u_j\,|\,&u_1,\dots,u_{j-1})=\frac{\psi^{(j-1)}\bigl(\sum_{l=1}^{j}\psii(u_l)\bigr)}{\psi^{(j-1)}\bigl(\sum_{l=1}^{j-1}\psii(u_l)\bigr)},\label{cond:cop:Clayton}
  \end{align}
  where $u_l\in[0,1],\ l\in\{1,\dots,j\}$, and thus
  \begin{align}
    &C^-(u_j\,|\,u_1,\dots,u_{j-1}) \notag\\
    &=\psi\biggl(\!\!{\psi^{(j-1)}}^{-1}\biggl(\!u_j\psi^{(j-1)}\biggl(\,\sum_{l=1}^{j-1}\psii(u_l)\!\!\biggr)\!\biggr)
    -\sum_{l=1}^{j-1}\psii(u_l)\!\biggr).\label{CDM:Clayton}
  \end{align}
  The generator derivatives $\psi^{(j-1)}$ and their inverses
  ${\psi^{(j-1)}}^{-1}$ can be challenging to compute. The former are
  known explicitly for several Archimedean families and certain generator transformations;
  see \cite{HofertMcNeilMachler2012} for more details. To compute the inverses, one can
  use numerical root-finding on $[0,1]$; see \texttt{rtrafo(..., inverse=TRUE)} in
  the \R\ package \texttt{copula}. This can be applied, e.g., in the case of Gumbel copulas.

  The following example shows the case of a Clayton copula family, for which
  \eqref{CDM:Clayton} can be given explicitly and thus where the CDM is
  tractable; this explicit formula is also utilized by \texttt{rtrafo(,
    inverse=TRUE)}.
  \begin{example}[Clayton copulas]
    If $\psi(t)=(1+t)^{-1/\theta}$, $t\ge0$, $\theta>0$, denotes a generator of
    the Archimedean Clayton copula, then
    $\psi^{(j)}(t)=(-1)^j(1+t)^{-(j+1/\theta)}\prod_{l=0}^{j-1}(l+1/\theta)$. Therefore, \eqref{cond:cop:Clayton} equals
    \begin{align*}
      C(u_j\,|\,u_1,\dots,u_{j-1})=\Biggl(\frac{1-j+\sum_{l=1}^ju_l^{-\theta}}{2-j+\sum_{l=1}^{j-1}u_l^{-\theta}}\Biggr)^{j-1/\theta}
    \end{align*}
    and \eqref{CDM:Clayton} equals
    \begin{align*}
      &C^-(u_j\,|\,u_1,\dots,u_{j-1})=\\
      &\biggl(1+\biggl(1-(j-1)+\sum_{l=1}^{j-1}u_l^{-\theta}\biggr)\Bigl({u_j}^{-\frac{1}{j-1+1/\theta}}-1\Bigr)\biggr)^{-\frac{1}{\theta}}.
    \end{align*}
    Figure~\ref{fig:qrng:C} displays 1000 samples from a Clayton copula with
    $\theta=2$ (Kendall's tau equals 0.5), once drawn with a PRNG (top) and
    once with a QRNG (bottom).
  \end{example}

\subsubsection*{Marshall--Olkin copulas}
A class of bivariate copulas for which $C^-(u_2\,|\,u_1)$ is explicit is the
class of Marshall--Olkin copulas
$C(u_1,u_2)=\min\{u_1^{1-\alpha_1}u_2,u_1u_2^{1-\alpha_2}\}$,
$\alpha_1,\alpha_2\in(0,1)$, where one can show that
\begin{align*}
  &\phantom{{}={}}C^-(u_2\,|\,u_1)\\
  &=\begin{cases}
    \frac{u_1^{\alpha_1}}{1-\alpha_1}u_2,&\text{if}\,\
    u_2\in[0,(1-\alpha_1)u_1^{\alpha_1(1/\alpha_2-1)}],\\
    u_1^{\alpha_1/\alpha_2},&\text{if}\,\
    u_2\in((1-\alpha_1)u_1^{\alpha_1(1/\alpha_2-1)},\
    u_1^{\alpha_1(1/\alpha_2-1)}),\\
    u_2^{\frac{1}{1-\alpha_2}},&\text{if}\,\
    u_2\in[u_1^{\alpha_1(1/\alpha_2-1)},1].
  \end{cases}
\end{align*}
Figure~\ref{fig:qrng:MO} shows 1000 samples, once drawn from a PRNG (top)
and once from a QRNG (bottom). Here again we can visually confirm the low discrepancy.
\begin{figure}[htbp]
  \centering
  \includegraphics[width=0.425\textwidth]{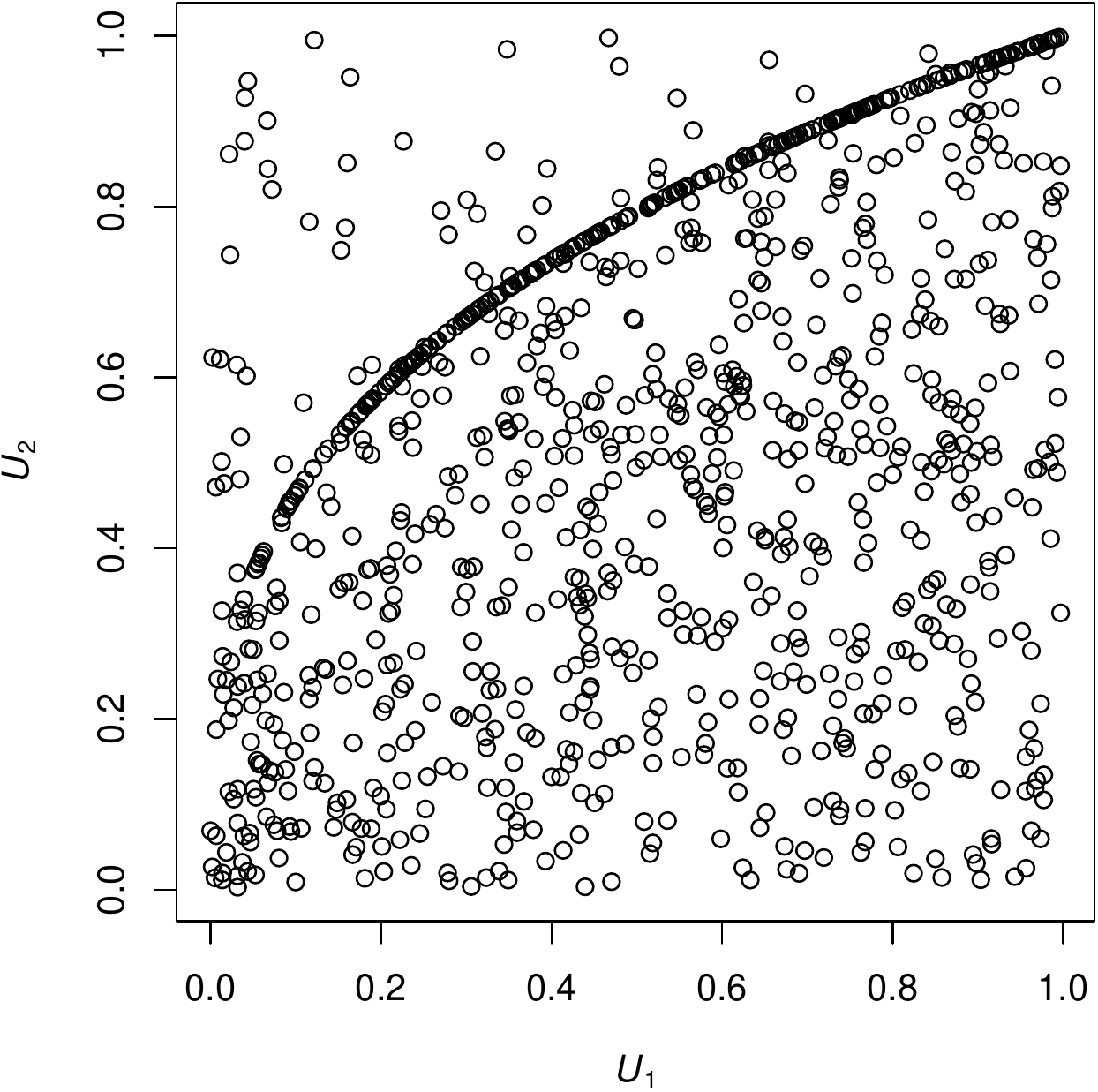}\\[2mm]
  \includegraphics[width=0.425\textwidth]{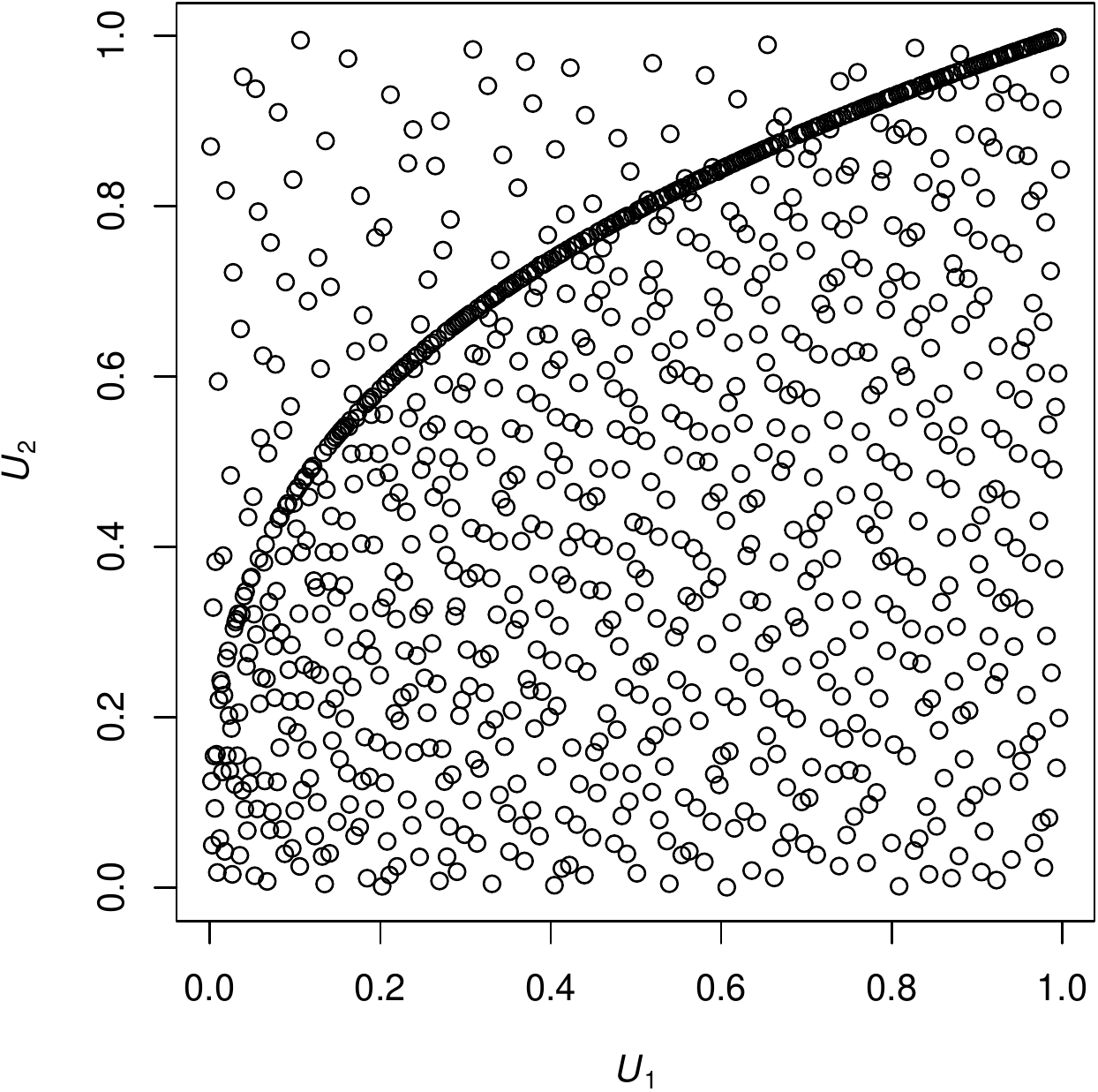}%
  \caption{1000 realizations of a Marshall--Olkin copula with
    $\alpha_1=0.25$ and $\alpha_2=0.75$ (Kendall's tau equals roughly 0.23),
    generated by a PRNG (top) and by a QRNG (bottom).}
  \label{fig:qrng:MO}
\end{figure}

Another class of copulas not discussed in this section which is naturally
sampled with the CDM and thus can easily be adapted to construct corresponding
quasi-random numbers is the class of pair copula constructions; see, e.g.,
\cite{kurowickacooke2007}. For this purpose, we modified the function
\texttt{RVineSim()} in the \R\ package \texttt{VineCopula} (version $\ge$
1.3). It now allows to pass a matrix of quasi-random numbers to be transformed
to the corresponding samples from a pair copula construction; see
the vignette \texttt{qrng} in the \R\ package \texttt{copula} for examples. Note that
if sampling of the R-vine involves numerical root-finding (required for
certain copula families), the corresponding numerical inaccuracy may have an
effect on the low discrepancy of the generated samples.
%; as far as we are aware
%such effects have not been studied yet.
% It is known that a projection of a low-discrepancy sequence to $d'<d$
% dimensions does not possess the same (desirable) properties as a
% low-discrepancy sequence in $d'$ dimensions. This needs to be kept in mind
% when looking at scatter-plot matrices (so $d'=2$) of $\phi_C$-transformed
% $d$-dimensional low-discrepancy sequences; see
% Figure~\ref{fig:qrng:t3:pairs} for a pairs plot of three-dimensional
% quasi-random samples from a $t$ copula with three degrees of freedom and
% correlation parameter $\rho=1/\sqrt{2}$ (Kendall's
% tau equals 0.5).
% \begin{figure}[htbp]
%   \centering
%   \includegraphics[width=0.8\textwidth]{fig_qrng_t3_0.5_pairs}%
%   \caption{Scatter plot matrix of 1000 realizations of three-dimensional
%   quasi-random samples from a $t$ copula with three degrees of freedom and
%   correlation parameter $\rho=1/\sqrt{2}$ (Kendall's tau equals
%   0.5).}
%   \label{fig:qrng:t3:pairs}
% \end{figure}

\subsection{Stochastic representations ($k\ge d$, typically $k>d$)}
\label{subsec:StochRep}
\subsubsection{Theoretical background}
As mentioned above, pair-copula constructions are one of the rare copula classes
for which the CDM is applied in practice. For most other copula classes and
families, faster sampling algorithms derived from stochastic representations of
$\bm{U}\sim C$ are known, especially for $d\gg2$. They are mostly class- and
family-specific, as can be seen in the examples below.
% ; see, e.g., the Marshall--Olkin sampling
% algorithm for Archimedean copulas below. However, if both
% the conditional copulas and their inverses are given explicitly, the CDM has the
% advantage that it is based on a one-to-one transformation, thus being expected
% to preserve features such as low discrepancy. The same argument also applies to
% other one-to-one transformations $\phi_C$, such as $\phi_C^{\text{WVSM}}$
% (applying to Archimedean copulas only) addressed in the appendix.
%As\footnote{MH: I would omit (or at least
%  weaken) this sentence as for most copulas, $k$ is only $d+1,\dots,d+3$, so not
%really increasing the dimension by much...} the performance of RQMC algorithms generally deteriorate when dimensions increase, algorithms that transform $k$-dimensional samples into $d$-dimensional samples may not be preferred. This will be further discussed in Section~\ref{sec:discrep}.

\subsubsection{Examples}
\label{subsubsec:ExStochRep}
\subsubsection*{Elliptical copulas}
Gauss and $t$ copulas are typically sampled via their stochastic
representations.
\paragraph{Gauss copulas.} A random vector $\bm{X}\sim\Phi_P$ admits the
stochastic representation $\bm{X}=A\bm{Z}$ where $A$ denotes the lower
triangular matrix from the Cholesky decomposition $P=AA\T$ and $\bm{Z}$ is a
vector of $d$ independent standard normal random variables. A random vector
$\bm{U}\sim C_P^{\text{Ga}}$ thus admits the stochastic representation
$\Phi(\bm{X})=\Phi(A\bm{Z})$ for
$\bm{Z}=(\Phi^{-1}(U_1'),\dots,$ $\Phi^{-1}(U_d'))$ and $\bm{U}'\sim\U[0,1]^d$;
here $\Phi$ is assumed to act on $A\bm{Z}$ component-wise. Note that for Gauss
copulas, this sampling approach is equivalent to the CDM.
% This is easily seen for d=2: The Cholesky factor is
%    A = (1,   0,
%         rho, sqrt(1-rho^2))
% and thus the stochastic representation of (U_1,U_2) satisfies
% (Phi^-(U_1), Phi^-(U_2)) = (Phi^-(U_1'), rho*Phi^-(U_1') + sqrt(1-rho^2)*Phi^-(U_2')).
% The same formula is obtained when using (U_1', U_2') in the CDM.
% For d > 2 this is less obvious, but can be shown with R (checked in three dimensions).
% It can probably be shown via linearity as the jth component of the stochastic
% representation is Phi(sum_{k=1}^j a_{jk} Phi^-(U_k')) and the jth component of
% the CDM is Phi(mu(<previous Us and thus U's>)+sqrt(P.)*Phi^-(U_j')).

\paragraph{$t$ copulas.} A random vector $\bm{X}\sim t_{\nu,P}$ admits the
stochastic representation $\bm{X}=\sqrt{W}A\bm{Z}$ where $A$ and $\bm{Z}$ are as
above and $W=1/\Gamma$ for $\Gamma$ following a Gamma distribution with shape
and rate equal to $\nu/2$. A random vector $\bm{U}\sim C_{\nu,P}^{t}$ thus
admits the stochastic representation $t_\nu(\bm{X})=t_{\nu}(\sqrt{W}A\bm{Z})$;
as before, $t_{\nu}$ is assumed to act on $\sqrt{W}A\bm{Z}$ component-wise. Note
that for $t$ copulas with finite $\nu$, this sampling approach is different from the
CDM. % as the former uses d+1 random variables and
                                % the latter only d

\subsubsection*{Archimedean copulas}
The conditional independence approach behind the Marshall--Olkin algorithm for
sampling Archimedean copulas is one example for transformations $\phi_C$ for
$k>d$; see \cite{marshallolkin1988}. For this algorithm, $k=d+1$ and one uses
the fact that for an Archimedean copula $C$ with completely monotone generator
$\psi$,
\begin{align}
  \bm{U}=(\psi(E_1/V),\dots,\psi(E_d/V))\sim C,\label{eq:stoch:rep:AC}
\end{align}
where $V\sim F=\LSi[\psi]$, independent of $E_1,\dots,E_d\sim\Exp(1)$; here,
$F=\LSi[\psi]$ denotes the distribution function corresponding to $\psi$ by
Bernstein's Theorem ($\LS^{-1}[\cdot]$ denotes the inverse Laplace--Stieltjes
transform). To give an explicit expression for the transformation
$\phi_C=\phi_C^{\text{MO}}$ in this case, we
assume that $v_1$ is used to generate $V$ via the inversion method,
%(which we denote as $G_{V,\theta}^{-1})$,
and $v_{j+1}$ is used to generate $E_j$, for $j\in\{1,\ldots,d\}$. Then we have that
$\phi_C^{\text{MO}} = (\phi^{\text{MO}}_{C,1},\ldots,\phi^{\text{MO}}_{C,d})$, where
\begin{align}
  \phi^{\text{MO}}_{C,j}=\phi^{\text{MO}}_{C,j}(v_1,v_{j+1}) = \psi \left(\frac{-\log v_{j+1}}{F^-(v_1)}\right),
  \quad j\in\{1,\ldots,d\}.\label{eq:PhiC_MO}
\end{align}

%\blue{
We can use a low-discrepancy sequence in $k=d+1$ dimensions to produce a
sample based on this method. Having $k=d+1$ instead of $k=d$ is a slight
disadvantage, since it is well known that the performance of QMC methods tends to deteriorate with
increasing dimensions. %}
%%%Feb 12
\iffalse
In addition, since the transformation used is not
one-to-one, this creates more doubt as to whether or not any smoothness of
$\Psi$ will be preserved when composed with $\phi_C^{\text{MO}}$ (following the first
approach mentioned in the introduction), or, alternatively, whether or not a
low-discrepancy point set $P_n$ with respect to the uniform distribution will
also have a low discrepancy with respect to $C$ when $\phi_C^{\text{MO}}$ is applied to each
of its points.
\fi

% It currently remains unclear whether and how transformations $\phi_C$ with $k>d$
% can be used to construct quasi-random numbers from copula models. The purposes
% of this section is to briefly show the main challenge in following such a
% route.

%%%Feb 12
%\blue{
To explore the effect of the transformation $\phi_C$ on $P_n$, we generated 1000 realizations of a
three-dimensional Halton sequence; see the top of
Figure~\ref{fig:qrng:col:U:CDM} where we colored points falling in two
non-overlapping regions in $[0,1)^2$. The first two of the three dimensions are then mapped
via $\phi_C^{\text{CDM}}$ (see the bottom of
Figure~\ref{fig:qrng:col:U:CDM}) to a Clayton copula with parameter
$\theta=2$ (such that Kendall's tau equals 0.5). As we can see, the
non-overlapping colored regions remain non-overlapping after the one-to-one
transformations have been applied.
To study the effect of the Marshall-Olkin algorithm, we look at when the first dimension of the Halton sequence is mapped to a Gamma $\Gamma(1/\theta,1)$ distribution by inversion of $v_1$(the distribution of $V$ in \eqref{eq:stoch:rep:AC} for a Clayton copula) and the last two to unit
exponential distributions (by inversion of $1-v_j$ for $j=2,3$, so that the obtained $u_j$ is increasing in each of $v_1$ and $v_{j+1}$ for $j=1,2$). The top of  Figure~\ref{fig:qrng:col:U:MO} shows the second and third coordinates of the Halton sequence, and colors the points belonging to two different three-dimensional intervals (this is why not all two-dimensional points are coloured in the two-dimensional projected regions).  We see on the bottom of Figure~\ref{fig:qrng:col:U:MO} that here again, the colored regions remain non-overlapping. However, it should also be clear that two points in a given interval defined over the second and third dimension could end up in very different locations after this transform, if the corresponding first coordinates are far apart.
Hence, the fact that the Marshall-Olkin transform uses $k=d+1$ uniforms (and thus is not one-to-one)
makes it more challenging to understand its effect when used with quasi-random numbers. On the other hand, because it is designed so that the first uniform $v_1$ is very important, it may work quite well with QMC since these methods are known to perform better when a small number of variables are important (i.e., see  \cite{DiPi10,Lemieux2009}).  %}
This combination (QRNG with the Marshall--Olkin approach) is studied further in Section \ref{sec:discrep}, with
numerical results provided in Section \ref{sec:num}. % support this claim.

\iffalse
On the contrary, on the bottom of
Figure~\ref{fig:qrng:col:U:CDM:MO} we show what happens when the first two
dimensions of the three-dimensional Halton sequence are mapped to unit
exponential distributions and the last one to a Gamma $\Gamma(1/\theta,1)$
distribution (the distribution of $V$ in \eqref{eq:stoch:rep:AC} for a Clayton
copula). After applying this transform, the two sets of colored
points ``overlap''; what we see is indeed a projection of a marginally
transformed three-dimensional quasi-random sequence (with $\Exp(1)$, $\Exp(1)$,
$\Gamma(1/\theta,1)$) to two dimensions, where the
quasi-randomness \emph{seems} to have been destroyed.
\fi
\begin{figure}[htbp]
  \centering
  \includegraphics[width=0.425\textwidth]{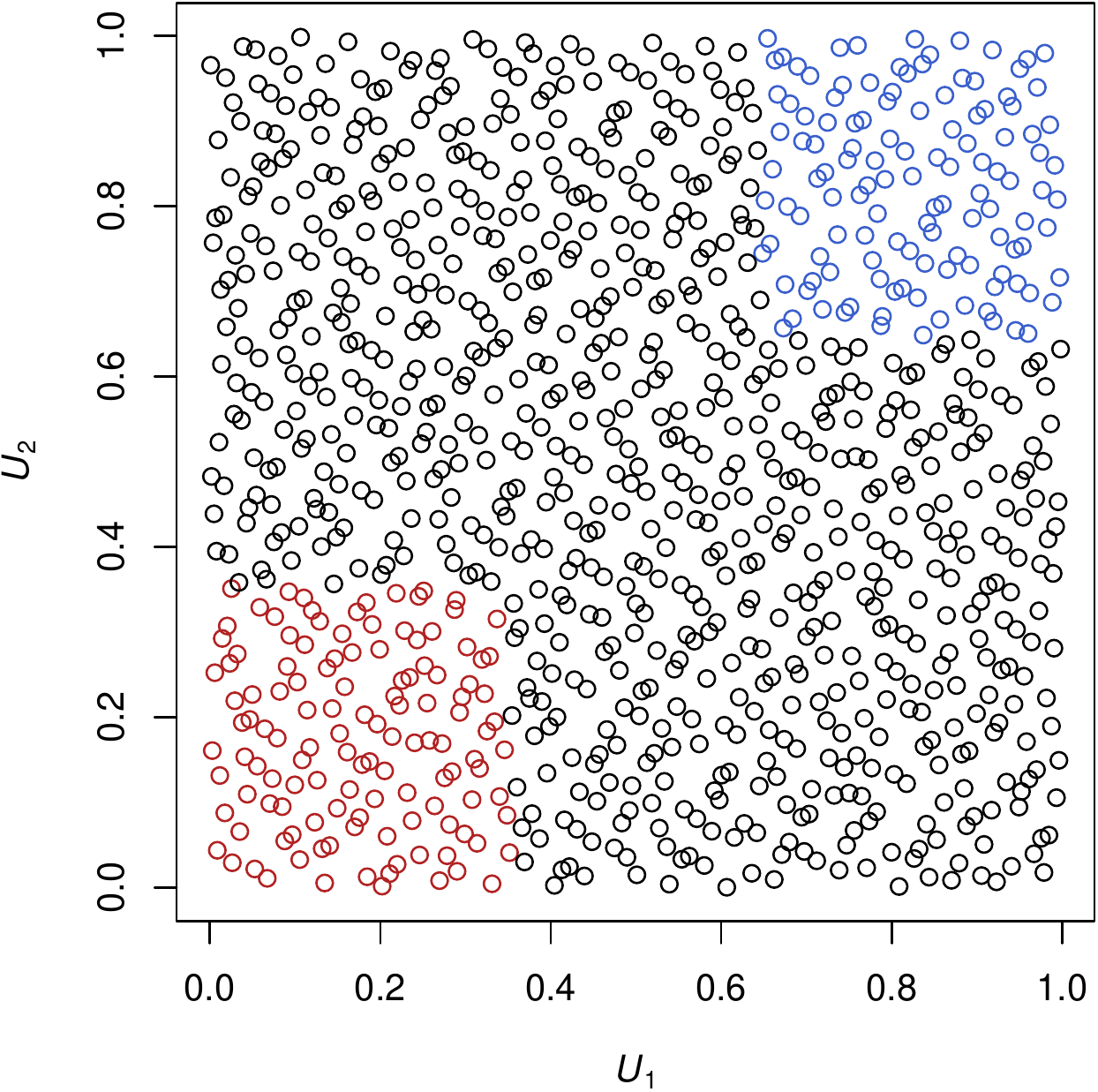}\\[2mm]
  \includegraphics[width=0.425\textwidth]{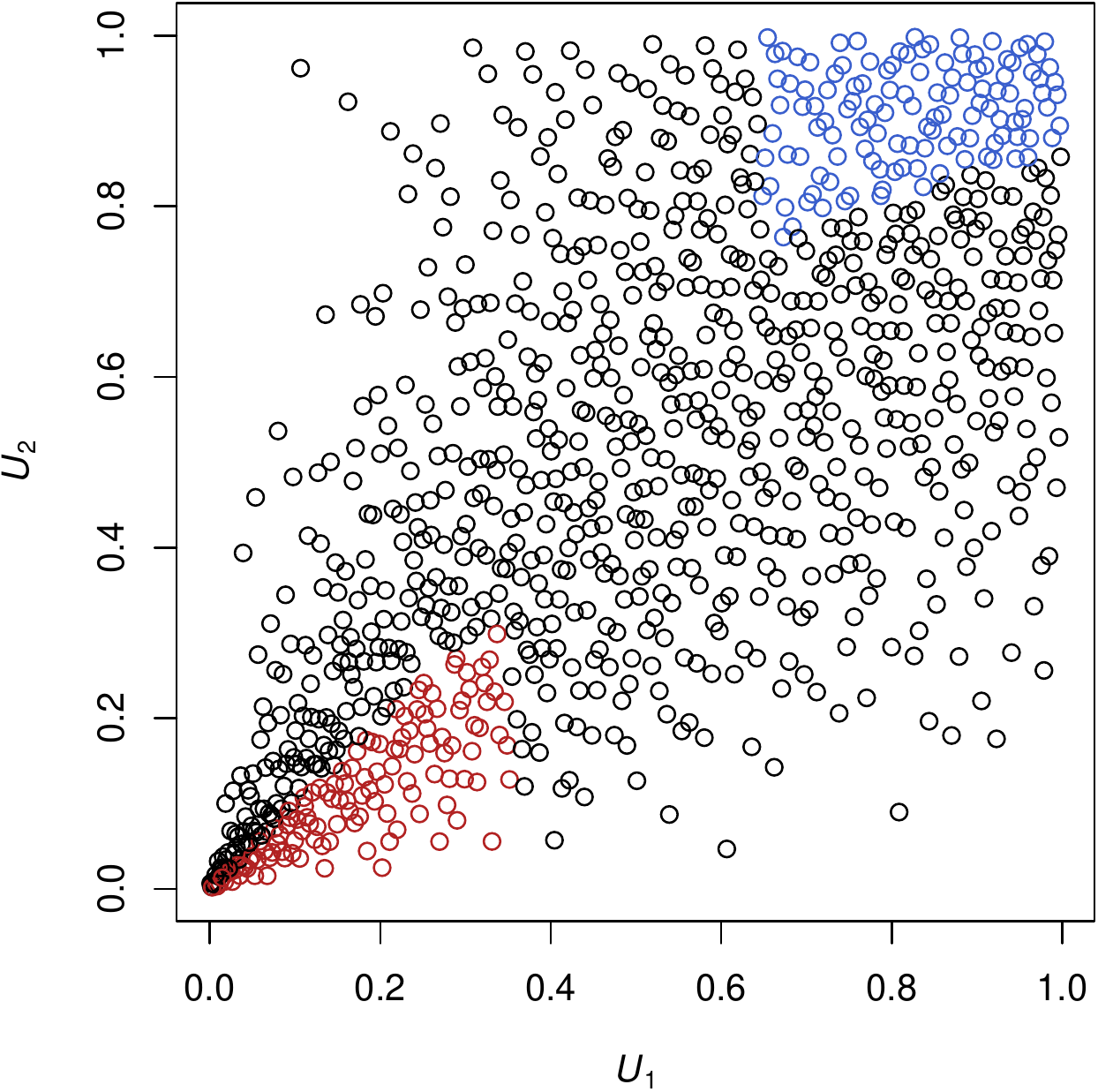}%
  \caption{1000 realizations of the first two components of a three-dimensional
    Halton sequence with colored points in the regions $[0,\sqrt{1/8}]^2$ and
    $[1-\sqrt{1/8},1]^2$ (top): corresponding $\phi_C^{\text{CDM}}$-transformed
    points (bottom) to a Clayton copula with $\theta=2$ (Kendall's tau equals
    0.5).}
  \label{fig:qrng:col:U:CDM}
\end{figure}

\begin{figure}[htbp]
  \centering
  \includegraphics[width=0.425\textwidth]{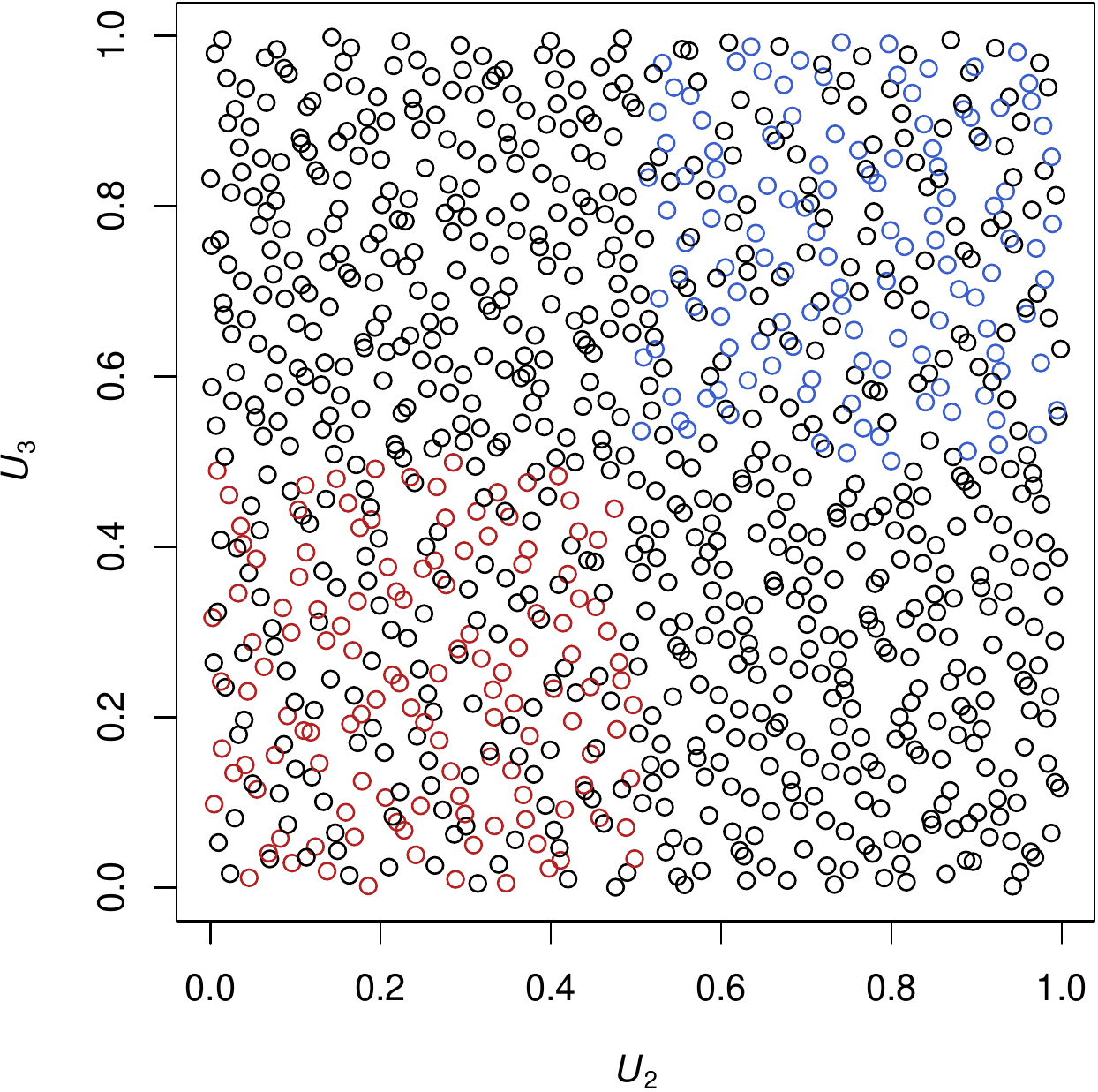}\\[2mm]
  \includegraphics[width=0.425\textwidth]{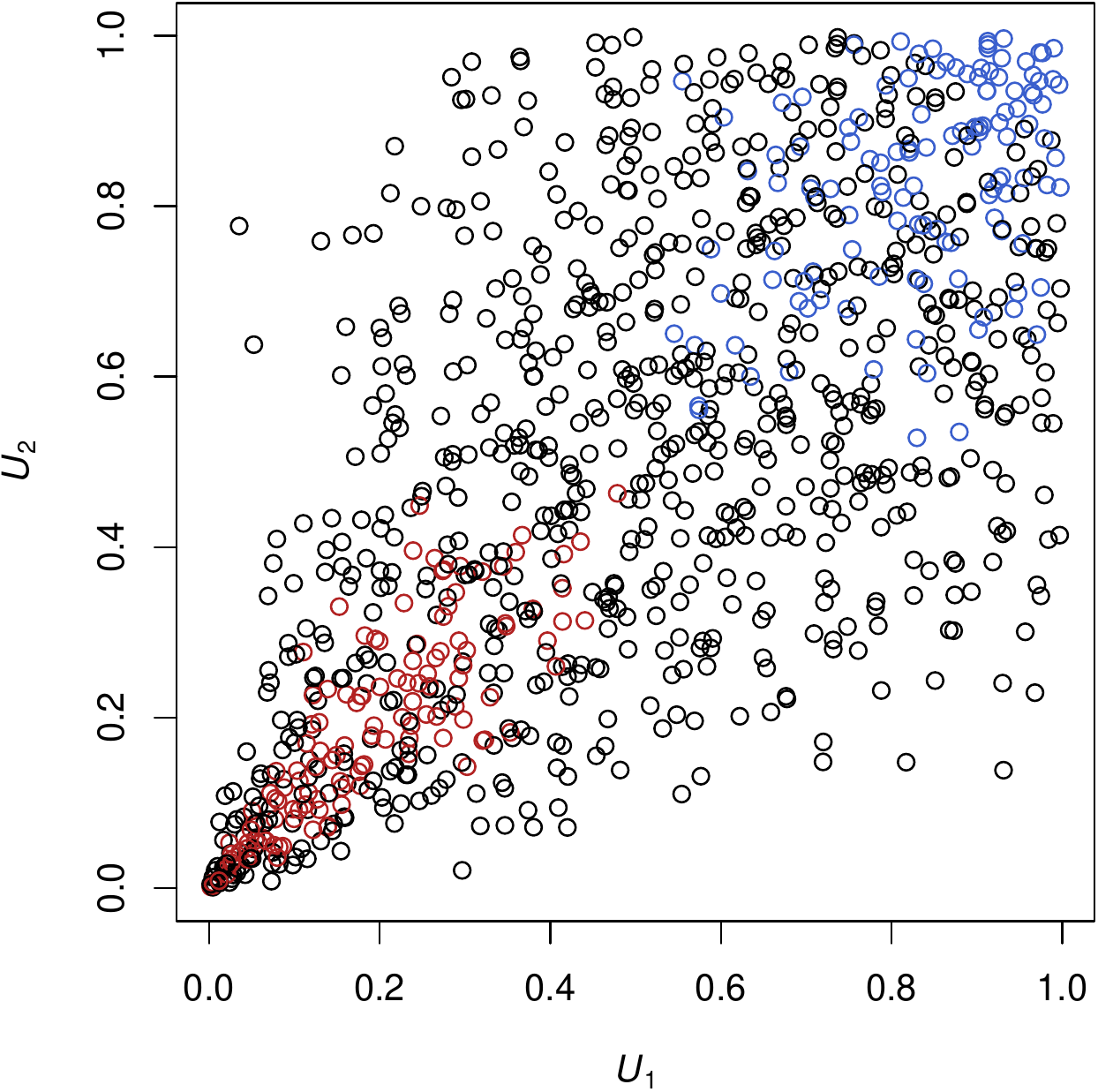}%
  \caption{1000 realizations of the second and third components of a three-dimensional
    Halton sequence with colored points corresponding to the regions $[0,0.5]^3$ and
    $[0.5,1]^3$ (top): corresponding $\phi_C^{\text{MO}}$-transformed
    points (bottom) to a Clayton copula with $\theta=2$ (Kendall's tau equals
    0.5).}
  \label{fig:qrng:col:U:MO}
\end{figure}
%A similar experiment based on how colored squares are mapped by the same
%$\phi_C$'s can be found in the \R\ package \texttt{copula}; see the vignette \texttt{qrng}.
%

\subsubsection*{Marshall--Olkin copulas}
Bivariate ($d=2$)
Marshall--Olkin copulas $C$ also allow for a stochastic representation in our
framework $\phi_C$ for $k>d$. For example, it is easy to check that
for $(U_1',U_2',U_3')\sim\U[0,1]^3$,
\begin{align*}
  (\max\{U_1'^{\frac{1}{1-\alpha_1}},\ U_3'^{\frac{1}{\alpha_1}}\},\ \max\{U_2'^{\frac{1}{1-\alpha_2}},\ U_3'^{\frac{1}{\alpha_2}}\})\sim C.
\end{align*}
This construction can be generalized to $d>2$ (but we omit further details about
Marshall--Olkin copulas in the remaining part of this paper).

\subsection{Words of caution}
The plots showing copula samples obtained from QRNGs that we have seen so far
have been promising, in that the additional uniformity (or low discrepancy)
compared to pseudo-sampling was visible. Here we want to add a word of caution
to the effect that it is crucial to work with high quality quasi-random numbers,
as defects that exist with respect to their uniformity on the unit cube will
translate into poor copula samples. Figure~\ref{fig:CDM_Hal_and_Ghal_proj=20_21}
illustrates this by showing two-dimensional copula samples obtained from
quasi-random numbers of poor quality, corresponding to the projection on
coordinates (20,21) of the first 1000 points of the Halton sequence (top) and a
similar sample obtained from a generalized Halton sequence (bottom), which was
designed to address defects of this type in the Halton sequence. More precisely,
here the problem is that this particular projection is based on the twin prime
numbers 71 and 73 for the base. Defects of this type are discussed further,
e.g., in \cite{morokoffcaflisch94}.

\begin{figure}[htbp]
  \centering
  \includegraphics[width=0.425\textwidth]{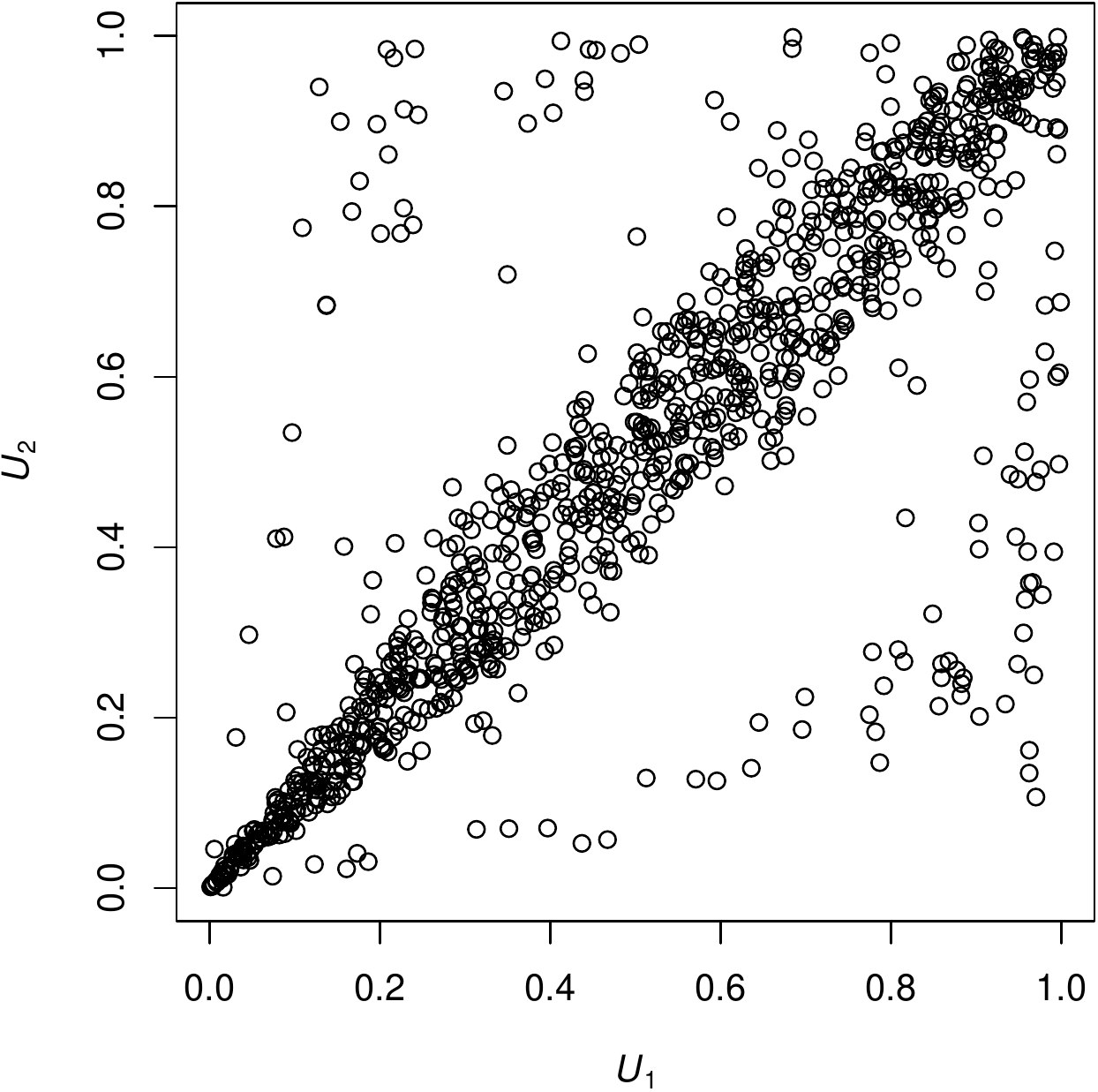}\\[2mm]
  \includegraphics[width=0.425\textwidth]{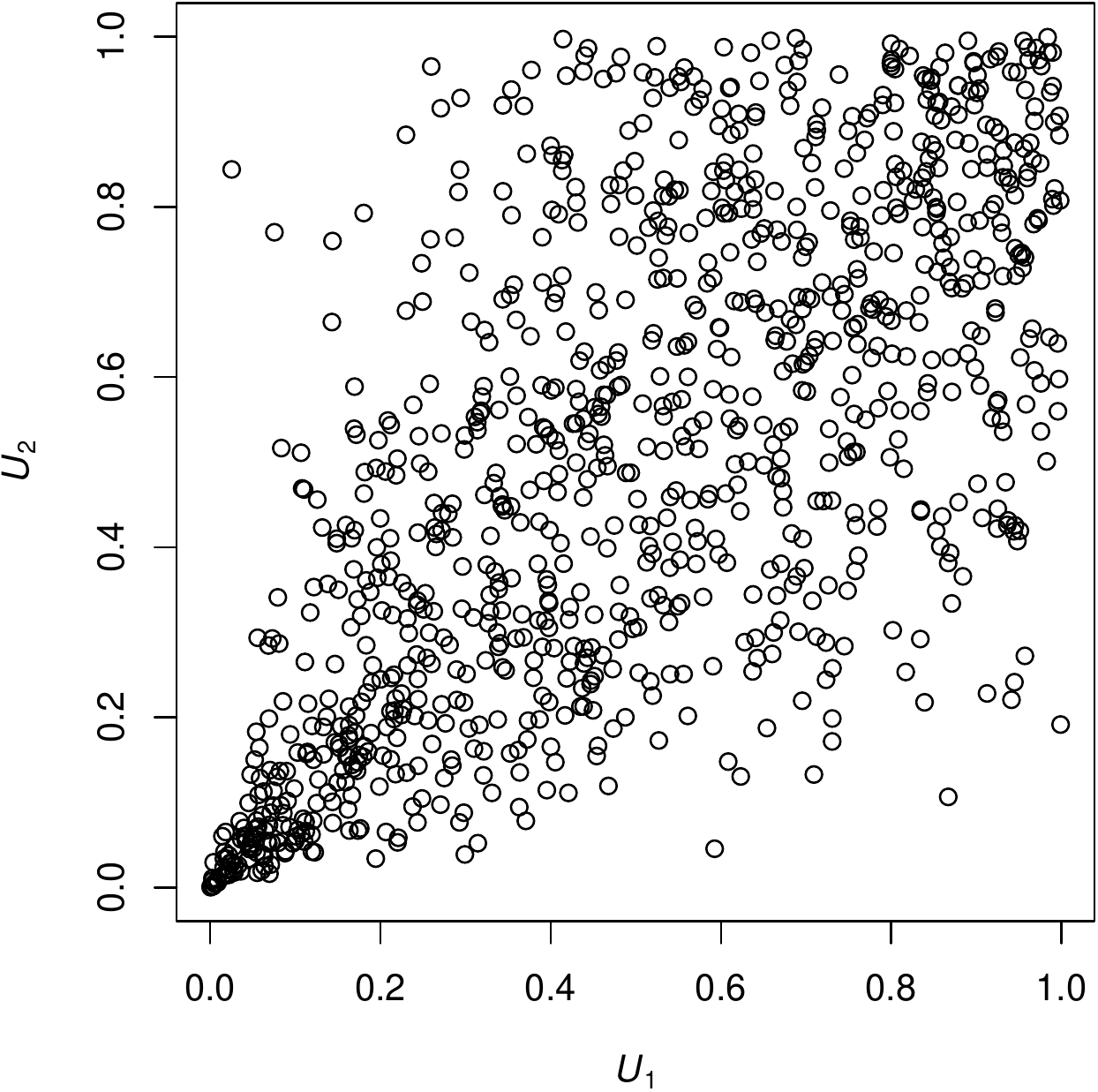}%
  \caption{Samples obtained from Clayton copula with $\theta=2$ with the CDM
    method based on coordinates 20 and 21 of the Halton sequence (top) and
    the generalized Halton sequence (bottom).}
  \label{fig:CDM_Hal_and_Ghal_proj=20_21}
\end{figure}

%%% END INSERT MARIUS PORTION

%%%%%%%%%%%%%%%%%%%%%%%%%%%%%%%%%%%%%%%%%%%%
%\section{Minimizing the discrepancy with respect to the copula measure}
\section{Analyzing the performance of copula sampling with quasi-random numbers}
\label{sec:discrep}

In this section, we discuss the two approaches outlined in the introduction to analyze the validity of sampling algorithms for copulas that are based on low-discrepancy sequences.

%\subsection{Sup-discrepancy}\label{subsec:sup-discrepancy}
\subsection{Composing the sampling method with the function of interest $\Psi$}
\label{subsec:composition}
Our goal here is to assess the quality of a quasi-random sampling method for
copula models by viewing the transformation $\phi_C$ as being composed with the
function $\Psi$ of interest, so that we can work in the usual Koksma--Hlawka
setting based on uniform discrepancy measures.

Given that a copula transform $\phi_C=(\phi_{C,1},\dots,\phi_{C,d})$ is regular enough, denote its
Jacobian by
\begin{align*}
  J_{\phi_C}=\frac{\partial\left(\phi_{C,1},\dots,\phi_{C,d}\right)}{\partial\left(u_1,\dots,u_d\right)},
\end{align*}
and write
\begin{align*}
  \IE\left[\Psi(\bm{U})\right]&=\int_{[0,1]^d}\Psi(\bm{u})c(\bm{u})d\bm{u}\\
  &=\int_{[0,1]^d}\Psi(\phi_C(\bm{v}))c(\phi_C(\bm{v}))|J_{\phi_C}(\phi_C(\bm{v}))|\,d\bm{v}.
\end{align*}

In the case of $\phi_C=\phi_C^{\text{CDM}}$, one can easily show that
$|J_{\phi_C}(\phi_C(\bm{v}))|=c(\phi_C(\bm{v}))^{-1}$, and thus
\begin{align}\label{eq:canonical_rep}
  \IE\left[\Psi(\bm{U})\right]&=\int_{[0,1]^d}\Psi(\phi_C(\bm{v}))d\bm{v}.
\end{align}

While the properties of the CDM approach allow one to directly show
\eqref{eq:canonical_rep} in its integral form as done above, this equality holds
more generally for any transformation
$\phi_C:[0,1]^k\to[0,1]^d$ such that
$\phi_C(\bm{U}) \sim C$ whenever $\bm{U} \sim\U[0,1]^k$; see also
\cite{CaflischSurvey,PillardsCools2006}.

In the case where~\eqref{eq:canonical_rep} holds, one can apply
the Koksma--Hlawka error bound~\eqref{eq:Koksma-Hlawka} to transformed samples.

%%%%%%%%% insert CL material: Feb. 13

\begin{proposition}[Koksma--Hlawka bound for a change of variables]\label{prop:Koksma-Hlawka_variable}
  Let $\bm{U}\sim C$, $\phi_C$ such that~\eqref{eq:canonical_rep} holds, and $\bm{u}_i=\phi_C(\bm{v}_i)$ for $P_n = \{\bm{v}_i,i=1,\ldots,n\}$ in $[0,1]^d$. Then
  \begin{align*}
    \biggl|\frac{1}{n} \sum_{i=1}^n \Psi(\bm{u}_i) -\IE[\Psi(\bm{U})]\biggr|\le D^*(P_n)V(\Psi\circ\phi_C).
  \end{align*}
\end{proposition}

Note that $V(\Psi)<\infty$ does not imply $V(\Psi\circ\phi_C)<\infty$ in
general. To get further insight into the conditions required to have a finite
bound on the integration error, we work with a slight variation of the
above bound that is given in \cite[pp.~19--20]{Niederreiter1992} (see also
\cite[(4')]{hlawkamuck1972} and \cite[(4)]{hlawka1961}), where the term
$V(\Psi\circ\phi_C)$ is replaced by an expression given in terms of the partial
derivatives of $\Psi\circ\phi_C$ assuming the latter exist and are
continuous. It is given by
  \begin{align*}
    \biggl|\frac{1}{n} \sum_{i=1}^n \Psi(\bm{u}_i) -\IE[\Psi(\bm{U})]\biggr|\le D^*(P_n) \|\Psi \circ \phi_C \|_{d,1},
  \end{align*}
  where
   \begin{align*}
  &\|\Psi \circ \phi_C \|_{d,1} =  \\
  & \sum_{l=1}^s\sum_{\bm{\alpha}} \int_{[0,1)^l} \left| \frac{\partial^l \Psi \circ \phi_C (v_{\alpha_1},\ldots,v_{\alpha_l},\bm{1})}{\partial v_{\alpha_1} \cdots \partial v_{\alpha_l}}\right| dv_{\alpha_1}\ldots dv_{\alpha_l}
   \end{align*}
   and the second sum is taken over all nonempty subsets
   $\bm{\alpha} = \{\alpha_1,\ldots,\alpha_l\} \subseteq
   \{1,\ldots,d\}$.
   Furthermore, the notation $\bm{1}$ in
   $\Psi \circ \phi_C (v_{\alpha_1},\ldots,v_{\alpha_l},\bm{1})$
   means that each variable $v_j$ with $j \notin \{\alpha_1,\ldots,\alpha_l\}$ is set to 1. \\

   The following proposition provides sufficient conditions on the functional
   $\Psi$ and on the copula $C$ to ensure that
   $\|\Psi \circ \phi_C \|_{d,1} < \infty $ when $\phi_C=\phi_C^{\text{CDM}}$.

\begin{proposition}[Conditions to have bounded variation with variable change in the CDM]
Assume that $\Psi$ has continuous mixed partial derivatives up to total order $d$ and there exist
$m,M,K>0$ such that for all $\bm{u}\in(0,1)^d$, $c(\bm{u})\geq m>0$ and 
\begin{align}\label{eq:bounded_cond_cop}
  \left|\frac{\partial^k C(u_i\,|\,u_1,\dots,u_{i-1})}{\partial u_{\alpha_1}\cdots\partial u_{\alpha_k}}\right|\leq M,\,\,\alpha_1,\dots,\alpha_k\in\{1,\dots,i\},
\end{align}
for each $1\leq k\leq i\leq d$. 
Furthermore, assume that for all $1\le k \le  l\le d$ and $\{\alpha_1,\ldots,$ $\alpha_l\}$ $\subseteq \{1,\ldots,d\}$, we have
\begin{align}\label{eq:AddConstraintPhiDerivatives}
  \left|\frac{\partial^k \Psi(u_{1},\ldots,u_d) }{\partial u_{\beta_1}\cdots\partial u_{\beta_k}}\right|\leq K,\quad \beta_j  \in \{\alpha_1,\ldots,\alpha_l\},\ 1 \le j \le k.
\end{align}
Then there exists a constant $C^{(d)}$ (independent of $n$ but dependent on $\Psi$) such that for the choice $\phi_C = \phi_C^{\text{CDM}}$, we have
\begin{align*}
   \biggl|\frac{1}{n} \sum_{i=1}^n \Psi(\bm{u}_i) -\IE[\Psi(\bm{U})]\biggr|\le D^*(\bm{v}_1,\dots,\bm{v}_n)  K C^{(d)} (M^d/m)^{2d-1},
\end{align*}
where $\bm{u}_i = \phi_C^{\text{CDM}}(\bm{v}_i)$, $i=1,\ldots,n$.
\end{proposition}

\begin{proof}
See~\cite[(11) and the remark thereafter]{hlawkamuck1972}.
\end{proof}

\begin{remark}\label{rem:HighOrder}
  \begin{enumerate}
  \item As we will see in the discussion preceding the next proposition, in general,
    to ensure that $\|\Psi \circ \phi_C \|_{d,1} < \infty $ holds, a possible
    approach is to bound the mixed partial derivatives involving $\Psi$ and then
    to verify that the mixed partial derivatives involving $\phi_C$ are
    integrable. As explained in \cite{hlawkamuck1972}, Condition
    \eqref{eq:bounded_cond_cop} ensures that the latter condition is verified in
    the case of the CDM (or Rosenblatt) transform, and avoids having to deal with
    the function $\phi_C$ and its partial derivatives. Unfortunately (and while it
    may seem easier to work with the conditional copulas
    $C(u_j\,|\,u_1,\ldots,u_{j-1})$ than with $\phi_C$), in many cases the copulas
    involved do not have bounded mixed partial derivatives everywhere, with
    singularities appearing near the boundaries when one or more arguments are 0
    or 1.  A non-trivial case where we were able to verify
    \eqref{eq:bounded_cond_cop} is for the
    Eyraud-Farlie-Gumbel-Morgenstern copula (see \cite{jaworskidurantehaerdlerychlik2010}), assuming the parameters %$\theta_J$
    are chosen so that the density $c(\bm{u})$ and thus the denominator %$D$
    of
    $C(u_j|u_1,\ldots,u_{j-1})$ is bounded away from 0 for all
    $\bm{u} \in [0,1)^d$.
  \item\label{rem:HighOrder:2} We note that the conditions given in  (\ref{eq:AddConstraintPhiDerivatives}) are not the same as those required to prove that
    \begin{align*}
      \|\Psi\|_{d,1} =\sum_{l=1}^d\sum_{\bm{\alpha}} \int\limits_{[0,1)^l}\!\! \left| \frac{\partial^l \Psi (u_{\alpha_1},\ldots,u_{\alpha_l},\bm{1})}{\partial u_{\alpha_1} \cdots \partial u_{\alpha_l}}\right| du_{\alpha_1}\ldots du_{\alpha_l}
    \end{align*}
    is bounded; in the latter case, we only need to consider mixed partial derivatives of order at most one in each variable (since the $\alpha_j$'s are distinct). However, in \eqref{eq:AddConstraintPhiDerivatives}, the $\beta_j$'s are not necessarily distinct. In particular, this means that we need to consider the partial derivative of $\Psi$ of order $d$ with respect to each variable and make sure it is bounded.
  \end{enumerate}
\end{remark}

 Let us now move away from the CDM method and consider a general transformation $\phi_C$.
 In order to study $\|\Psi \circ \phi_C \|_{d,1}$, we first need to decompose mixed partial derivatives of the form
 \begin{align*}
   \frac{\partial^l (\Psi \circ \phi_C)(v_{\alpha_1},\ldots,v_{\alpha_l},\bm{1})}{\partial v_{\alpha_1} \cdots \partial v_{\alpha_l}}
 \end{align*}
 in terms of $\Psi$ and $\phi_C$ separately. To do so, we
 follow \cite{hlawkamuck1972}, as well as \cite[Theorem~2.1]{const1996}, and obtain
 an expression for the mixed partial derivative of a composition of functions
 via the representation
 \begin{align}
   &\phantom{{}.{}}\frac{\partial^l \Psi \circ \phi_C(v_{\alpha_1},\ldots,v_{\alpha_l},\bm{1})}{\partial v_{\alpha_1} \cdots \partial v_{\alpha_l}}\notag\\
   &=\!\!\!\!\!\!\sum_{1 \le |\bm{\beta}| \le l} \!\!\frac{\partial^{|\bm{\beta}|} \Psi }{\partial^{\beta_1} u_{1} \ldots \partial^{\beta_d} u_{d}}
     \sum_{s=1}^l \sum_{\gamma,\bm{k}}
     c_{\gamma} \prod_{j=1}^s \!\!\frac{\partial^{|\gamma_j|} \phi_{C,k_j}(v_{\alpha_1},\ldots,v_{\alpha_l},\bm{1})}{\partial^{\gamma_{j,1}} v_{\alpha_1}\ldots \partial^{\gamma_{j,l}} v_{\alpha_l}} \label{eq:FaaDiBruno}
 % \sum_{k=1}^l
 %\sum_{v_1 \in \bm{\alpha}} \ldots \sum_{v_k \in \bm{\alpha}} \left| \frac{\partial \Psi(u_{\alpha_1},\ldots,u_{\alpha_l},\bm{1}}{\partial u_{v_1}\ldots \partial u_{v_k}} \right||P_l(\partial \phi_{C,v_1}\ldots\phi_{C,v_k})|
 \end{align}
 where $\bm{\beta} \in \mathbb{N}_0^d$ and
 $|\bm{\beta}| = \sum_{j=1}^d \beta_j$. Here we do not specify over which values
 of $\gamma_j$ and $k_j$ the inner sum in the above expression is taken: details
 can be found in the proof of Proposition~\ref{prop:BdErrorArchMO}. But let us
 point out that in the product over $j$, each index $\alpha_1,\ldots,\alpha_l$
 appears exactly once. On the other hand -- and as noted in item \ref{rem:HighOrder:2} of
 Remark~\ref{rem:HighOrder} above -- in the mixed partial
 derivative of $\Psi$, a given variable can appear with order larger than 1.

 From \eqref{eq:FaaDiBruno}, we see that a sufficient condition to show that $ \|\Psi \circ \Phi_C \|_{d,1} < \infty$ is to establish that all products  of the form
 \begin{align}
 \label{eq:PartDerivL1Phi_C_Gen}
  \frac{\partial^{|\bm{\beta}|} \Psi }{\partial^{\beta_1} u_{1} \ldots \partial^{\beta_d} u_{d}}
 %\sum_{s=1}^l \sum_{(\bm{k},\bm{\gamma}) \in p_s(\bm{\beta},\bm{\alpha})}
 \prod_{j=1}^s \frac{\partial^{|\gamma_j|}
   \phi_{C,k_j}(v_{\alpha_1},\ldots,v_{\alpha_l},\bm{1})}{\partial^{\gamma_{j,1}}
   v_{\alpha_1}\ldots\partial^{\gamma_{j,l}} v_{\alpha_l}},\ s\in\{1,\ldots,l\},
  %\mbox{ and }(\bm{k},\bm{\gamma}) \in p_s(\bm{\beta},\bm{\alpha})
 % \frac{\partial^{|\bm{\beta}|} \Psi }{\partial^{\beta_1} u_{1} \ldots \partial^{\beta_d} u_{d}}
 %\prod_{j=1}^d \frac{\partial^{|\gamma_j|} \phi_{C,j}(v_{\alpha_1},\ldots,v_{\alpha_l},\bm{1})}%{\partial^{\gamma_{j,1}} v_{\alpha_1}\ldots\partial^{\gamma_{j,l}} v_{\alpha_l}}
%% \left| \frac{\partial^l \phi_{C,v}}{\partial u_{j_1} \ldots \partial u_{j_l}}\right| \le K(d,l)
% %\quad 1 \le l \le d; 1 \le j_i \le d
 \end{align}
 are in $L_1$.
% valid for all points in $[0,1)^d$ \christiane{check}, where $K(d,l)$ is a constant that depends only on $d$ and $l$.

 We note that for the MO algorithm (assuming as we did in \eqref{eq:PhiC_MO}
 that $v_1$ is used to generate $V$ and $v_{j+1}$ is used to generate $E_{j}$),
 $\phi_{C,j}$ is a function of $v_1$ and $v_{j+1}$ only, for
 $j\in\{1,\ldots,d\}$. Hence the only partial derivatives of $\phi_{C,j}$ that are
 nonzero are those with respect to variables in $\{v_1,v_{j+1}\}$. This
 observation is helpful to prove the following result, which shows that the
 error bound obtained when using the MO algorithm has the desired behavior
 induced by the low-discrepancy point set used to generate the copula
 samples; note that it does not show that $\Psi \circ \Phi_C$ has bounded
 variation. Its proof can be found in the appendix. %~\ref{sec:proof}.

\begin{proposition}[Error behaviour for MO for continuous $V$]
\label{prop:BdErrorArchMO}
 Let $\phi_C^{\text{MO}}$ be the transformation associated with the
 Mar\-shall--Olkin algorithm, as given in \eqref{eq:PhiC_MO}, and that $V\sim F$ is
 continuously distributed.  Let
$P_n=\{\bm{v}_i,i=1,\ldots,n\}$ be the point set in $[0,1)^{d+1}$ used to
produce copula samples via the transformation $\phi_C^{\text{MO}}$ and let $\bm{u}_i = \phi_C^{\text{MO}}(\bm{v}_i)$. If
 \begin{enumerate}
  \item the point set $P_n$ excludes the origin and there exists some $p \ge 1$ such that $\min_{1 \le i \le n} v_{i,1} \ge 1/pn$;
 \item\label{prop:err:2} the function $\Psi$ satisfies $|\Psi(\bm{u})| < \infty$ for all $\bm{u} \in [0,1)^{d+1}$ and
  \begin{align}
  \label{eq:CondPartPhiMO}
   \frac{\partial^{|\bm{\beta}|} \Psi }{\partial^{\beta_1} u_{1} \ldots \partial^{\beta_d} u_{d}} < \infty
   \mbox{ for all $\bm{\beta}=(\beta_1,\ldots,\beta_d)$},
 \end{align}
 with $\beta_l \in \{0,\ldots,d\}$ and $|\bm{\beta}| \le d$;
  \item\label{prop:err:3} and the generator $\psi(\cdot)$ of the Archimedean
    copula $C$ is such that
    \begin{enumerate}[label=\alph*)]
    \item\label{prop:err:3:a} $\psi'(t)+t\psi''(t)$ has at most
    % \TODO[Should this be ``at most'' or ``precisely'' instead of ``only''?]
      one zero $t^*$ in $(0,\infty)$
      and it satisfies $-t^*\psi'(t^*)<\infty$; and
    \item\label{prop:err:3:b} $F^{-1}(1-1/pn)$ is in $O(n^{a})$ for some constant $a>0$;
    \end{enumerate}
 \end{enumerate}
then there exists a constant $C^{(d)}$ (independent of $n$ but dependent of $\Psi$ and $\phi_C^{\text{MO}}$) such that
\begin{align*}
  \biggl|\frac{1}{n} \sum_{i=1}^n \Psi(\bm{u}_i) -\IE[\Psi(\bm{U})]\biggr|\le
  C^{(d)} (\log n) D^*(P_n).
\end{align*}
\end{proposition}

\begin{remark}
  We note that if $\IE[V]<\infty$, as is the case for Clayton's copula family,
  Condition~\ref{prop:err:3}~\ref{prop:err:3:b} can be easily checked via
  Markov's inequality. In the case of the Gumbel copula, $V$ has an $\alpha$-stable
  distribution and it can be shown that $P(V>x) \le c x^{-\alpha}$ for
  $x \ge x_0$ and for some constant $c$, where $c$ and $x_0$ depend both on the
  parameters of the distribution; see \cite[Theorem~1.12]{Nolan}. Therefore
  $F^{-1}(1-1/pn)$ can be bounded by a constant time $n^{a}$ in this case
  (namely by $\max\{x_0,(cpn)^{1/\alpha}\}$). As for Condition~\ref{prop:err:3}~\ref{prop:err:3:a},
  %\TODO[Shouldn't this be 3) a)?]
  one can show that $t^* = \theta$ and $t^*=1$ are the only zeros for the
  Clayton and Gumbel copulas, respectively.
\end{remark}

When $F$ is discrete,
%then we can use the following result, whose proof can be found in Appendix~\ref{sec:proof}. The idea here is that when $F$ is discrete,
we can split the problem
into subproblems based on the value taken by $V$. Then, in each case, the bounded variation condition is much easier to verify, because the transformation $\phi_C$ given $V$ is essentially one-dimensional as it is mapping each $v_j$ to an exponential $E_{j-1}$ for $j\in\{2,\ldots,d+1\}$.
%We note that the Frank, Joe and Ali-Mikhail-Haq copulas are such that $F$ is discrete, and thus the next result applies to them.
Its proof can be found in the appendix. %~\ref{sec:proof}.

\begin{proposition}[Error behaviour for the Marshall--Olkin algorithm for discrete $V$]
\label{prop:MOdisc}
Let $\phi_C^{\text{MO}}$ be the transformation associated with the Marshall--Olkin
algorithm, as given in \eqref{eq:PhiC_MO} and assume $C$ is an Archimedean
copula whose distribution function $F$ of $V$ is discrete. Let
$P_n=\{\bm{v}_i,i=1,\ldots,n\}$ be the point set in $[0,1)^{d+1}$ used to
produce copula samples via the transformation $\phi_C^{\text{MO}}$ and let $\bm{u}_i = \phi_C^{\text{MO}}(\bm{v}_i)$. If
\eqref{eq:CondPartPhiMO} holds and
 \begin{enumerate}
  \item  there exists some $p \ge 1$ such that the point set $P_n$ satisfies $\max_{1 \le i \le n} v_{i,1} \le 1-1/pn$;
  \item there exist constants $c>0$ and $q \in (0,1)$ such that $1-F(l) \le  cq^l$ for $l \ge 1$;
\end{enumerate}
then there exists a constant $C^{(d)}$  (independent of $n$ but dependent of $\Psi$ and $\phi_C^{\text{MO}}$) such that
\begin{align*}
\left|\frac{1}{n} \sum_{i=1}^n \Psi(\bm{u}_i) -\IE[\Psi(\bm{U})]\right|\le
C^{(d)} (\log n) D^*(P_n).
\end{align*}
\end{proposition}

%new Feb 2016
\begin{remark}
We note that the Frank, Joe and Ali-Mikhail-Haq copulas are such that $F$ is discrete. The condition on the tail of $F$ stated in the proposition can be shown to hold  for the Frank and Ali-Mikhail-Haq copulas, but not for the Joe copula (the distribution of $V$ in this case has a Sibuya distribution, for which no moments exists, i.e., it has a very fat tail).
\end{remark}

Let us now move on to RQMC methods. We already mentioned that an advantage they have over their
deterministic counterparts is that much weaker conditions are required to provide variance expressions for their corresponding estimators.
The following result shows that this also holds after composing $\Psi$ with
$\phi_C$.
%that the expressions for the variance of RQMC estimators given in Theorem
%\ref{thm:VarFourierWalsh} (see the appendix) remain valid after composing $\Psi$ with
%$\phi_C$, compared to what is needed to make sure the error remains bounded in
%the deterministic case. 
 \begin{proposition}[Variance expression with a change of variables]
 \label{prop:VarFourComp}
 If $\tilde{P}_n = \{\widetilde{\bm{v}}_1,\ldots,\widetilde{\bm{v}}_n\}$ is a randomly digitally shifted net with corresponding RQMC estimator $\widehat{\mu}_n=\frac{1}{n}\sum_{i=1}^n \Psi(\phi_C(\widetilde{\bm{v}}_i))$ and if
 ${\rm Var}(\Psi(\bm{U})) < \infty$ with $\bm{U}\sim C$, then we have that
 \begin{align}
 \label{eq:VarCompFctFourWalsh}
   {\rm Var}(\widehat{\mu}_n) = \sum_{\bm{0} \neq \bm{h} \in {\cal L}_d^*} |\widehat{\Psi \circ \phi_C}(\bm{h})|^2,
 \end{align}
 where
 ${\cal L}_d^*$ is the dual net of the deterministic net that has been shifted to get $\tilde{P}_n$, and $\hat{f}(\bm{h})$ is the Walsh coefficient of $f$ at $\bm{h}$, while
 \begin{align*}
  (\widehat{\Psi\circ\phi_C})(\bm{h})&=\sum_{\bm{k} \in \mathbb{Z}^d} \widehat{\Psi}(\bm{k})P(\bm{h},\bm{k}),\\
   P(\bm{h},\bm{k})&=
 \int_{[0,1)^d}
  e^{2\pi i (\bm{k}\cdot \phi_C(\bm{w})-\bm{h} \cdot \bm{w})}d\bm{w}\\
  &=
  \IE\left[e^{2\pi i (\bm{k}\cdot \phi_C(\bm{W})-\bm{h} \cdot
                     \bm{W})}\right],\quad \bm{W}\sim \U[0,1]^d.
 \end{align*}
\end{proposition}

 \begin{proof}
   %We provide a proof for lattices and their variance expression based on a
   %Fourier series expansion: a similar approach can be used to obtain the Walsh
   %representation for digital nets.

   It is clear from Theorem \ref{thm:VarFourierWalsh} in the appendix and using
   Representation \eqref{eq:canonical_rep} that \eqref{eq:VarCompFctFourWalsh}
   holds and the condition ${\rm Var}(\Psi(\bm{U})) < \infty$ with
   $\bm{U}\sim C$ ensures it is finite. So what remains to be shown is the
   expression for the Walsh coefficient of the composed function
   $\Psi \circ \phi_C$. It is obtained as follows:
   \begin{align*}
     (\widehat{\Psi\circ\phi_C})(\bm{h})&= \int_{[0,1)^d} \Psi (\phi_C(\bm{w}))e^{-2\pi i \langle \bm{h}, \bm{w} \rangle_b}d\bm{w}\\
     &= \int_{[0,1)^d} \sum_{\bm{k} \in \mathbb{Z}^d} \hat{\Psi}(\bm{k}) e^{2\pi i \langle \bm{k},\phi_C(\bm{w}\rangle_b)}e^{-2\pi i \langle \bm{h},  \bm{w} \rangle_b} d\bm{w} \\
                                        &= \sum_{\bm{k} \in \mathbb{Z}^d} \hat{\Psi}(\bm{k})\int_{[0,1)^d}
                                          e^{2\pi i (\langle \bm{k}, \phi_C(\bm{w})\rangle_b-\langle \bm{h},\bm{w}\rangle_b)}d\bm{w}\\
                                          &=\sum_{\bm{k} \in \mathbb{Z}^d} \widehat{\Psi}(\bm{k})P(\bm{h},\bm{k}),
   \end{align*}
   where the third equality holds thanks to Fubini's theorem.
%where $P(\bm{h},\bm{k})=\int_{[0,1]^d}e^{2\pi
%  i(\bm{h}\cdot\bm{u}-\bm{k}\cdot\phi_C(\bm{u}))}d\bm{u}=\IE\left[e^{2\pi
 %   i(\bm{h}\cdot\bm{V}-\bm{k}\cdot\phi_C(\bm{U}))}\right]$, where $\bm{V}\sim
%\U[0,1]^d$ and $\bm{U}\sim C$, which may be analytically computed for some
%copula examples.
% An expression for the Walsh coefficient of this composed function  can similarly be obtained.
 \end{proof}
%We note that an expression for $P(\bm{h},\bm{k})$ can be obtained for some copulas, as shown in Example \ref{ex:P} in the appendix.

 By adding assumptions on the smoothness of $\Psi$ and thus on the behavior of its Walsh coefficients, one could obtain improved convergence rates for the variance given in \eqref{eq:VarCompFctFourWalsh} compared to the $O(1/n)$ we get with MC, something we plan to study in future work.

\subsection{Transforming the low-discrepancy samples}
\label{subsec:TransfSamples}

As mentioned in the introduction, we can think of $\phi_C$ as transforming the point set $P_n$ instead of being composed with $\Psi$. The integration error can then be analyzed via a generalized version of the Koksma--Hlawka inequality such as the one studied in \cite{dick2014}, which we now explain.

Similarly to the Lebesgue case we define the {\em copula-discrepancy function} with respect to a
copula-induced measure $P_C$ on an interval $B$ (i.e., $P_C(B)=\IP(\bm{U}\in
B)$ for $\bm{U}\sim C$) as
%\TODO[I introduced $P_C$ for the copula-induced
%measure here (instead of $C(B)$). Does this change of notation affect any other places?]
\begin{align*}
E_C(B;P_n)=\frac{A(B;P_n)}{n}-P_C(B).
\end{align*}
Let $\mathcal{J}$ be the set of intervals of $[0,1)^d$ of the form $[\bm{a},\bm{b})=\prod_{j=1}^d [a_j,b_j)$, where $0\leq a_j\leq b_j\leq 1$. The {\em copula-discrepancy} $D_C$ of $P_n$ is then defined as
\begin{align}\label{eq:copula_discrepancy}
	D_C(P_n)=\sup_{B \in \mathcal{J}} \vert E_C(B;P_n) \vert,
\end{align}
and similarly for $D_C^*(P_n)$, the {\em star-copula-discrepancy function} when the $\sup$ in~\eqref{eq:copula_discrepancy} is taken over $\mathcal{J}^*$ instead.

The generalization of the Koksma--Hlawka inequality studied in \cite[Theorem~1]{dick2014} then provides
 \begin{align*}
    \biggl|\frac{1}{n} \sum_{i=1}^n \Psi(\bm{u}_i) -\IE[\Psi(\bm{U})]\biggr|\le V(\Psi)D^*_C(\bm{u}_1,\dots,\bm{u}_n),
   % \label{eq:Koksma-Hlawka_Copulas}
  \end{align*}
  where we assume $\bm{u}_i=\phi_C(\bm{v}_i)$, $i\in\{1,\ldots,n\}$. To get some
  insight on this upper bound, we need to know how
  $D^*_C(\bm{u}_1,\ldots,$ $\bm{u}_n)$ behaves as a function of $n$. Unfortunately,
  in general we cannot prove that
  $D^*(\bm{v}_1,\ldots,\bm{v}_n) \in O(n^{-1} \log^d n)$ implies that
  $D^*_C(\bm{u}_1,\ldots,\bm{u}_n) \in O(n^{-1} \log^d n)$. Here are a few
  things we can say, though.

First, %if $\widetilde{P}_n=\{\bm{u}_1,\dots,\bm{u}_n\}\subseteq [0,1)^d$ such that $\bm{u}_i=\phi_C(\bm{v}_i)$ and
an obvious case for which discrepancy is preserved is when $\phi_C$ maps rectangles to rectangles, because then  $\phi_C(B)\in\mathcal{J}$ for all $B\in\mathcal{J}$, and thus
\begin{align*}
  D_C(\widetilde{P}_n)&\leq D(P_n),\\
  D_C^\ast(\widetilde{P}_n)&\leq D^\ast(P_n),
\end{align*}
where $\widetilde{P}_n=\{\bm{u}_1,\dots,\bm{u}_n\}$. However, this only happens when $C$ is the independence copula, and in this case the equality holds. This is not a very interesting case since our focus here is on dependence modelling.

%In general, there is however no guarantee that the copula-induced discrepancy is controlled by $O(n^{-1}\log^d n).$

For the more realistic setting where $\phi_C$ does not map rectangles to rectangles, the following result from~\cite{hlawkamuck1972}
%for which a proof
%\footnote{MH: Is this new? It reads as if this is known and so the reader may
 %wonder why we give a proof} is given in Appendix~\ref{sec:proof},
 holds and gives a much slower convergence rate for $D^*_C(\widetilde{P}_n)$.

\begin{proposition}\label{prop:bounded_transform_discrep}
  Let $C$ be such that the Rosenblatt transform $\phi_C^{-1}$ is Lipschitz continuous on $[0,1]^d$ w.r.t.\ the sup-norm $\|\cdot\|_{\infty}$, and $\{\bm{u}_i=\phi_C(\bm{v}_i)\}$ for
  some sequence of points $\{\bm{v}_i\}$
  in $[0,1]^d$. Then
  \begin{align*}
    D_C(\{\bm{u}_1,\dots,\bm{u}_n\})\le c(d)D(\{\bm{v}_1,\dots,\bm{v}_n\})^{1/d},
  \end{align*}
  for some function $c(d)$, constant in $n$.
\end{proposition}

Note that the above results fully depend on the properties of $\phi_C$. The aim would then be to choose $\phi_C$
such that a low-discrepancy sequence $\{\phi_C(\bm{v}_i)\}$ w.r.t.\ the copula measure
$P_C$ results whenever applied to a low-(Lebesgue-)discrepancy sequence $\{\bm{v}_i\}$. A more fundamental approach would be to directly produce a low-discrepancy sequence $\{\bm{u}_i\}$ where the discrepancy is measured w.r.t.\ the copula measure $C$. This is something we intend to study in future work.

% In this spirit, the following gives an extension of~\eqref{eq:Koksma-Hlawka} to a copula setting.

% \begin{theorem}\label{thm:Koksma-Hlawka_Copulas}
%   Let $\bm{U}\sim C$ and $\{\bm{u}_i\}$ a sequence of points in $[0,1]^d$. Then we have
%   \begin{align}\label{eq:Koksma-Hlawka_Copulas}
%     \biggl|\frac{1}{n} \sum_{i=1}^n \Psi(\bm{u}_i) -\IE[\Psi(\bm{U})]\biggr|\le V(\Psi)D_C^\ast(\bm{u}_1,\dots,\bm{u}_n).
%   \end{align}
% \end{theorem}
% \begin{proof}
%   See \cite[Theorem 1]{dick2014}.
% \end{proof}

Now, computing $D_C$ or $D^\ast_C$ is usually not feasible in practice. If we
replace the sup-norm by the $L_2$-norm, we obtain $L_2$-discrepancies which
are usually more practical to compute. Let $L_2$-discrepancies
$T_C(\bm{u}_1,\dots,\bm{u}_n)$ and $T^\ast_C(\bm{u}_1,\dots,\bm{u}_n)$ be
defined by
\begin{align*}
  &\phantom{{}={}}T_C(\bm{u}_1,\dots,\bm{u}_n)\\
  &=\!\biggl(\int_{\{(\bm{y},\bm{z})\in[0,1]^{2d};y_i<z_i\}}\!\biggl(\!
  \frac{A([\bm{y},\bm{z});P_n)}{n}-P_C([\bm{y},\bm{z}))
  \!\!\biggr)^2\!\!d\bm{y}d\bm{z}\biggr)^{1/2},
\end{align*}
and
\begin{align*}
  T^\ast_C(\bm{u}_1,\dots,\bm{u}_n)&=\biggl(\int_{[0,1]^{d}}\biggl(
  \frac{A( [\bm{0},\bm{z});P_n)}{n}-C(\bm{z})
  \biggr)^2d\bm{z}\biggr)^{1/2},
\end{align*}
respectively. %We first compute $T^\ast_C$. \TODO[The reader might wonder where
%$T_C$ is addressed here]
Proceeding similarly to~\cite{morokoffcaflisch94},  $T^\ast_C$ can be computed as
\begin{align*}
  &\phantom{{}={}}T^\ast_C(\bm{u}_1,\dots,\bm{u}_n)\\
  &=\frac{1}{n^2}\sum_{k=1}^n\sum_{l=1}^n\prod_{i=1}^d(1-\max(u_{k,i},u_{l,i}))+\int_{[0,1]^d}C(\bm{z})^2 d\bm{z}\\
  &\phantom{{}={}}-\frac{2}{n}\sum_{k=1}^n\int_{u_{k,1}}^1\cdots\int_{u_{k,d}}^1 C(\bm{z})d\bm{z}.
\end{align*}

%As in Example~\ref{ex:P},
If we consider a convex combination $C(u_1,\dots,u_d)=\lambda
 \prod_{i=1}^d u_i+(1-\lambda)\min(u_1,\dots,u_d)$, $\lambda\in(0,1)$, of the
independence copula and the upper Fr\'echet--Hoeffding bound, then one can
compute $T^\ast_C$ explicitly via
 \begin{align*}
&\phantom{{}={}} T^\ast_C(\bm{u}_1,\dots,\bm{u}_n)=\frac{1}{n^2}\sum_{k=1}^n\sum_{l=1}^n\prod_{i=1}^d(1-\max(u_{k,i},u_{l,i}))+\frac{\lambda^2}{3^d}\\
 &+\frac{2(1-\lambda)^2}{(d+1)(d+2)}+\frac{2\lambda(1-\lambda)d!}{\prod_{i=1}^d(2i+1)}
 -\frac{\lambda}{n 2^{d-1}}\sum_{k=1}^n \prod_{i=1}^d (1-u_{k,i}^2)\\
 &-\frac{2(1-\lambda)}{n}\sum_{k=1}^n\left(\sum_{i_1=1}^d\sum_{i_2\neq i_1}\sum_{i_d\neq i_1,\dots,i_{d-1}}\frac{1-u_{k,i_d}^{d+1}}{(d+1)!}\right.\\
&-\left.\sum_{l=1}^{d-1} \sum_{i_1=1}^d\dots\sum_{i_l\neq i_1,\dots,i_{l-1}}\frac{u_{k,i_{l}}^{l+1}(1-u_{k,i_{l+1}})}{(l+1)!}
 \right).
 \end{align*}

\section{Numerical results}
\label{sec:num}

% {\em
% \begin{enumerate}
% \item Options: basket, best-of, worst-of, each done with call and put and two different strike prices (in-the-money and out-of-the money), in dimensions 5, 10, and 20, with Monte Carlo, (scrambled or digitally shifted?) Sobol', and GHalton. Only for Clayton? For each copula, try two different dependence levels.  Note that this means 72 different functions to try for each copula...! We could try to reduce by not doing call and put and/or not all three types of options?
% \item Insurance example: similar to your IS paper. Try Clayton and t-copula. Two different dependence levels for each copula.
% \item Test functions. 1) Try
%   $f(\bm{u}) = (w_1(\alpha_1+1)u_1^{\alpha_1} + \ldots +
%   w_d(\alpha_d+1)u_d^{\alpha_d})$
%   with $\sum_{i=1}^d w_i =1$ and $\alpha_j\ge 0$, $j\in\{1,\dots,d\}$. 2) Try
%   $f(\bm{u}) = K(C(\bm{u}))+1/2$. 3) try .
%   Issue with typical test functions used in QMC studies is that their integral
%   is difficult to compute w.r.t.\ to copula. We could still use them but would have
%   to compute their ``exact'' integral numerically, e.g., with very large Sobol'
%   sequence.

%   Note: We could have also taken $\frac{1}{\binom{d}{2}}\sum_{1\le i<j\le
%     d}(4C(\bm{u}_{ij})-1)/\tau_C$ where $\tau_C$ is the
%   Kendall's tau corresponding to $C$ and $\bm{u}_{ij}$, the two-column matrix
%   consisting of the $i$th and $j$th column of $\bm{u}$, respectively.
% \end{enumerate}
% }

Through typical examples from the realm of finance and insurance and a few
test functions, we now illustrate in this section the efficiency of QRNG in
comparison to standard (P)RNG for copula sampling. More precisely, we compare
Monte Carlo sampling approaches with two types of QRNGs based on randomized low-discrepancy
sequences: The Sobol' sequence and the generalized Halton sequence, both randomized
with a digital shift. Variance/error estimates are obtained by using $B=25$
i.i.d.\ copies of the randomized sequence and comparisons are made with MC
sampling based on the same total number of replications. Each plot includes lines showing $n^{-0.5},n^{-1}$ and/or $n^{1.5}$ convergence rates. In addition, on top of each plot and for each QRNG method, we provide the regression estimate of $\alpha$ such that the variance/error  is in  $O(n^{-\alpha})$.
%\blue{
For PRNG, we only show the results with the CDM sampling algorithm, since the choice of method does not affect the error or variance very much. On the other hand, for QRNG we show the results both with  CDM and MO (when applicable), as this seems to sometimes make a difference. Understanding better why it is so and under what circumstances a sampling algorithm perform better when used in conjunction with QRNG will be a subject of further research.  %}

%\blue{
While the examples given in the next section illustrates the use of our proposed method in typical contexts where they might be used, the test functions results in the section that follows are meant to focus on assessing the performance of QRNG compared to PRNG  on the sole basis of generating copula samples $\bm{U}$ -- without including the effect of the marginal distributions -- and also to see if the sampling algorithm (CDM or MO) has an effect on the performance of QRNG.  %}

%\blue{
Finally, we note that QRNG based on Sobol' point sets is typically slightly faster than PRNG, while the generalized Halton sequence runs slower than PRNG.  %}

\subsection{Examples from the realm of finance and insurance}
Consider a random vector $\bm{X}=(X_1,\dots,X_d)$ modeling $d$ risks in a
portfolio of stocks or insurance losses. We assume that the $j$th marginal
distribution is either log-normal with $X_j\sim\LN(\log(100)+\mu-\sigma^2/2,\sigma^2)$,
$j\in\{1,\dots,d\}$, where $\mu=0.0001$ and $\sigma=0.2$, or Pareto distributed with the same mean and variance as in the log-normal case. The copula $C$ of
$\bm{X}$ throughout this numerical study is either a Clayton or an exchangeable
$t$ copula with three degrees of freedom. To allow a comparable degree
of dependence, we will use the same Kendall's tau for both models. This easily
translates to the parameter $\theta$ of a Clayton copula via the relationship
$\theta=2\tau(1-\tau)^{-1}$ and to the correlation parameter $\rho$ of an
exchangeable $t$ copula via $\rho=\sin(\pi\tau/2)$. We denote
$S=\sum_{j=1}^d X_j$ and consider the estimation of the following functionals
$\Psi(\bm X)$:
\begin{itemize}
\item the Best-Of Call option payoff $(\max X_i - K)^+$;
\item the Basket Call option payoff $(d^{-1}S - K)^+$;
\item the Value-at-Risk at level 0.99 on the aggregated risks
  \begin{align*}
    \VaR_{0.99}\left(S\right)=F_S^{-1}(0.99)=\inf\left\{x\in\IR:F_S(x)\ge 0.99\right\},
  \end{align*}
\item the expected shortfall at level 0.99 on the aggregated risks
  \begin{align*}
    \ES_{0.99}\left(S\right)=\frac{1}{1-0.99}\int_{0.99}^1 F_S^{-1}(u)\text{d}u;
  \end{align*}
\item the contribution of the first and middle margin to $\ES_{0.99}$ of the sum
  under the Euler principle, see \cite{tasche2008},
  \begin{align*}
    \IE[X_1\,|\,S>F^{-1}_S(\alpha)] \text{ and }\IE[X_{d/2}\,|\,S>F^{-1}_S(\alpha)].
  \end{align*}
  These two functionals are referred to as Allocation First and Allocation
  Middle, respectively. %Corresponding results for these functionals are given in Figures \ref{fig:clayton-t:allocations:first}  and \ref{fig:clayton-t:allocations:mid} in the appendix.
\end{itemize}

Figures~\ref{fig:clayton:options:basket}, \ref{fig:clayton:options:best}, \ref{fig:clayton-t:riskmeasures:var}, and~\ref{fig:clayton-t:riskmeasures:es}
(as well as Figures \ref{fig:clayton-t:allocations:mid} and
\ref{fig:clayton-t:allocations:first} in the online supplement) display selected variance estimates for
Clayton and $t$ copulas with Kendall's tau parameter equal to 0.2
and 0.5, using either lognormal or Pareto margins, in dimensions $d=5,10,20$ (displayed in different rows)
and sample sizes $n\in\{10\,000,15\,000,\dots,200\,000\}$. In the Clayton case,
the experiment uses both the MO and CDM sampling methods. For the $t$
copulas, while there is a sampling approach based on a stochastic representation (as seen in Section \ref{subsubsec:ExStochRep}), there is no version of the MO algorithm available, so we only use the CDM method. In addition, both the Sobol' and
generalized Halton QRNGs are used. In all cases, we see that
the variances associated with the Sobol' and generalized Halton quasi-random sequences are smaller
and converge faster than the Monte Carlo variance. It is not clearly determined whether one sampling method is performing
considerably better than the other. But we note that in some cases, such as the estimate of the
Basket Call with $\tau=0.2$ in $d=20$ dimensions
(Figure~\ref{fig:clayton:options:basket}, bottom) the MO sampling seems to perform
better than CDM.
%(supporting the
%idea that the MO algorithm does not deteriorate the performance of the
%quasi-random method).
%, while in others such as the estimate of the expected shortfall
%with $\tau=0.2$ in $d=5$ dimensions (Figure~\ref{fig:clayton-t:riskmeasures:es},
%top) the CDM shows a smaller variance than MO.

\begin{figure}[htbp]
\centering
\includegraphics[width=0.425\textwidth]{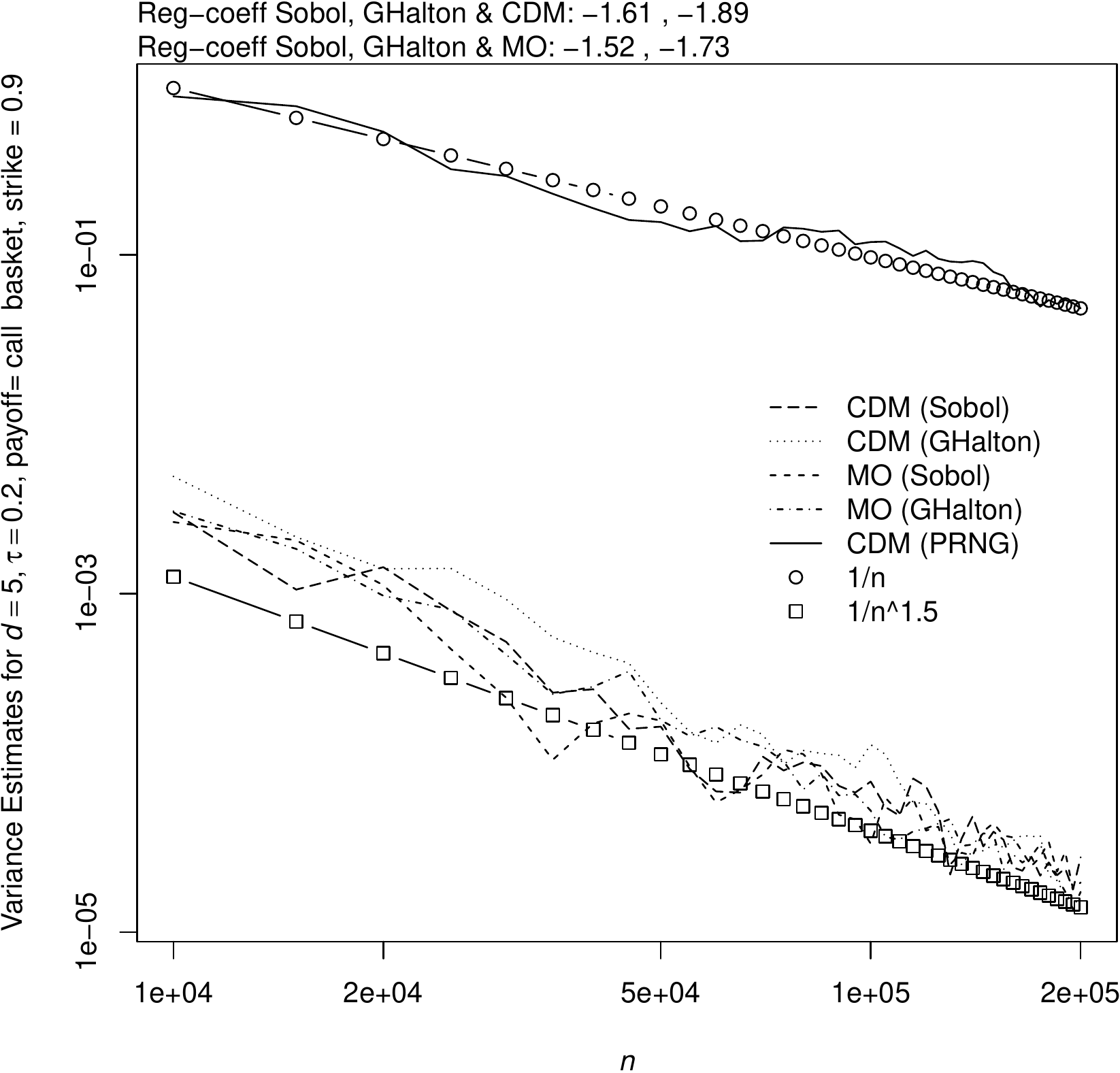}\\[2mm]
\includegraphics[width=0.425\textwidth]{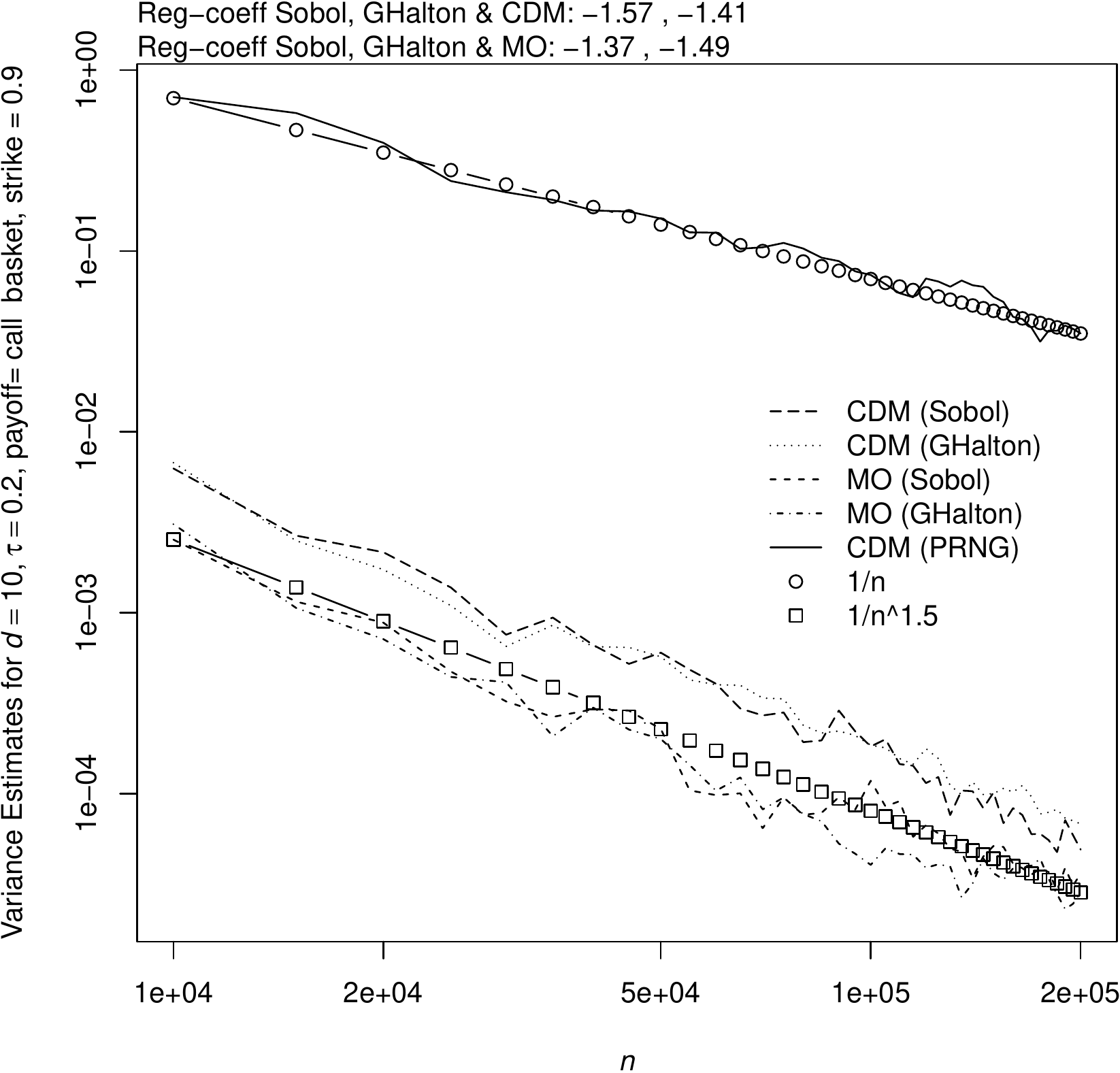}\\[2mm]
\includegraphics[width=0.425\textwidth]{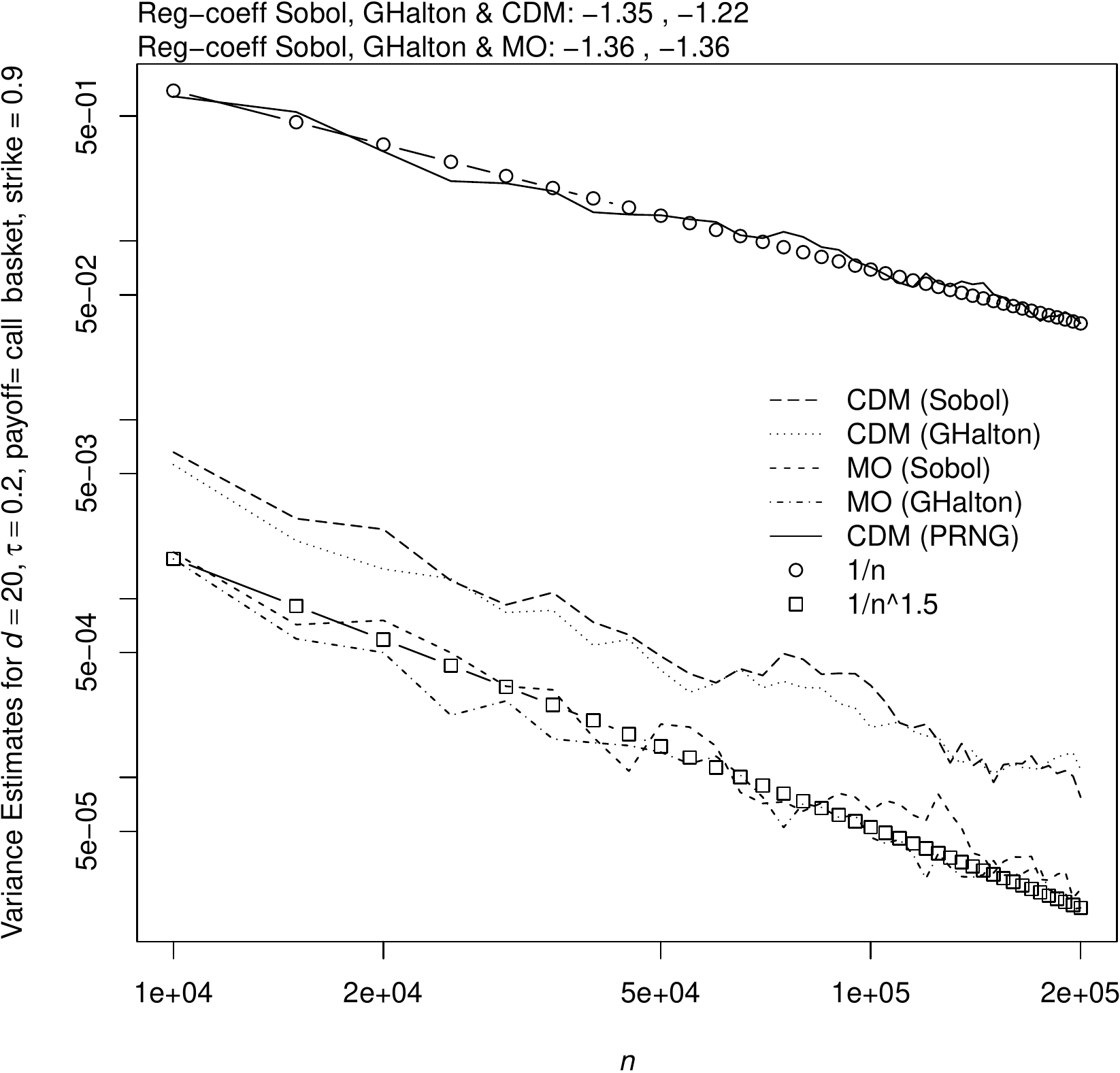}%
\caption{Variance estimates for the functional Basket Call with lognormal margins based on $B = 25$ repetitions for a Clayton copula with parameter such
  that Kendall's tau equals 0.2  for $d = 5$ (top), $d = 10$ (middle) and $d = 20$ (bottom).}
  %Convergence rates for QRNG: $n^{-\alpha}$ with $\alpha \in [1.22,1.89]$}
\label{fig:clayton:options:basket}
\end{figure}

\begin{figure}[htbp]
\centering
\includegraphics[width=0.425\textwidth]{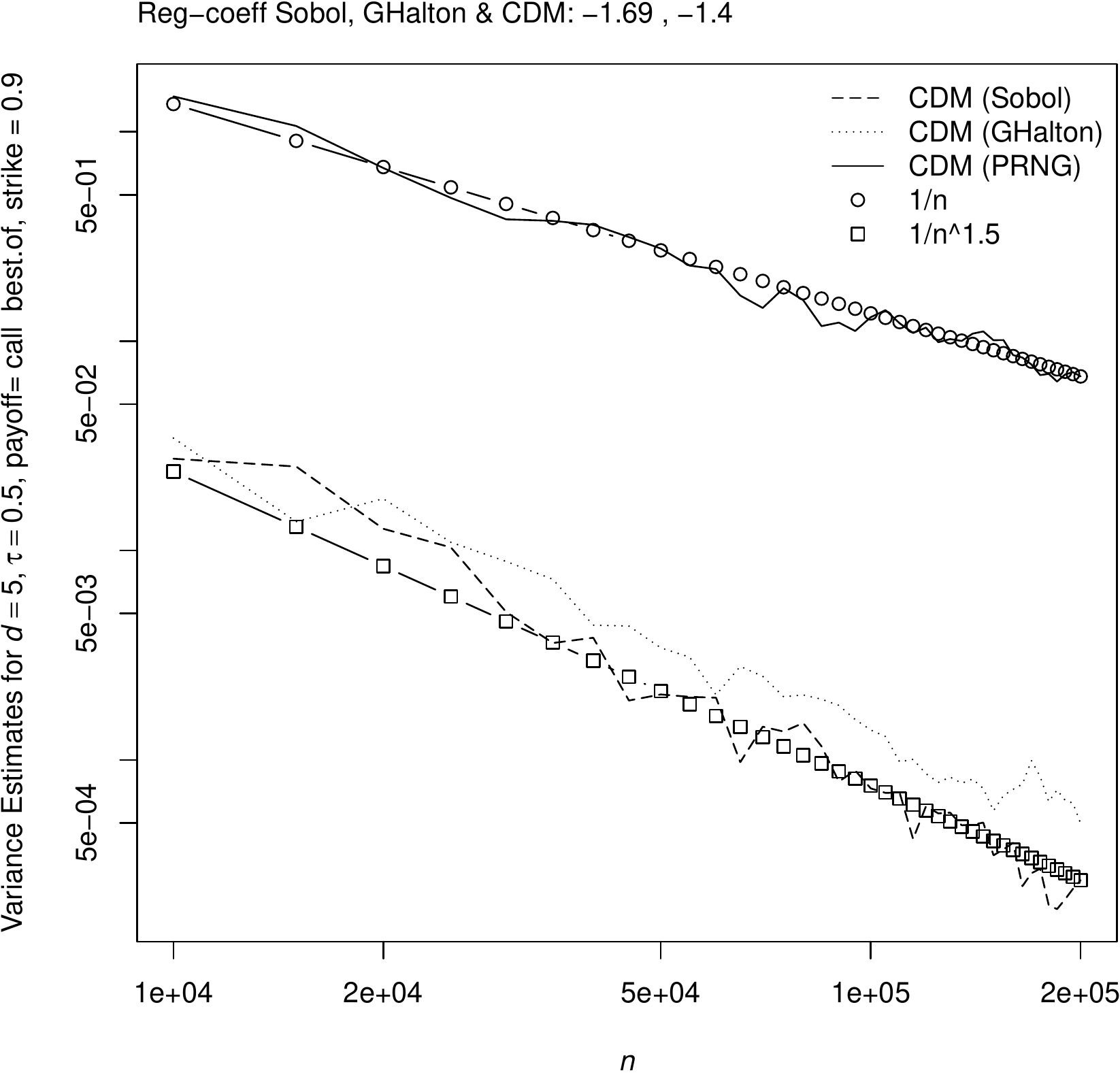}\\[2mm]
\includegraphics[width=0.425\textwidth]{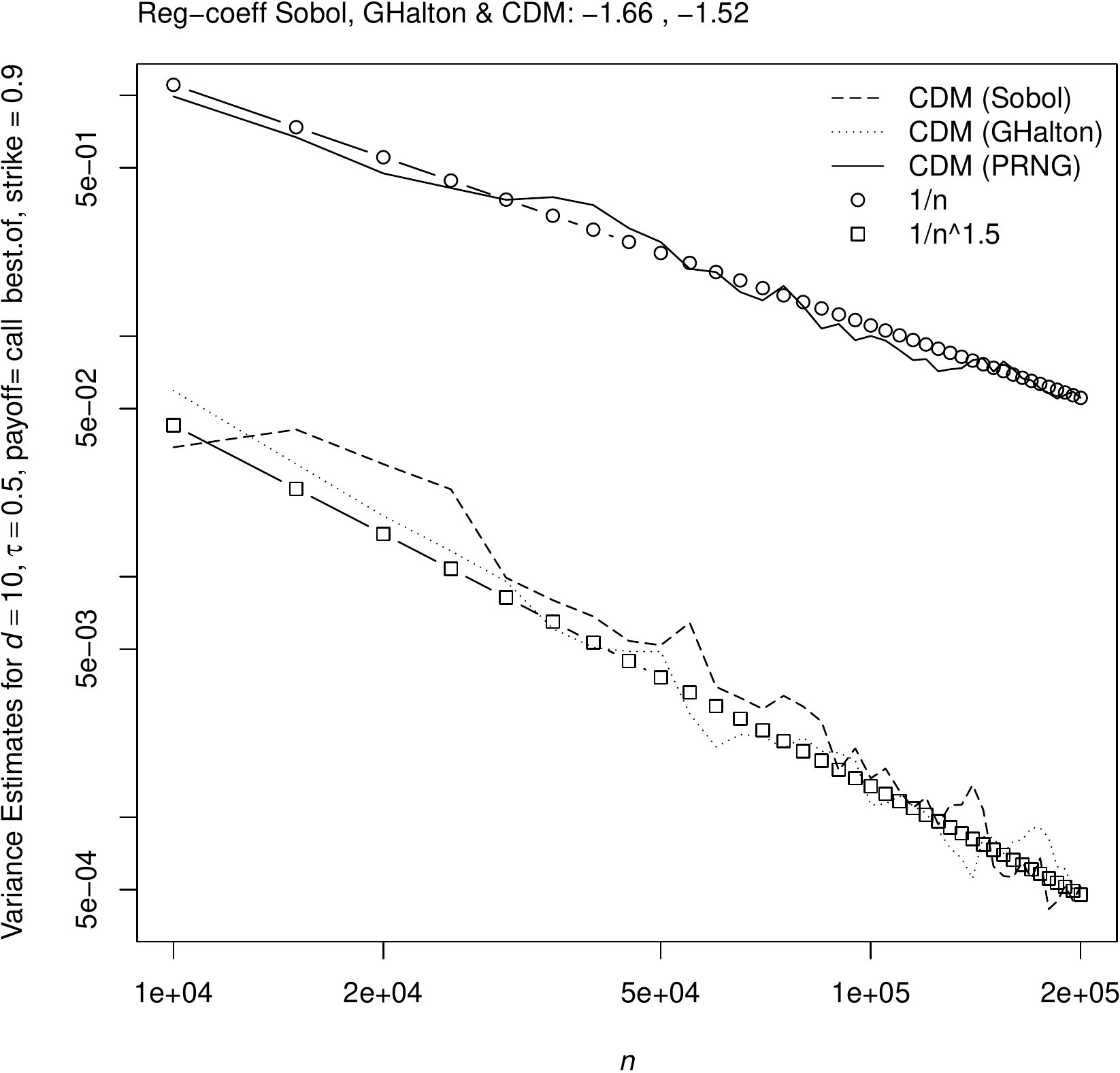}\\[2mm]
\includegraphics[width=0.425\textwidth]{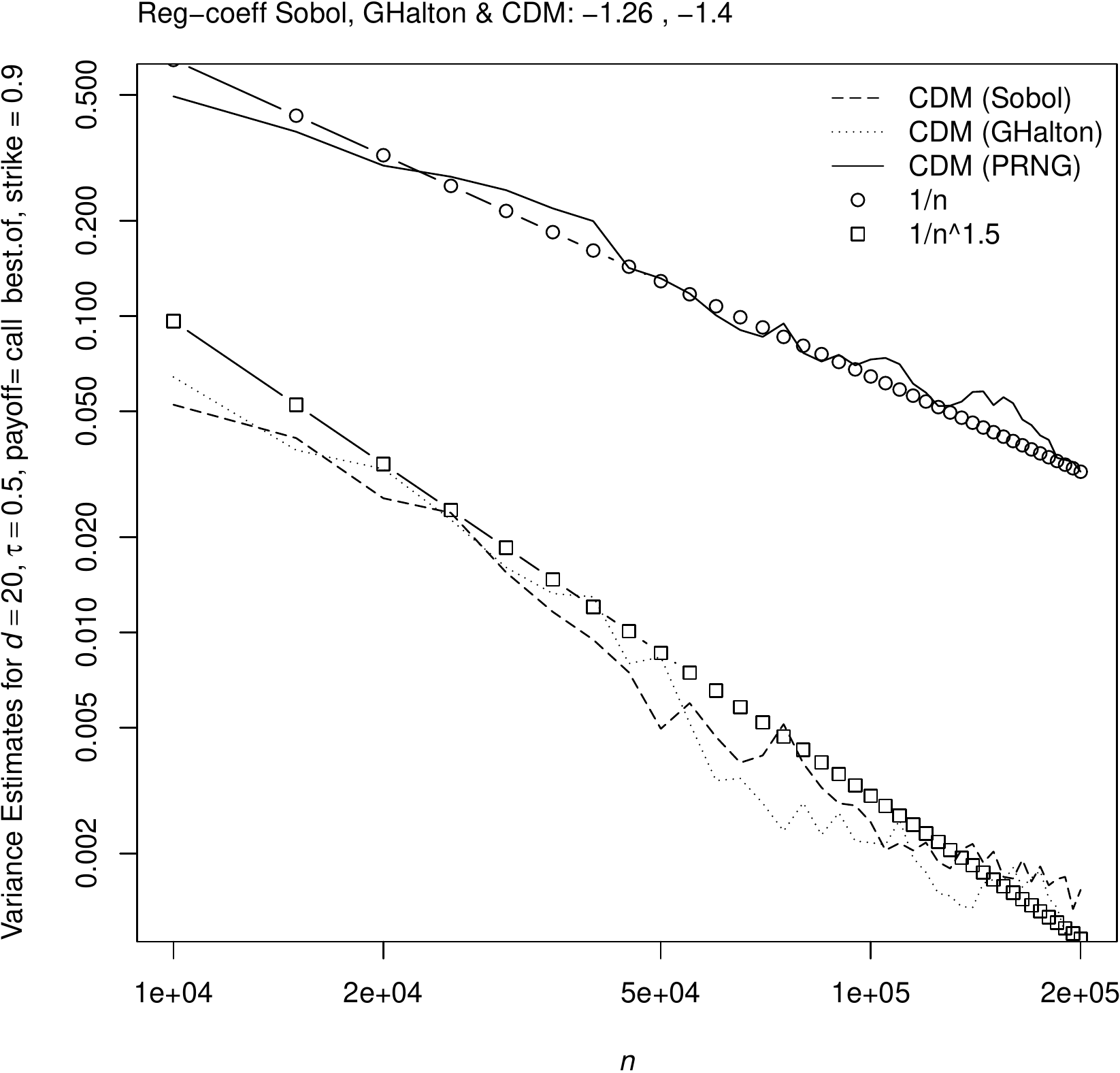}%
\caption{Variance estimates for the functional Best-Of
  Call with Pareto margins based on $B = 25$ repetitions for  an exchangeable $t$  copula with three degrees of freedom such that Kendall's tau
  equals 0.5 for $d = 5$ (top), $d = 10$ (middle) and $d = 20$ (bottom).}
  %Convergence rates for QRNG: $n^{-\alpha}$ with $\alpha \in [1.26,1.69]$}
\label{fig:clayton:options:best}
\end{figure}

\begin{figure}[htbp]
\centering
\includegraphics[width=0.425\textwidth]{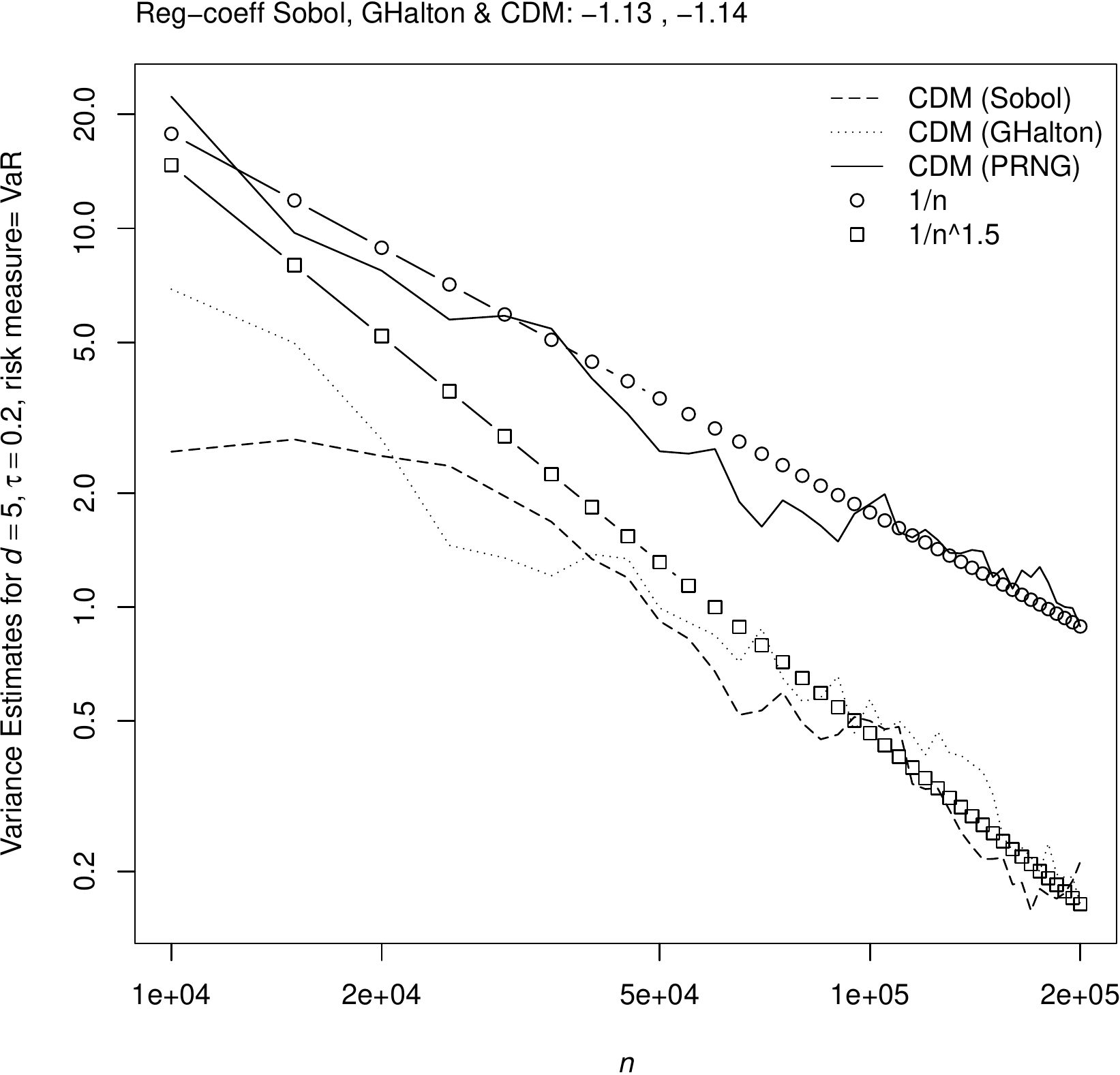}\\[2mm]
\includegraphics[width=0.425\textwidth]{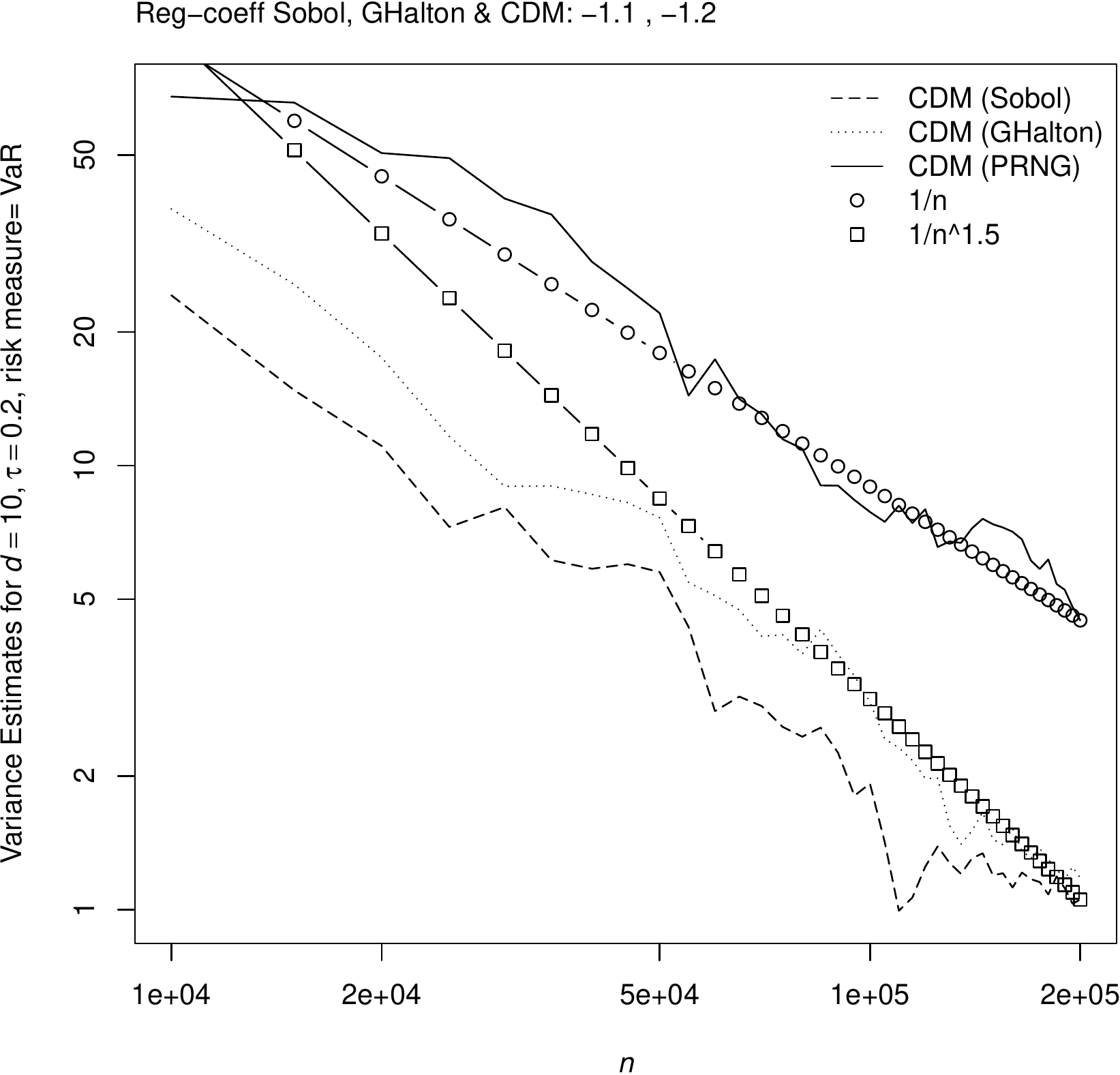}\\[3.2mm]
\includegraphics[width=0.425\textwidth]{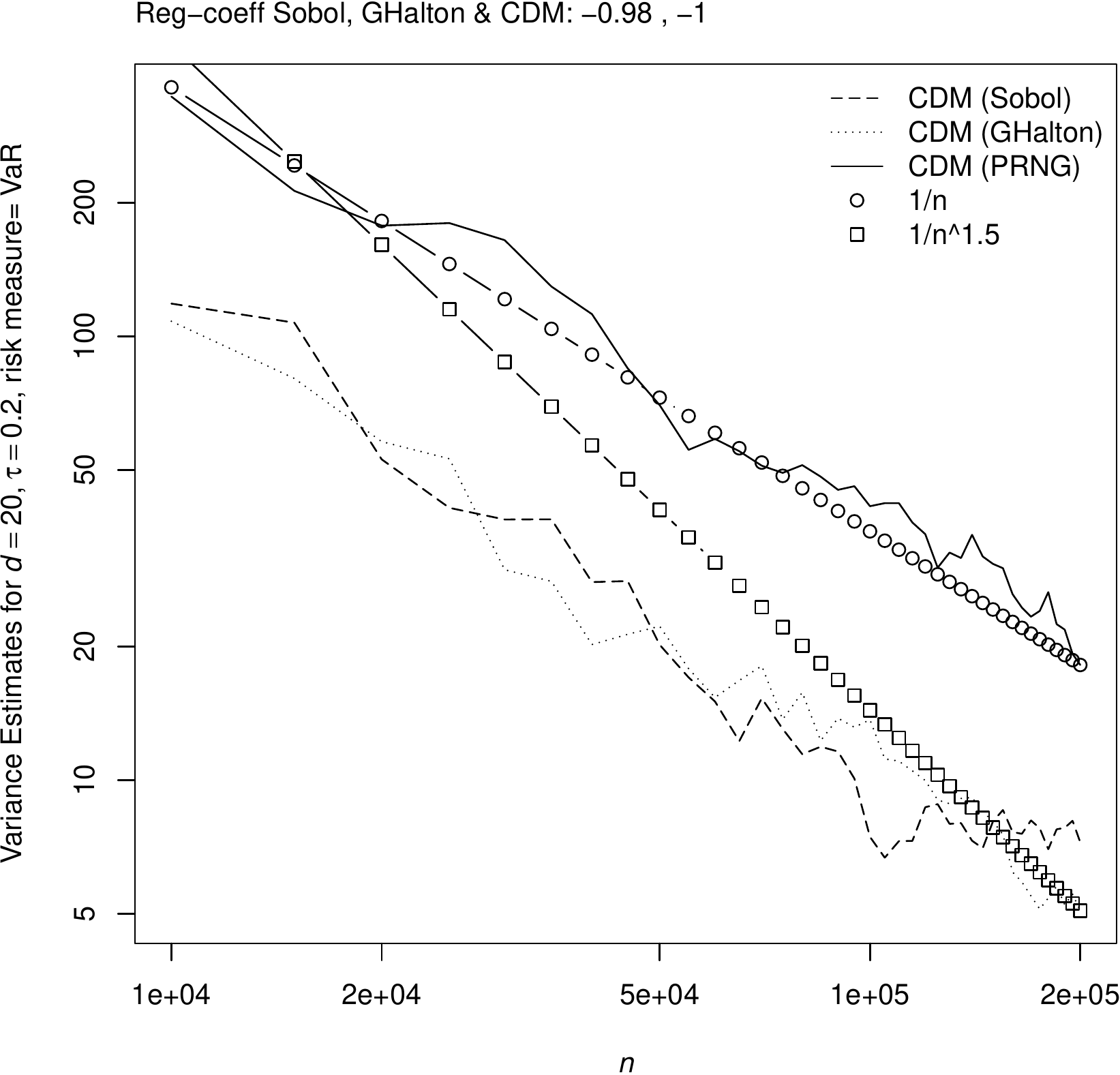}%
\caption{Variance estimates for the functional $\VaR_{0.99}$ with lognormal margins for an
  exchangeable $t$ copula with
  three degrees of freedom such that Kendall's tau equals 0.2 based on $B = 25$ repetitions  for $d=5$ (top), $d=10$ (middle) and
  $d=20$ (bottom).}
  %Convergence rates for QRNG: $n^{-\alpha}$ with $\alpha \in [0.98,1.13]$ for Sobol' and $\alpha \in [1.00,1.20]$ for GHalton}
\label{fig:clayton-t:riskmeasures:var}
\end{figure}

\begin{figure}[htbp]
\centering
\includegraphics[width=0.425\textwidth]{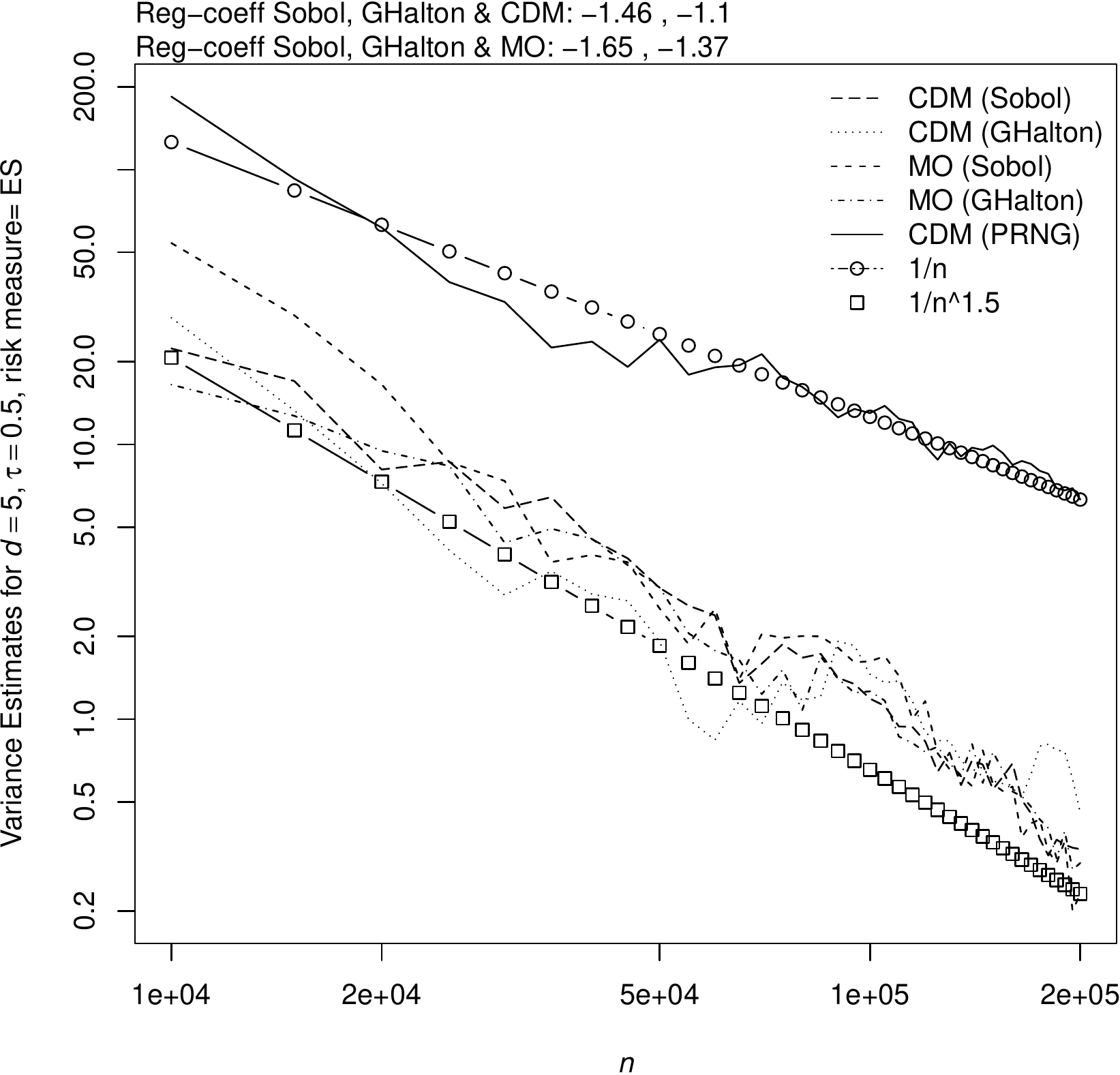}\\[2mm]
\includegraphics[width=0.425\textwidth]{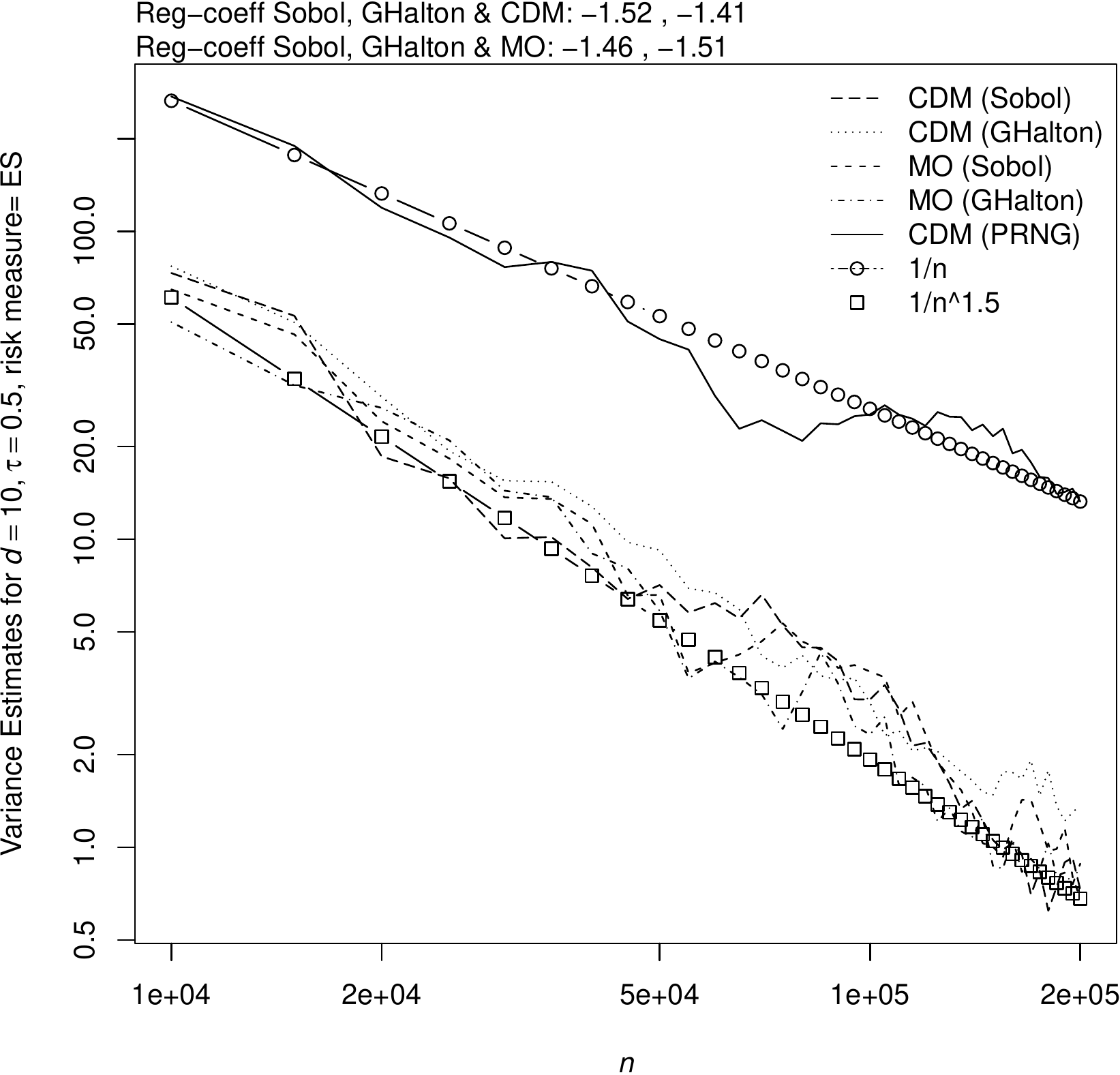}\\[2mm]
\includegraphics[width=0.425\textwidth]{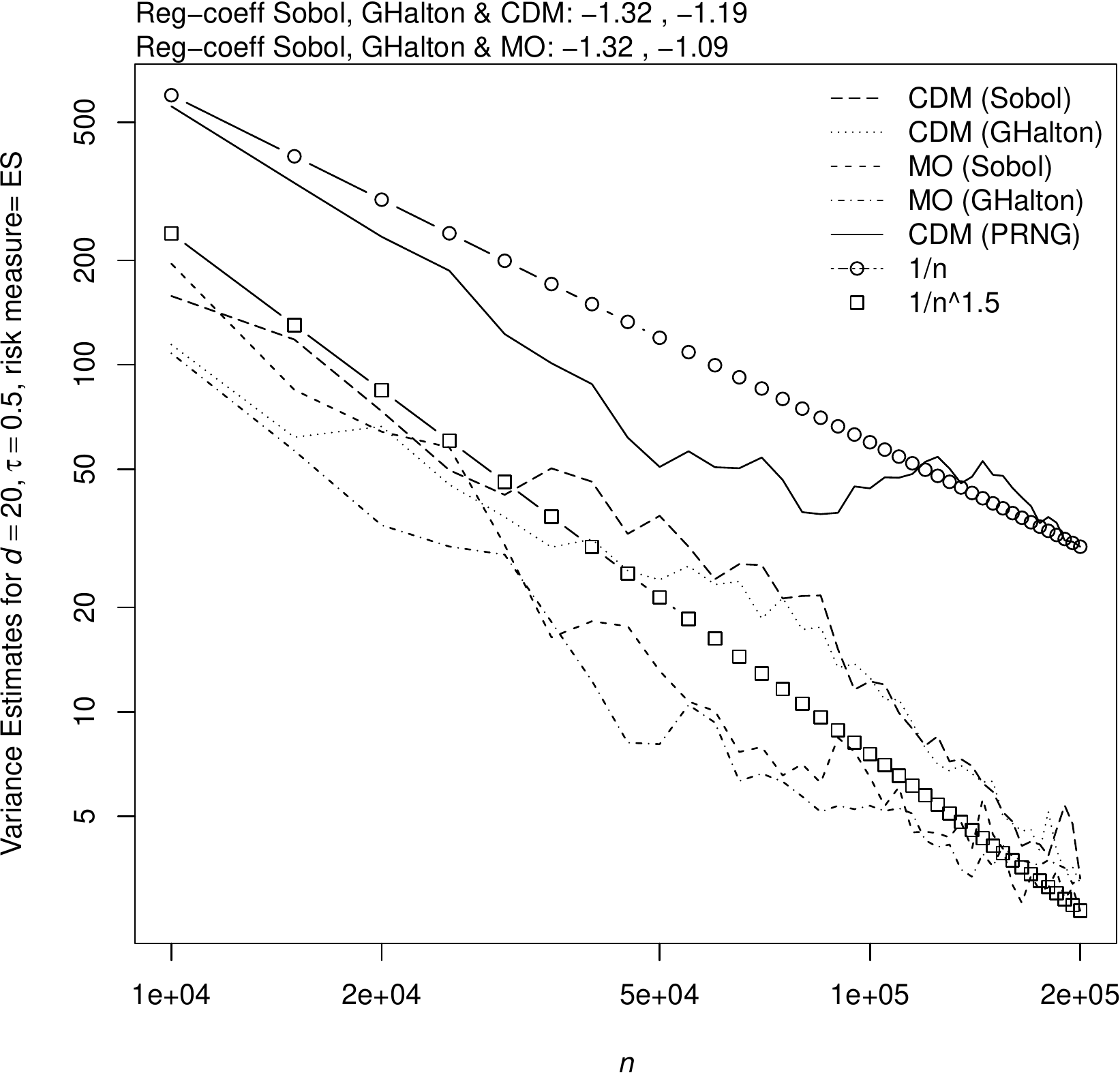}%
\caption{Variance estimates for the functional
  $\ES_{0.99}$ with Pareto margins for a Clayton copula with parameter such that Kendall's tau
  equals 0.5 based on $B = 25$ repetitions for $d=5$ (top), $d=10$ (middle) and
  $d=20$ (bottom).}
  %Convergence rates for QRNG: $n^{-\alpha}$ with $\alpha \in [1.32,1.65]$ for Sobol' and $\alpha \in [1.09,1.51]$ for GHalton}
\label{fig:clayton-t:riskmeasures:es}
\end{figure}

\subsection{Test functions}
We now consider integration results on two different test functions. The results are shown in Figures~\ref{fig:test:tau:0.2:d:5}, \ref{fig:test:tau:0.2:d:15}, \ref{fig:test:tau:0.5:d:5}, and
\ref{fig:test:tau:0.5:d:15} (the latter is in the online supplement), which are based on a Clayton (or $t$) copula with $\tau=0.2$
and $\tau=0.5$, respectively. The first
test function is given by
\begin{align*}
  \Psi_1(\bm{u}) = 3(u_1^2 + \ldots + u_d^2)/d,
\end{align*}
where the vector $\bm{u}$ is obtained after transforming the uniform points
$\bm{v}$ using either the CDM transform or the MO algorithm. Recall that the
former requires $d$-dimensional points (using either a PRNG or a QRNG), whereas
the latter requires $(d+1)$-dimensional points. Note that $\Psi_1$ integrates
exactly to 1 with respect to the copula-induced measure, since $U_j\sim\U[0,1]$,
$j\in\{1,\dots,d\}$. While we know the exact value of the integral in this case,
we still compare estimators based on $B$ i.i.d.\ copies of either MC or RQMC,
but we plot the average absolute error rather than the estimated
variance.

\iffalse
Our second test function is given by
\begin{align*}
  \Psi_2(\bm{u})=K_C(C(\bm{u}))+1/2,
\end{align*}
where $C$ denotes the copula the samples are drawn from and $K_C$ is the corresponding
Kendall distribution function. Here we have again that $\Psi_2$ integrates to
1 w.r.t.\ the copula-induced measure, that is $\IE[\Psi_2(\bm{U})]=1$ for
$\bm{U}\sim C$. %The results are given in Figures~\ref{fig:test:Kendall:0.2} and
%\ref{fig:test:Kendall:0.5} and summarized below.
\fi

The second test function is given by
\begin{align*}
  \Psi_2(\bm{u})=g_1((\phi^{\text{CDM}})^{-1}(\bm{u})),
\end{align*}
where
\begin{align*}
  g_1(\bm{v}) = \prod_{j=1}^d \frac{|4v_j-1|+\alpha_j}{1+\alpha_j},\quad\alpha_j=j,
\end{align*}
which is often used as a test function in the QMC literature; see, e.g.,
\cite{qFAU07a} and the references therein. So here we first apply the inverse of
the CDM transform to the copula-transformed points obtained either using the CDM
approach or MO, and then apply the $d$-dimensional function $g_1$. While this
has the effect of simply applying the standard test function $g_1$ to the
original sample points $\bm{v}_i$ in the case of the CDM, in the case of
the MO algorithm, we are not falling back on the original points $\bm{v}_i$. The
hope is that if MO does not preserve so well the low discrepancy of the original
points, this function would be able to detect this problem.
%The results are given in Figures~\ref{fig:test:Faure:0.2} and \ref{fig:test:Faure:0.5}.

While the second test function is mostly interesting for Archimedean copulas, the first one can be used more generally. This is why in the results reported in Figures \ref{fig:test:tau:0.2:d:5} and \ref{fig:test:tau:0.2:d:15}, we also consider an exchangeable $t$ copula with three degrees of freedom and Kendall's $\tau$ equal to 0.2. % (the online supplement also includes the case $\tau=0.5$ in Figure \ref{fig:test:t:tau:0.5}).

For both test functions, we see that the Sobol' and generalized Halton sequences always clearly outperform Monte Carlo, with an error often converging in $O(n^{-1})$ rather than the $O(n^{-0.5})$ corresponding to Monte Carlo. %, with often an advantage for Sobol' compared to generalized Halton.
%We also note that for the second function $\Psi_2$, when $d=15$, the MO method
%seems to do better than the CDM method, while for $d=5$, the two methods perform
%similarly. This advantage of MO in higher dimensions is not as clearly seen for
For the first function $\Psi_1$, both the CDM and MO methods perform about the same. We believe this might be due to the simplicity of
$\Psi_1$---a sum of univariate powers of the $u_j$'s---and the fact that both
methods perform equally well in the univariate case when combined with RQMC.
%This could be because $\Psi_1$ is a sum of one-dimensional functions w.r.t.\ the copula sample coordinates $u_j$'s, and both CDM and MO %We believe this might be due to the fact that MO requires a $(d+1)$-dimensional point set  and  the effect of working with an underlying  $(d+1)$-dimensional point set rather than $d$-dimensional is more pronounced when $d=5$ than when $d=10$. The same observation does not hold for the Sobol' sequence, as its error behaves a bit more erratically and an ordering between CDM and MO is harder to establish.
Looking at the results for $\Psi_2$, we see that with RQMC the CDM method performs better than MO, as there is no copula transformation performed in the case of CDM. However, RQMC with MO is still better than Monte Carlo, which suggests that the MO algorithm %does not really  ``destroy''
manages to preserve the low discrepancy of the original point set.

\begin{figure}[htbp]
\centering
\includegraphics[width=0.425\textwidth]{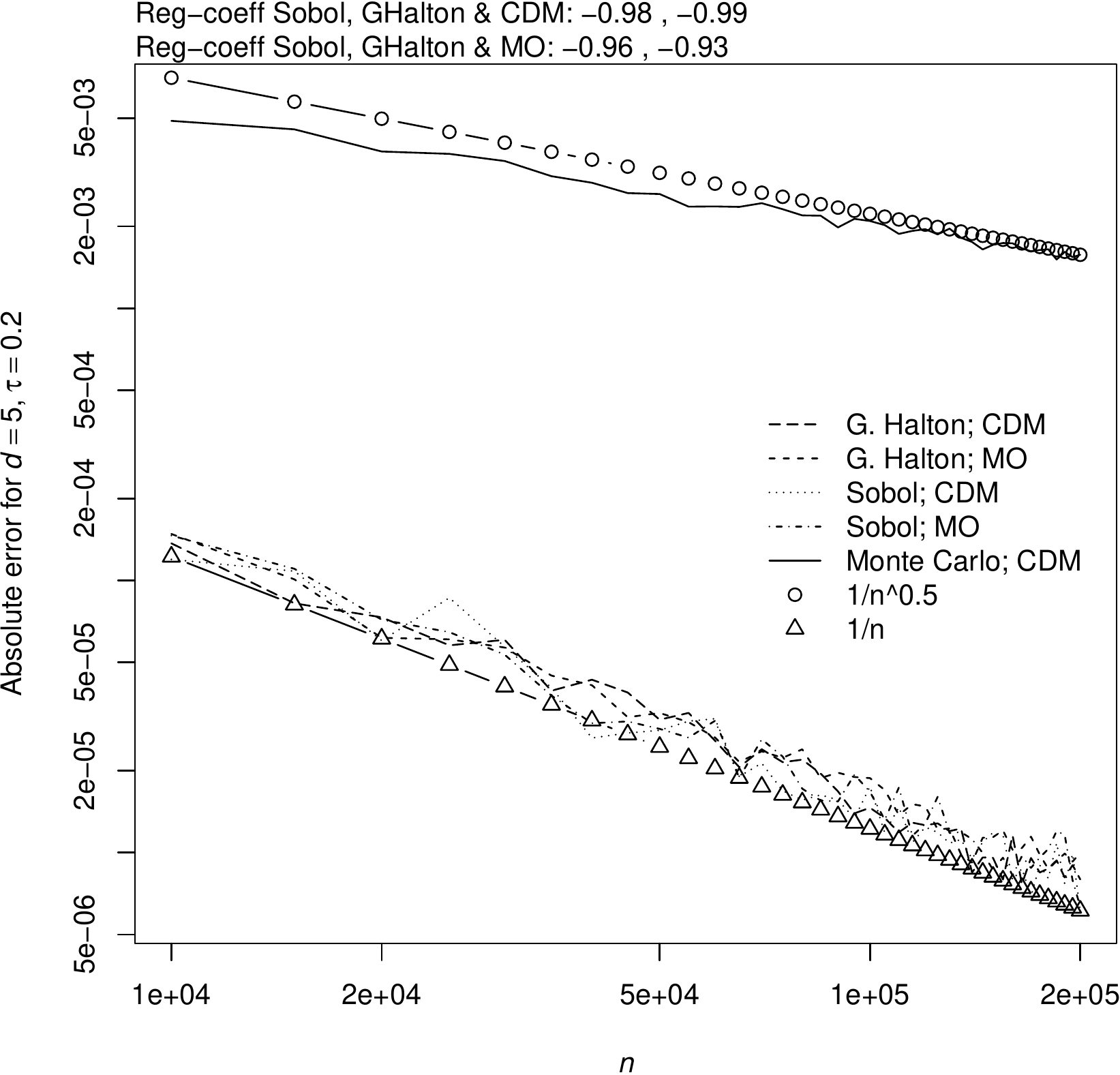}\\[2mm]
\includegraphics[width=0.425\textwidth]{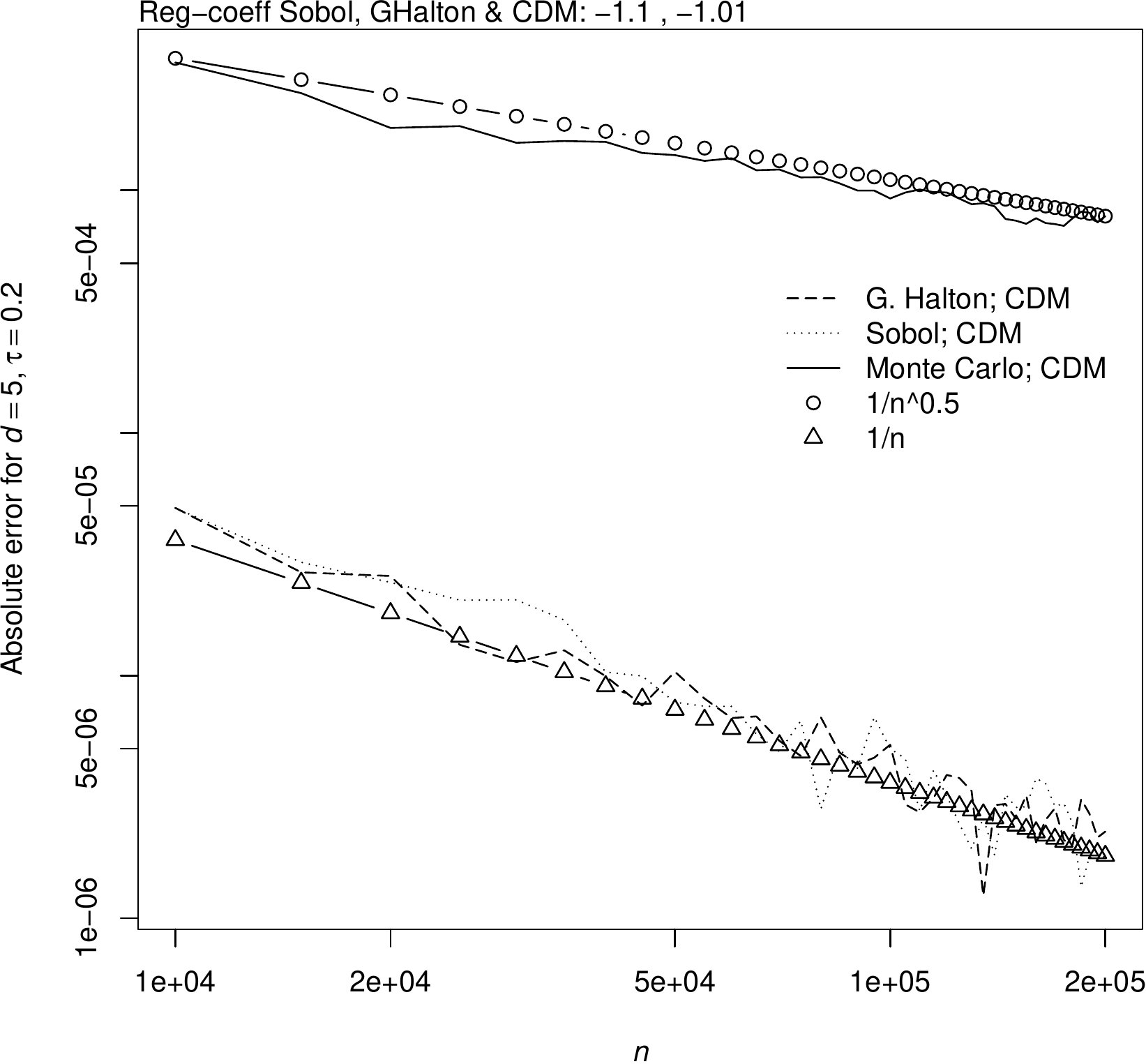}\\[-0.8mm]
\includegraphics[width=0.425\textwidth]{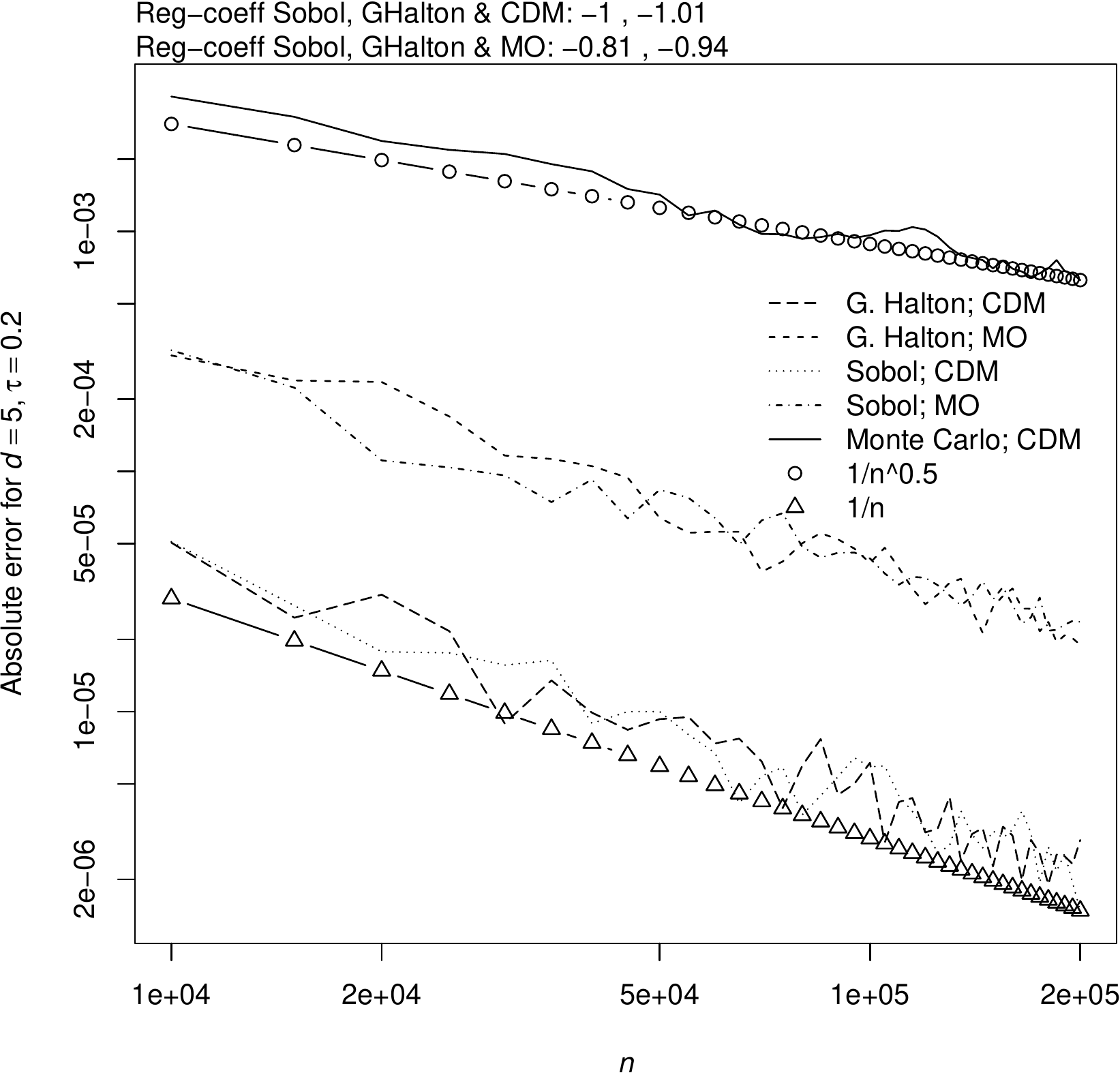}%
\caption{Average absolute errors for the test functions
    $\Psi_1(\bm{u})=3(u_1^2 + \ldots + u_d^2)/d$ (top) and $\Psi_2(\bm{u})=g_1((\phi^{\text{CDM}})^{-1}(\bm{u}))$ (bottom) for a Clayton copula with parameter such that Kendall's tau
  equals 0.2 based on $B=25$ repetitions  for $d=5$; the middle plot shows results for $\Psi_1(\bm{u})$ and an exchangeable $t$ copula with 3 degrees of freedom and Kendall's tau
  of 0.2.}
  %Error convergence rates for QRNG: $n^{-\alpha}$ with $\alpha \in [0.81,1]$ }
\label{fig:test:tau:0.2:d:5}
\end{figure}

\begin{figure}[htbp]
\centering
\includegraphics[width=0.425\textwidth]{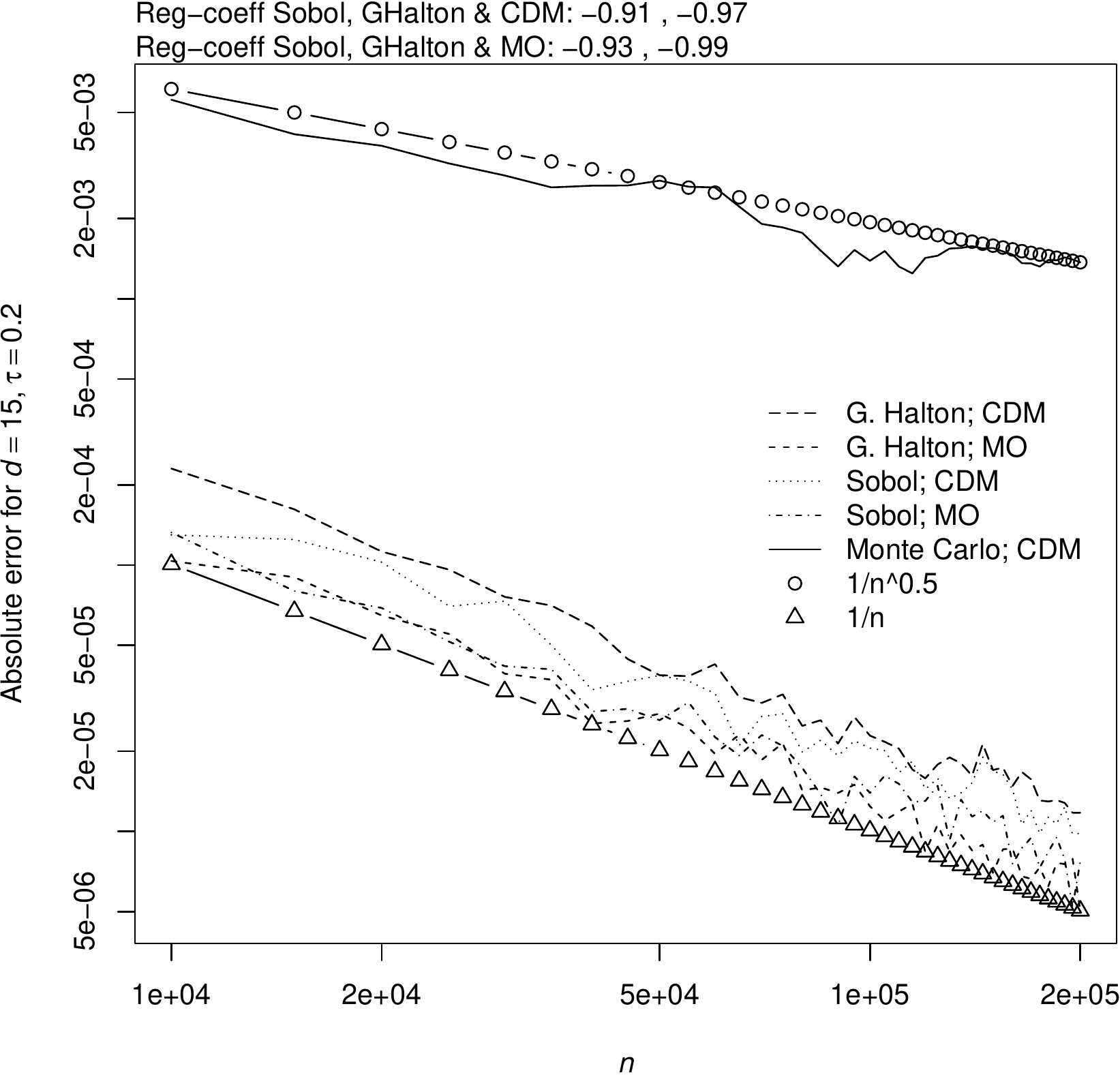}\\[2mm]
\includegraphics[width=0.425\textwidth]{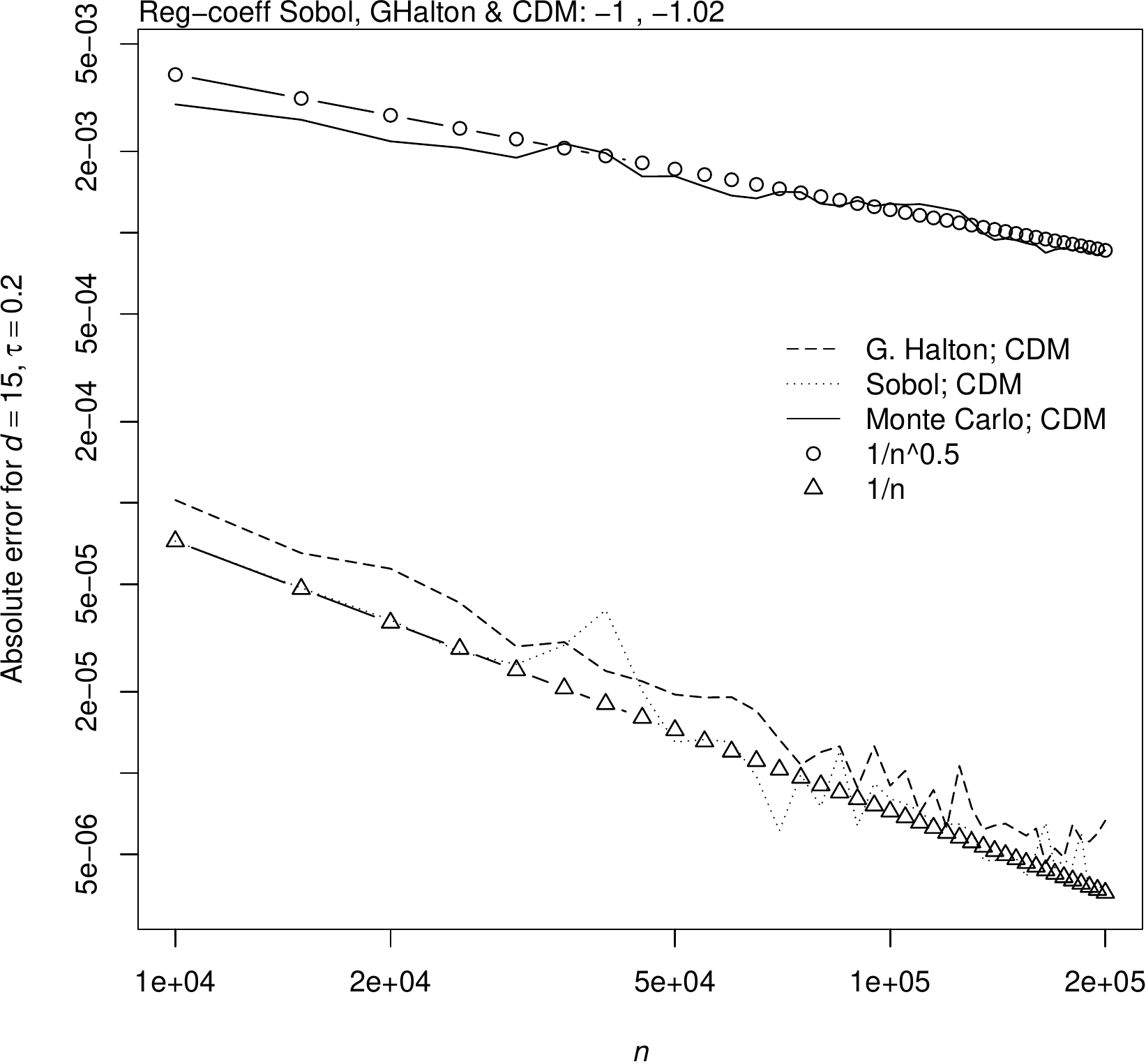}\\[-0.8mm]
\includegraphics[width=0.425\textwidth]{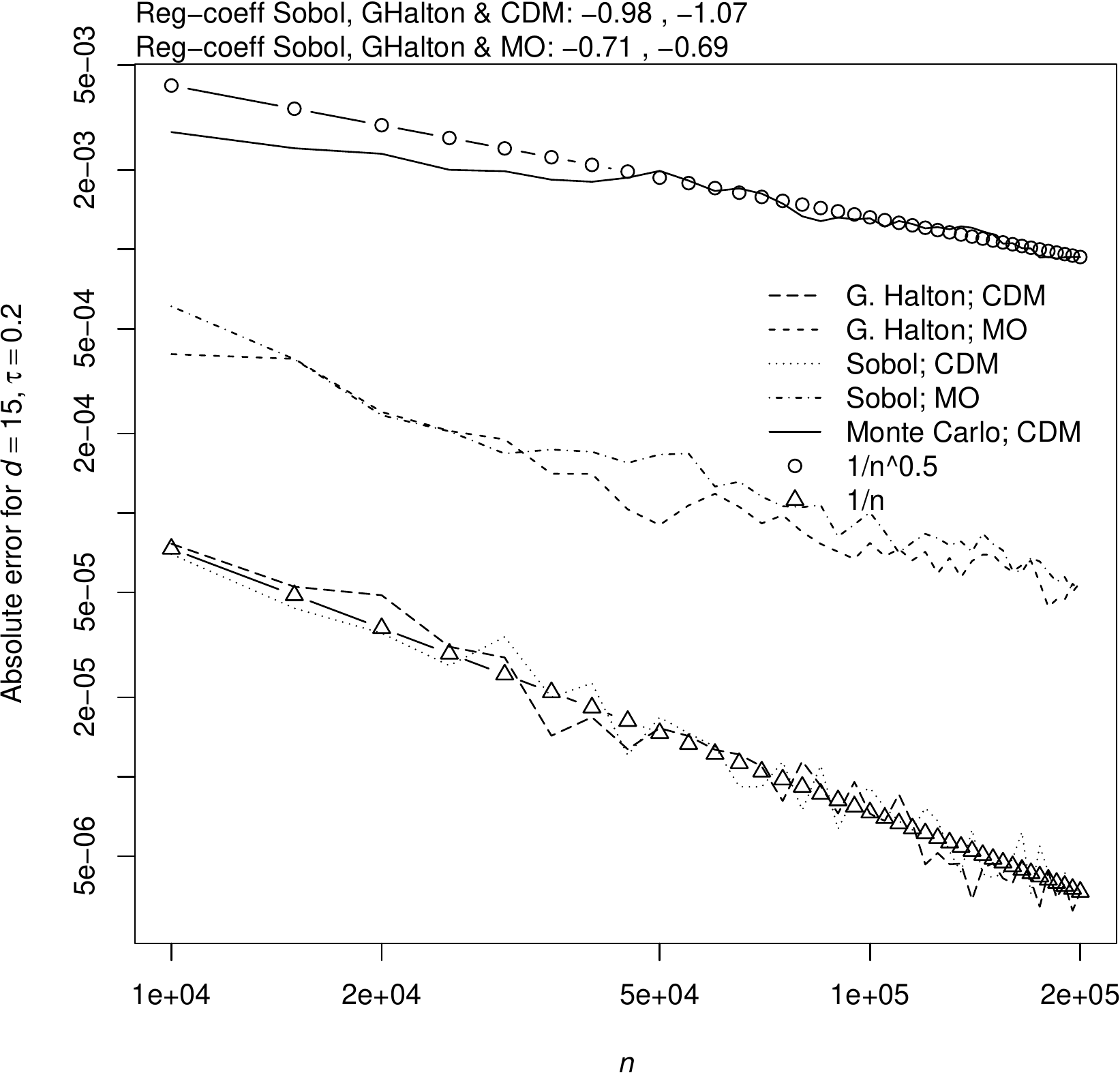}%
\caption{Average absolute errors for the test functions
    $\Psi_1(\bm{u})=3(u_1^2 + \ldots + u_d^2)/d$ (top), $\Psi_2(\bm{u}) = K_C(C(\bm{u}))+1/2$ (middle), and $\Psi_3(\bm{u})=g_1((\phi^{\text{CDM}})^{-1}(\bm{u}))$ (bottom) for a Clayton copula with parameter such that Kendall's tau
  equals 0.2 based on $B=25$ repetitions for $d=15$; the middle plot shows results for $\Psi_1(\bm{u})$ and an exchangeable $t$ copula with 3 degrees of freedom and Kendall's tau
  of 0.2.}
  %Error convergence rates for QRNG: $n^{-\alpha}$ with $\alpha \in ?$}
\label{fig:test:tau:0.2:d:15}
\end{figure}

\section{Conclusion and discussion}\label{sec:conclusion}
The main goal of this paper was to show how copula samples can be generated
using quasi-random numbers. This is of interest when replacing PRNGs by QRNGs in applications  involving dependent samples, possibly in higher dimensions. We
have seen that different sampling approaches can be used, with a focus on the
CDM approach and, additionally for Archimedean copulas, on the Marshall--Olkin algorithm. We have
studied the error behaviour for both methods and have seen that in order to prove that the
composed function $\Psi \circ \phi_C$ is smooth enough to satisfy the
Koksma--Hlawka bound for the error, sufficient conditions on the function $\Psi$ are that it must have
smooth higher order mixed partial derivatives. For the Marshall--Olkin algorithm,
we have shown that for several Archimedean copula families, the corresponding
transformation $\phi_C$ is smooth enough to guarantee the good behaviour of
the error bound. The superiority of QRNG over PRNG for copula sampling was illustrated on several examples, including a simulation addressing an application in the realm of finance and insurance. Most of the results in this paper are reproducible using the \R\
packages {\tt copula} and {\tt qrng}.

Some ideas for future work would be to follow-up on
Proposition~\ref{prop:VarFourComp} and to analyze the speed of decay of the
Walsh coefficients of the composed function $\Psi \circ \phi_C$, based on
assumptions on the speed of decay of the Walsh coefficients of $\Psi$ and the
properties of the sampling method $\phi_C$.
%Another possibility for future work on this
%topic would be to try to prove Lipschitz conditions on $\Psi \circ \phi_C$ such
%as those used in \cite{Owen97b} to analyze the variance of RQMC estimators.

Concerning the copula-induced discrepancy studied in
Section~\ref{subsec:TransfSamples}, a possible avenue for future research would
be to construct point sets that directly minimize this discrepancy, without
first transforming a (uniform-based) low-discrepancy sample. In addition,
proving error bounds based on the $L_2$-discrepancy would be useful, as this
discrepancy measure is easier to compute. Finally, numerically computing the
copula-induced discrepancy for samples transformed either using the CDM or the MO
algorithm would give us some insight as to how conservative the bound given in
Proposition~\ref{prop:bounded_transform_discrep} is. \nocite{mariusphdthesis}

%\printbibliography[heading=bibintoc]
%\clearscrheadfoot

\begin{acknowledgements}
We thank the Associate Editor and the two anonymous reviewers for their helpful comments, which helped us improve this paper. The first author wishes to thank SCOR for their financial support. The second and third authors acknowledge the support of NSERC through grants \#5010 and \#238959, respectively.   
\end{acknowledgements}

%\bibliographystyle{spbasic}  
%\bibliography{qrng} 

\appendix
\section*{Appendix}

%\subsection*{Proofs}

\iffalse
We first provide a review of integration lattices, focusing on the special case of {\em rank-1
  lattices}, which are the most commonly used in practice. They are obtained by
choosing a generating vector $\bm{z} = (z_1,\ldots,z_d)$ where
$z_j \in \{1,\ldots,n-1\}$, $j\in\{1,\dots,d\}$, and then the lattice $P_n$ is
given by
\begin{align*}
  P_n = \left\{ \frac{i-1}{n} \bm{z} \bmod{1},\ i=1,\ldots,n \right\},
\end{align*}
where $\bmod{\,1}$ denotes the coordinate-wise modulo operation. The $z_j$ are
typically chosen so that $\gcd(z_j,n)=1$, which implies that each
one-dimensional projection of the lattice corresponds to the $n$-point grid
$\{0,1/n,\ldots,(n-1)/n\}$ (if $\gcd(z_j,n)>1$ then the corresponding projection
is mapped to $n/\gcd(z_j,n) < n$ distinct points). A {\em Korobov lattice} is a
special form of a rank-1 lattice in which $z_j =a^{j-1} \bmod{n}$ for a so-called
generator $a \in \{1,\ldots,n-1\}$.
\fi

For all the randomization schemes mentioned in Section \ref{sec:qrng}, in addition to having a
simple method to estimate the variance of the corresponding RQMC estimator, %(based on one instance of a randomized point set),
results giving exact expressions for this variance are also known and typically
rely on using a well-chosen series expansion of the function $\Psi$ of
interest. The following result recalls this variance expression in the case of
randomly digitally shifted net; see \cite{Lemieux2009} for a detailed
proof. This result  is used in the proof of Proposition \ref{prop:VarFourComp} in Section \ref{subsec:composition}.

\begin{theorem}[Variance for randomly digitally shifted nets]\label{thm:VarFourierWalsh}
 Let $\tilde{P}_n = \{\tilde{\bm{v}}_1,\ldots,\tilde{\bm{v}}_n\}$ be a
 randomly digitally shifted net in base $b$
 with corresponding RQMC estimator $\widehat{\mu}_n$ given by
 \begin{align*}
 \widehat{\mu}_n = \frac{1}{n} \sum_{i=1}^n \Psi(\tilde{\bm{v}}_i)
 \end{align*}
 and assume
 ${\rm Var}(\Psi(\bm{U})) < \infty$ for $\bm{U} \sim U[0,1)^d$. Then we have that
 \begin{align*}
   {\rm Var}(\widehat{\mu}_n) = \sum_{\bm{0} \neq \bm{h} \in {\cal L}_d^*} |\widehat{\Psi}(\bm{h})|^2,
 \end{align*}
 where
  $\widehat{\Psi}(\bm{h})$ is the Walsh
 coefficient of $\Psi$ at $\bm{h}$, given by
 \[
 \widehat{\Psi}(\bm{h})  =%= \begin{cases}
% \int_{[0,1)^d} f(\bm{u}) e^{-2\pi i \bm{h} \cdot \bm{u}}d\bm{u} & \mbox{ (Fourier)}\\
  \int_{[0,1)^d} f(\bm{u}) e^{-2\pi i  \langle \bm{h}, \bm{u} \rangle_b}d\bm{u} %& \mbox{ (Walsh)}
%  \end{cases}
\]
where $\langle \bm{h}, \bm{u} \rangle_b = \frac{1}{b} \sum_{j=1}^d \sum_{l=0}^{\infty} h_{j,l}u_{j,l+1}$ with $h_{j,l}$ and $u_{j,l}$ obtained from the base $b$ expansion of $h_j$ and $u_j$, respectively, and ${\cal L}_d^*=\{\bm{h} \in \mathbb{Z}^d: \langle \bm{h}, \bm{v}_i \rangle_b \in \mathbb{Z}, \forall i=1,\ldots,n \}$
 is the dual net of the deterministic net
 $P_n = \{\bm{v}_i,i=1,\ldots,n\}$ that has been shifted to get $\tilde{P}_n$.
\end{theorem}

\subsection*{Proofs}
\label{sec:proof}
\iffalse
\begin{proof}[Proof of Corollary~\ref{cor:cond_cop}]
  By Sklar's Theorem,
  $h(x_1,\dots,x_j)=c(F_1(x_1),\dots,F_j(x_j))\prod_{k=1}^jf_k(x_k)$. Therefore,
  \begin{align*}
    H(x_j\,|\,x_1,\dots,x_{j-1})&=\int_{-\infty}^{x_j}h(z_j\,|\,x_1,\dots,x_{j-1})\,dz_j\\
    &=\frac{\int_{-\infty}^{x_j}h(x_1,\dots,x_{j-1},z_j)\,dz_j}{h(x_1,\dots,x_{j-1})}\\
    &=\frac{\int_{-\infty}^{x_j}c(F_1(x_1),\dots,F_{j-1}(x_{j-1}),F_j(z_j))f_j(z_j)\,dz_j}{c(F_1(x_1),\dots,F_{j-1}(x_{j-1}))}\\
    &=\frac{\int_0^{F_j(x_j)}c(F_1(x_1),\dots,F_{j-1}(x_{j-1}),y_j)\,dy_j}{c(F_1(x_1),\dots,F_{j-1}(x_{j-1}))}.
  \end{align*}
  Letting $x_k=F_k^-(u_k)$, $k\in\{1,\dots,j\}$, the result follows from \eqref{eq:CDM:dens}.
\end{proof}
\fi

\begin{proof}[Proof of Proposition~\ref{prop:cond_t_copula}]
  % be *extremely* careful with the notation here... P^{-1}_{1:(j-1),1:(j-1)}
  % is the ``upper-left'' entry of P^{-1} but in general not the same as (P_{1:(j-1),1:(j-1)})^{-1}
  Assume $P=\Bigl(\begin{smallmatrix}
    P_{1:(j-1),1:(j-1)} & P_{1:(j-1),j} \\ P_{j,1:(j-1)} & P_{j,j}
  \end{smallmatrix}\Bigr)$ and
  $P^{-1}=\Bigl(\begin{smallmatrix}
    P^{-1}_{1:(j-1),1:(j-1)} & P^{-1}_{1:(j-1),j} \\ P^{-1}_{j,1:(j-1)} & P^{-1}_{j,j}
  \end{smallmatrix}\Bigr)$ to be positive definite
  matrices. Corollary~\ref{cor:cond_cop} implies that
  \begin{align*}
    C(u_j\,|\,u_1,\dots,u_{j-1})=\frac{\int_{-\infty}^{x_j}h_{\nu, P}(x_1,\dots,x_{j-1},z_j)\,dz_j}{h_{\nu, P_{1:(j-1),1:(j-1)}}(x_1,\dots,x_{j-1})},
  \end{align*}
  where
  \begin{align}\label{eq:tmult_dens}
    h_{\nu, P}(x_1,\dots,x_j)=\frac{\Gamma\bigl(\frac{\nu+j}{2}\bigr)}{\Gamma\bigl(\frac{\nu}{2}\bigr)(\nu\pi)^{j/2}\sqrt{|P|}}\Bigl(1+\frac{\bm{x}\T P^{-1} \bm{x}}{\nu}\Bigr)^{-\frac{\nu+j}{2}}
  \end{align}
  is the density function of $t_{\nu,P}$. Using the block matrix equality
  \begin{align*}
    P^{-1}_{1:(j-1),1:(j-1)}-P^{-1}_{1:(j-1),j}(P^{-1}_{j,j})^{-1}P^{-1}_{j,1:(j-1)}=\bigl(P_{1:(j-1),1:(j-1)}\bigr)^{-1},
  \end{align*}
  we have that
  \begin{align*}
    &\phantom{{}={}}\bm{x}\T P^{-1}\bm{x}\\
    &=\bm{x}_{1:(j-1)}\T
    P^{-1}_{1:(j-1),1:(j-1)}\bm{x}_{1:(j-1)}+x_j^2 P^{-1}_{j,j}+2x_j\bm{x}_{1:(j-1)}\T
    P^{-1}_{1:(j-1),j}\\
    &=\bm{x}_{1:(j-1)}\T\bigl(P_{1:(j-1),1:(j-1)}\bigr)^{-1}\bm{x}_{1:(j-1)}+x_j^2 P^{-1}_{j,j}+2x_j\bm{x}_{1:(j-1)}\T
    P^{-1}_{1:(j-1),j}\\
    &\phantom{={}}+\bm{x}_{1:(j-1)}\T P^{-1}_{1:(j-1),j}(P^{-1}_{j,j})^{-1}P^{-1}_{j,1:(j-1)}\bm{x}_{1:(j-1)}\\
    &=g+k(x_j)^2,
  \end{align*}
  where
  \begin{align*}
    g&=\bm{x}_{1:(j-1)}\T\bigl(P_{1:(j-1),1:(j-1)}\bigr)^{-1}\bm{x}_{1:(j-1)},\\
    k(x_j)^2&=\Bigl(x_j\sqrt{P^{-1}_{j,j}}+s_2\Bigr)^2,\\
    s_2&=\bm{x}_{1:(j-1)}^{\top}P^{-1}_{1:(j-1),j}/\sqrt{P^{-1}_{j,j}}.
  \end{align*}
  We can thus rewrite \eqref{eq:tmult_dens} as
  \begin{align*}
    h_{\nu, P}(x_1,\dots,x_j)=a\Bigl(1+\frac{g}{\nu}\Bigr)^{-\frac{\nu+j}{2}} h_{\nu+j-1}(l(x_j)),
  \end{align*}
  where $h_{\nu+j-1}$ is the density of $t_{\nu+j-1}$ and
  \begin{align*}
    a=\frac{\Gamma\bigl(\frac{\nu+j-1}{2}\bigr)\sqrt{(\nu+j-1)\pi}}{\Gamma\bigl(\frac{\nu}{2}\bigr)(\nu\pi)^{j/2}\sqrt{|P|}},\ l(x_j)=k(x_j)s_1,\ s_1=\sqrt{\frac{\nu+j-1}{\nu+g}}.
  \end{align*}
  Using a change of variable argument, we compute
  \begin{align*}
    \int_{-\infty}^{x_j}h_{\nu, P}(x_1,\dots,x_{j-1},z_j)\,dz_j=a\left(P^{-1}_{j,j}\frac{\nu+j-1}{\nu+g}\right)^{-1/2}t_{\nu+j-1}(l(x_j)).
  \end{align*}
  Consequently,
  \begin{align*}
    C(u_j\,|\,u_1,\dots,u_{j-1})=\sqrt{\frac{|P_{1:(j-1),1:(j-1)}|}{|P| P^{-1}_{j,j}}}t_{\nu+j-1}(l(x_j))=t_{\nu+j-1}(l(x_j)),
  \end{align*}
  where the last equality can be seen as follows. Let
  $P_{j|1:(j-1)}$ be as in \eqref{eq:cond:P}. Since $|P| = |P_{[1:(j-1),1:(j-1)]}| |P_{j|1:(j-1)}|$, and $P_{j|1:(j-1)} = (P^{-1}_{j,j})^{-1}$ by \cite[(0.7.3.1)]{matrixhorn}, we then have
  \begin{align*}
    |P|=|P_{1:(j-1),1:(j-1)}|\cdot|P_{j|1:(j-1)}|=|P_{1:(j-1),1:(j-1)}|/P^{-1}_{j,j}.
  \end{align*}
\end{proof}

\iffalse
\begin{proof}[Proof of Proposition~\ref{prop:Koksma-Hlawka_variable}]
  Follows from the inequality \eqref{eq:Koksma-Hlawka}, equality~\eqref{eq:canonical_rep} and the fact that $\Psi(\bm{u}_i)=\Psi(\phi_C(\bm{v}_i))$.
\end{proof}
\fi

\begin{proof}[Proof of Proposition \ref{prop:BdErrorArchMO}]
  We start by providing more details on the expression \eqref{eq:FaaDiBruno}, which is given by:
 \begin{align*}
& \frac{\partial^l \Psi \circ \phi_C(v_{\alpha_1},\ldots,v_{\alpha_l},\bm{1})}{\partial v_{\alpha_1} \cdots \partial v_{\alpha_l}}
=\\
& \sum_{1 \le |\bm{\beta}| \le l} \frac{\partial^{|\bm{\beta}|} \Psi }{\partial^{\beta_1} u_{1} \ldots \partial^{\beta_d} u_{d}}
 \sum_{s=1}^l \sum_{(\bm{k},\bm{\gamma}) \in p_s(\bm{\beta},\bm{\alpha})}
 c_{\gamma} \prod_{j=1}^s \frac{\partial^{|\gamma_j|} \phi_{C,k_j}(v_{\alpha_1},\ldots,v_{\alpha_l},\bm{1})}{\partial^{\gamma_{j,1}} v_{\alpha_1}\ldots \partial^{\gamma_{j,l}} v_{\alpha_l}}
 % \sum_{k=1}^l
 %\sum_{v_1 \in \bm{\alpha}} \ldots \sum_{v_k \in \bm{\alpha}} \left| \frac{\partial \Psi(u_{\alpha_1},\ldots,u_{\alpha_l},\bm{1}}{\partial u_{v_1}\ldots \partial u_{v_k}} \right||P_l(\partial \phi_{C,v_1}\ldots\phi_{C,v_k})|
 \end{align*}
 where $\bm{\beta} \in \mathbb{N}_0^d$,   $|\bm{\beta}| = \sum_{j=1}^d \beta_j$, and the set  $ p_s(\bm{\beta},\bm{\alpha})$ includes pairs $(\bm{k},\bm{\gamma})$ such that $\bm{k}$ is an $s$-dimensional vector $\bm{k}=(k_1,\ldots,k_s)$ where each $k_j\in\{1,\ldots,d\}$, and $\bm{\gamma}$ is an $sl$-dimensional vector $\bm{\gamma}=(\bm{\gamma}_1,\ldots,\bm{\gamma}_s)$ where each $\bm{\gamma}_j$ is an $l$-dimensional vector whose entries are either 0 or 1, and $\sum_{j=1}^s {\gamma}_{j,i} = 1$ for $i\in\{1,\ldots,l\}$. Finally, the $c_{\bm{\gamma}}$ are constants, which  are defined in detail in \cite{const1996}, along with further information on the precise definition of $p_s(\bm{k},\bm{\gamma})$.  %where $P_l$ is a polynomial of degree $l$ in the partial derivative $\partial \phi_{C,v_1},\ldots,\partial \phi_{C,v_k}$.

As mentioned in Section \ref{subsec:composition}, a sufficient condition to show that $ \|\Psi \circ \Phi_C \|_{d,1} < \infty$ is to establish that all products  of the form \eqref{eq:PartDerivL1Phi_C_Gen} are in $L_1$, which we recall is given by
 \begin{align*}
% \label{eq:PartDerivL1Phi_C_Gen}
  \frac{\partial^{|\bm{\beta}|} \Psi }{\partial^{\beta_1} u_{1} \ldots \partial^{\beta_d} u_{d}}
 %\sum_{s=1}^l \sum_{(\bm{k},\bm{\gamma}) \in p_s(\bm{\beta},\bm{\alpha})}
 \prod_{j=1}^s \frac{\partial^{|\gamma_j|} \phi_{C,k_j}(v_{\alpha_1},\ldots,v_{\alpha_l},\bm{1})}{\partial^{\gamma_{j,1}} v_{\alpha_1}\ldots\partial^{\gamma_{j,l}} v_{\alpha_l}},
 % \frac{\partial^{|\bm{\beta}|} \Psi }{\partial^{\beta_1} u_{1} \ldots \partial^{\beta_d} u_{d}}
 %\prod_{j=1}^d \frac{\partial^{|\gamma_j|} \phi_{C,j}(v_{\alpha_1},\ldots,v_{\alpha_l},\bm{1})}%{\partial^{\gamma_{j,1}} v_{\alpha_1}\ldots\partial^{\gamma_{j,l}} v_{\alpha_l}}
%% \left| \frac{\partial^l \phi_{C,v}}{\partial u_{j_1} \ldots \partial u_{j_l}}\right| \le K(d,l)
% %\quad 1 \le l \le d; 1 \le j_i \le d
 \end{align*}
 for $s\in\{1,\ldots,l\}$ and $(\bm{k},\bm{\gamma}) \in p_s(\bm{\beta},\bm{\alpha})$.
% valid for all points in $[0,1)^d$ \christiane{check}, where $K(d,l)$ is a constant that depends only on $d$ and $l$.

Recall also that for the MO algorithm,  $\phi_{C,l}$ is a function of $v_1$ and $v_{l+1}$ only, for $l=1,\ldots,d$. Hence the only partial derivatives of $\phi_{C,l}$ that are nonzero are those with respect to variables in $\{v_1,v_{l+1}\}$.

%We need to show that the conditions given by  \eqref{eq:PartDerivL1Phi_C_Gen} for a general transformation $\phi_C$ hold in the case of the Marshall--Olkin algorithm, under the conditions stated in the proposition.
Now, since we assume that \eqref{eq:CondPartPhiMO} holds, then it means we just need to show that the product found in \eqref{eq:PartDerivL1Phi_C_Gen}  is in $L_1$, under the conditions stated in the proposition. In turn, we first show that this holds if the following bounds hold for the mixed partial derivatives of $\phi_C$:
 \begin{align}
 &\int_0^1 \left| \frac{\partial \phi_{C,l}(v_1=1,v_{l+1})}{\partial v_{l+1}}\right| dv_{l+1} < \infty,
 \label{cond1}\\
 &\int_{[0,1)^{l}} \left|\frac{\partial^2 \phi_{C,1}(v_1,v_2)}{\partial v_1 \partial v_2} \prod_{j=2}^{l-1} \frac{\partial \phi_{C,j}(v_1,v_{j+1})}{\partial v_{j+1}}\right| dv_1 dv_2 \ldots dv_{l}
<  \infty,  \mbox{ and } \label{cond2}\\
& \int_{[0,1)^{l-1}} \left| \frac{\partial \phi_{C,r}(v_1,v_{r+1}=1)}{\partial v_1} \right| \!\!\left(\prod_{j=1, j\neq r}^{l-1} \left|\frac{\partial \phi_{C,j}(v_1,v_{j+1})}{\partial v_{j+1}}\right| dv_j \right) \!\!dv_{l}
 < \infty \label{cond3}
 \end{align}
 for all $l \le d+1$,

%First, we assume w.l.o.g.\  that $\alpha_j=j$ for $j=1,\ldots,l$, that is, $I=\{1,\ldots,l\}$. (We can do this because the copula is exchangeable.)
We have three cases to consider.

\noindent {\em Case 1: $1 \notin I$.} Then  the product in \eqref{eq:PartDerivL1Phi_C_Gen}  is given by
\begin{align*}
\prod_{j=1}^l\left| \frac{\partial \phi_{C,j}(v_1=1,v_{j+1})}{\partial v_{j+1}}\right|,
\end{align*}
where we assumed w.l.o.g.\  that $I = \{2,\ldots,l+1\}$, $s=l$ and $k_j=j+1$ for $j\in\{1,\ldots,s\}$.
Since each term in the product depends on a distinct variable, the product
 is in $L_1$ if \eqref{cond1} holds.

\noindent {\em Case 2: $1\in I$ and $j$ such that $\gamma_{j,1}=1$ has $k_j+1 \notin I$.} This case can be analyzed w.l.o.g.\  by assuming $I$ is of the form $I=\{1,\ldots,r,r+2,\ldots,l+1\}$ for some $r\ge 1$. In that case,  the products in \eqref{eq:PartDerivL1Phi_C_Gen} are of the form
\begin{align*}
\left| \frac{\partial \phi_{C,r}(v_1,v_{r+1}=1)}{\partial v_1} \prod_{j=1, j\neq r}^{l-1} \frac{\partial \phi_{C,j}(v_1,v_{j+1})}{\partial v_{j+1}}\right|
\end{align*}
and is thus in $L_1$ as long as \eqref{cond3} holds.

\noindent {\em Case 3: $1\in I$ and $j$ such that $\gamma_{j,1}=1$ has $k_j+1 \in I$.}
In this case, we can assume w.l.o.g.\  that $I=\{1,\ldots,l\}$ and therefore the products in \eqref{eq:PartDerivL1Phi_C_Gen} are of the form
\begin{align*}
\left| \frac{\partial^2 \phi_{C,r}(v_1,v_{r+1})}{\partial v_1\partial v_{r+1}} \prod_{j=1, j\neq r}^{l-1} \frac{\partial \phi_{C,j}(v_1,v_{j+1})}{\partial v_{j+1}}\right|
\end{align*}
and is thus in $L_1$ as long as \eqref{cond2} holds.

 The last part of the proof is to show that %when the partial derivatives involved in
 \eqref{cond1}, \eqref{cond2}, and \eqref{cond3}  hold. First we study the partial derivatives involved in these expressions and find they are given by:
 %
 %For \eqref{cond1} and \eqref{cond3}, do not change sign over the domain of their respective integral, then these integrals simplify and amount to take differences/sums of  $\phi_{C,r}(\cdot,\cdot)$ at different values over its domain, which obviously yields a finite value since  $\phi_{C,r}(\cdot,\cdot)$ always takes values in $[0,1]$.
 %
 %These partial derivatives are given by:
 \begin{align*}
 \frac{\partial \phi_{C,1}(v_1,v_2)}{\partial v_1} &= \psi'\left(-\frac{\log(v_2)}{x_1}\right) \frac{\log v_2}{x_1^2} \frac{\partial x_1}{\partial v_1}, \\
  \frac{\partial \phi_{C,1}(v_1,v_2)}{\partial v_2} &= -\psi'\left(-\frac{\log(v_2)}{x_1}\right) \frac{1}{x_1 v_2},\\
   \frac{\partial^2 \phi_{C,1}(v_1,v_2)}{\partial v_1 \partial v_2} &=
   \frac{\partial x_1}{\partial v_1} \frac{1}{v_2x_1^2}\left[\psi'\left(-\frac{\log v_2}{x_1}\right)-\frac{\log v_2}{x_1} \psi''\left(-\frac{\log v_2}{x_1}\right)\right],
 \end{align*}
 where $x_1 = F^{-1}(v_1)$ and $\frac{\partial x_1}{\partial v_1}= 1/f(x_1)$, where $f$ is the pdf corresponding to $F$, which exists since we assumed $F$ was continuous. Now, the partial derivatives with respect to either $v_1$ or $v_2$ are clearly non-negative for all $v_1$ and $v_2$.
 Hence it is easy to see that  \eqref{cond1} and  \eqref{cond3} hold, because we can remove the absolute value inside the integrals and therefore, these integrals amount to take differences/sums of  $\phi_{C,r}(\cdot,\cdot)$ at different values over its domain, which obviously yields a finite value since  $\phi_{C,r}(\cdot,\cdot)$ always takes values in $[0,1]$.

 As for the mixed partial derivative with respect to $v_1$ and $v_2$, our assumption on $\psi'(t)+t\psi''(t)$ implies we have at most one sign change over the domain of the integral. If there is no sign change, the argument used in the previous paragraph to handle \eqref{cond1} and \eqref{cond3} can be used to show \eqref{cond2} is bounded. If there is one sign change, then we let  $t^*$ be  such that
\begin{align*}
  \psi'(t)+t \psi''(t) \le 0\quad\text{for } 0 \le t \le  t^*\text{ and }\psi'(t)+t \psi''(t) \ge 0\quad\text{for } t^* \ge  t.
\end{align*}
Then let $q(v)$ be a function such that $-\log q(v)/F^{-1}(v) = t^*$. For instance, one can verify that for the Clayton copula, $q(v)=e^{-\theta F^{-1}(v)}$. When integrating the absolute value of the mixed partial derivative $\partial^2 \phi_{C,1}(v_1,v_2)/\partial v_1 \partial v_2$, we get
\begin{align}
&\int_0^1\!\!\frac{\partial x_1}{\partial v_1} \frac{1}{x_1^2}\!\left[\int_0^{q(v_1)}\!\!\frac{1}{v_2}\!\left(\!\psi'\left(-\frac{\log v_2}{x_1}\right)-\frac{\log v_2}{x_1}\psi''\left(-\frac{\log v_2}{x_1}\right)\!\right)dv_2 \right.\notag \\
&\left.+ \int_{q(v_1)}^1\frac{1}{v_2}\left(-\psi'\left(-\frac{\log v_2}{x_1}\right)+\frac{\log v_2}{x_1}\psi''\left(-\frac{\log v_2}{x_1}\right)\right) dv_2\right]dv_1 \notag\\
=&\,2\int_0^1 \frac{\partial x_1}{\partial v_1} \frac{1}{x_1^2} \left[\psi'(-\log q(v_1)/x_1) \log q(v_1)\right]dv_1 \notag \\
=&-2t^* \psi'(t^*) \int_0^1 \frac{1}{F^{-1}(v_1)} \frac{\partial F^{-1}(v_1)'}{\partial v_1} dv_1
=-2t^* \psi'(t^*)  \log F^{-1} (v_1)\bigl|_0^1. \label{eq:Finv}
\end{align}
Now, in most cases $F^{-1}(1)$ is not bounded, and thus we cannot prove that
$\Psi \circ \phi_C$ has bounded variation. However, from there we can still get
the upper bound on the error given in the result, by using a technique initially
developed by  \cite{Sobol73} to handle improper integrals, and later by
\cite{HartKeinTichy2004} to deal with unbounded integration problems taken
w.r.t.\ to a measure that is not necessarily uniform (as studied in Section \ref{subsec:TransfSamples}). Note that to apply their approach more easily, we need to make a small change and assume that rather than generating $V$ as $F^{-1}(v_1)$, we use $F^{-1}(1-v_1)$, so that in our study of the variation above (via the integral of the absolute value of the mixed partial derivatives), the boundedness condition fails at   $v_1=0$ instead of $v_1=1$. Following the approach in \cite{HartKeinTichy2004} (see their Equation (24)) and taking $\bm{c}=(1/pn,0,\ldots,0)$, the integration error satisfies
\begin{align*}
  &\phantom{{}={}}\biggl|\frac{1}{n} \sum_{i=1}^n \Psi(\bm{u}_i) -\IE[\Psi(\bm{U})]\biggr|\\
  &\le \frac{1}{pn} \Psi(1,\ldots,1) +D^*(P_n) V_{[\bm{c},\bm{1}]}(\Psi \circ \phi_C) + I_{rest}
\end{align*}
where $V_{[\bm{c},\bm{1}]}(\Psi \circ \phi_C)$ denotes the variation of $\Psi \circ \phi_C$ over $[\bm{c},\bm{1}]$ and
\begin{align*}
I_{rest} = \left|\int_{\bm{0}}^{\bm{1}} \Psi \circ \phi_C(\bm{v}) d\bm{v} - \int_{\bm{c}}^{\bm{1}} \Psi \circ \phi_C(\bm{v}) d\bm{v}  \right| \le \frac{M}{pn} \mbox{ for some } M>0,
\end{align*}
since we assumed $|\psi(\bm{u})|$ was bounded. As for $V_{[\bm{c},\bm{1}]}(\Psi \circ \phi_C)$, we can infer from the steps that led to \eqref{eq:Finv} that it is bounded by a constant times $\log F^{-1}(1-1/pn) \le a \log n + \log c$ by assumption. Therefore there exists a constant $K^{(d)}$ such that $V_{[\bm{c},\bm{1}]}(\Psi \circ \phi_C) \le K^{(d)} \log n$.
\end{proof}

\begin{proof}[Proof of Proposition \ref{prop:MOdisc}]
Let $p_l$ be such that $P(V=l)=p_l$, for $l \ge 1$. Let $P_l = \sum_{k=1}^l p_k$
for $l \ge 1$ and $P_0=0$. We also let $\phi_C^l(v_2,\ldots,v_{d+1})$ $= \phi_C(P_{l-1},v_2,\ldots,v_{d+1})$ for $l \ge 1$ (transformation $\phi_C$ when $v_1$ generates the value $l$ for $V$).
Consider a given value of $n$ and low-discrepancy point set $P_n$. If we use inversion to generate $V$, then we have that the subset $P_n^l = \{\bm{v}_i: P_{l-1} < v_{i,1} \le P_l\}$ will be used to produce copula samples with $V=l$. Let  $\tilde{n}_l = |P_n^l|$ and $n_l = np_l$. It is clear that if $l$ becomes too large, then $\tilde{n}_l$ will eventually be 0. Let $L(n)$ be the largest value of $l$ such that $\tilde{n}_{l} >0$, and let $\tilde{p}_l = \tilde{n}_l/n$. Then we can write
{\allowdisplaybreaks\begin{align*}
&\biggl| \int_{[0,1)^{d+1}} \Psi \circ \phi_C(\bm{v})d\bm{v} - \frac{1}{n} \sum_{i=1}^n \Psi \circ \phi_C(\bm{v}_i)\biggr| \\
\le&\biggl| \sum_{l=1}^{L(n)} p_l \left(\int_{[0,1)^d} \Psi \circ \phi_C^l (\bm{v}) dv_2 \ldots dv_{d+1}
- \frac{1}{n_l}\sum_{P_n^l}\Psi \circ \phi_C(\bm{v}_i) \right)\biggr| \\
&+ \sum_{l=L(n)+1}^{\infty} p_l \biggl|\int_{[0,1)^d} \Psi \circ \phi_C^l (\bm{v}) dv_2 \ldots dv_{d+1} \biggr| \\
\le& \sum_{l=1}^{L(n)} \tilde{p}_l \biggl| \int_{[0,1)^d} \Psi \circ \phi_C^l (\bm{v}) dv_2 \ldots dv_{d+1}
- \frac{1}{\tilde{n}_l}\sum_{P_n^l}\Psi \circ \phi_C(\bm{v}_i) \biggr|\\
&+  \sum_{l=L(n)+1}^{\infty} p_l \biggl|\int_{[0,1)^d} \Psi \circ \phi_C^l (\bm{v}) dv_2 \ldots dv_{d+1} \biggr|  \\
&+ \sum_{l=1}^{L(n)} \biggl|(p_l-\tilde{p}_l) \int_{[0,1)^d} \Psi \circ \phi_C^l (\bm{v}) dv_2 \ldots dv_{d+1} \biggr| \\
\le&  \sum_{l=1}^{L(n)}( \tilde{p}_l A(n,d)) + B(n,d)+C(n,d),
\end{align*}}%
where $A(n,d)$, $B(n,d)$, and $C(n,d)$ are bounds such that
\begin{align*}
\biggl|  \int_{[0,1)^d} \Psi \circ \phi_C^l (\bm{v}) dv_2 \ldots dv_{d+1}
- \frac{1}{\tilde{n}_l}\sum_{P_n^l}\Psi \circ \phi_C(\bm{v}_i) \biggr| &\le A(n,d) \\
\sum_{l=L(n)+1}^{\infty} p_l \biggl|\int_{[0,1)^d} \Psi \circ \phi_C^l (\bm{v}) dv_2 \ldots dv_{d+1} \biggr|
&\le B(n,d) \\
\sum_{l=1}^{L(n)} \biggl|(p_l-\tilde{p}_l) \int_{[0,1)^d} \Psi \circ \phi_C^l (\bm{v}) dv_2 \ldots dv_{d+1}\biggr|
&\le C(n,d).
\end{align*}

First, by definition of $D^*(P_n)$ we have $|\tilde{n}_l - n_l| \le 2nD^*(P_n)$ and thus
$|\tilde{p}_l-p_l| \le 2D^*(P_n)$. Hence we can take $C(n,d) = 2\E(|\Psi(\bm{U})|)D^*(P_n)$. Similarly, we can show that $\sum_{l=L(n)+1}^{\infty} p_l  \le D^*(P_n)$ and can therefore take $B(n,d) = \E(|\Psi(\bm{U})|) D^*(P_n)$. The analysis of the expression to be bounded by $A(n,d)$ is more complicated.
First, we note that under the assumption we have on $\Psi$ and its partial derivatives, we need to show that the product in \eqref{eq:PartDerivL1Phi_C_Gen} is in $L_1$, but where each $\phi_{C,k_j}$ is replaced by $\phi_{C,k_j}^l$ for a given $l$. Since  $\phi_{C,k_j}^l$ is solely a function of $v_{k_j+1}$, then it means that the only relevant products to consider are of the form
\begin{equation}
\label{eq:ProdDiscMO}
\prod_{j=1}^r \frac{\partial \phi_{C,k_j}^l(v_{k_j+1})}{\partial v_{k_j+1}}
\end{equation}
in which each term is of the form
$
-\psi^{'} \left(\frac{-\log v_{k_j+1}}{l}\right) \frac{1}{l v_{k_j+1}}
$
which is non-negative for any $v_{k_j+1}$. Using a similar reasoning to the one used in the proof of
Proposition \ref{prop:BdErrorArchMO} (to conclude that \eqref{cond1} and \eqref{cond3}  hold), it is easy to see that \eqref{eq:ProdDiscMO}  is in $L_1$.
%%%
%combined with and its proof, it is fairly easy to see that the function $\Psi \circ \phi_C^l$ has bounded variation in the sense of Hardy and Krause for each $l \ge 1$.
%%%

What remains to be done is to analyze the discrepancy of $P_n^l$. That is, here we are looking for a bound on $\sup_{\bm{z} \in {\cal J}^*} |E(\bm{z};P_n^l)|$, where we recall that $ {\cal J}^*$ is the set of intervals of $[0,1)^d$ of the form $\bm{z} = \prod_{j=1}^d [0,z_j)$, where $0 < z_j \le 1$. %, where ${\cal J}^*$ includes all intervals of the form $\prod_{j=2}^{d+1} [0,z_j)$.
So consider a given $\bm{z} \in [0,1)^d$. Then $E(\bm{z};P_n^l)  = A(\bm{z};P_n^l)/\tilde{n}_l - \lambda(\bm{z})$. Let $\bm{z}_1 = (P_l,\bm{z})$ and $\bm{z}_2 = (P_{l-1},\bm{z})$, which are both in   $[0,1)^{d+1}$. Note that $A(\bm{z}_1;P_n) - A(\bm{z}_2;P_n) = A(\bm{z};P_n^l)$. By definition of $D^*(P_n)$, it is not hard to see that
\begin{align*}
\biggl|\frac{(A(\bm{z}_1;P_n)-A(\bm{z}_2;P_n))}{n}-p_l \lambda(\bm{z})\biggr| \le 2D^*(P_n)
\end{align*}
and therefore
\begin{align*}
\biggl|\frac{A(\bm{z};P_n^l)}{\tilde{n}_l} -\frac{n_l}{\tilde{n}_l}\lambda(\bm{z})\biggr| \le 2D^*(P_n) \frac{n}{\tilde{n}_l}.
\end{align*}
Using the fact that $|\tilde{n}_l-n_l| \le 2nD^*(P_n)$, after some further simplifications we get that
\begin{align*}
\biggl|\frac{A(\bm{z};P_n^l)}{\tilde{n}_l} -\lambda(\bm{z})\biggr| \le 4D^*(P_n) \frac{n}{\tilde{n}_l}.
\end{align*}
Hence we can take $A(n,d) = 4D^*(P_n)\frac{n}{\tilde{n}_l}$ and then get $\sum_{l=1}^{L(n)} \tilde{p}_l A(n,d) \le 4 L(n)D^*(P_n)$. To show that the overall bound for the integration error is of the form $(\log n)D^*(P_n)$ times a constant, we simply need to show that $L(n) \in O(\log n)$. But this follows from our assumptions on $P_n$ and $F$, since by definition,  $L(n)$ is the largest integer such that
  $1-F(L(n)) > 1/pn$ but we also have $1-F(L(n)) \le cq^{L(n)}$, hence
  \[
  1/pn < cq^{L(n)} \Leftrightarrow  L(n) \log (1/q) - \log c <\log p + \log n
  \]
  and thus $L(n) \le (\log n + \log p + \log c)/\log(1/q)$,
  %\Leftrightarrow
%\log p + \log n > \log (1/(1-F(L(n)))).
%\]
  %Now $1-F(L(n)) \le cq^{L(n)}$ and $\log (1/(cq^l)) = -\log c + l \log (1/q)$ so $L(n) <( \log n+ \log p+\log c)/\log (1/q)$,
  as required.
\end{proof}

%%%%%%%%%%%%%%%%%%%%%%%%%%%%%%%%%%%%%%%%%%%%%%
\appendix
\section*{Online supplement}

\subsection*{Additional numerical results}
Here we provide a few additional results for the experimental setup described in
Section \ref{sec:num}, namely for the finance and insurance examples (see
Figures \ref{fig:clayton-t:allocations:mid} and
\ref{fig:clayton-t:allocations:first}) and then for the test functions (see
Figures \ref{fig:test:tau:0.5:d:5} and \ref{fig:test:tau:0.5:d:15}).

\begin{figure}[htbp]
\centering
\includegraphics[width=0.425\textwidth]{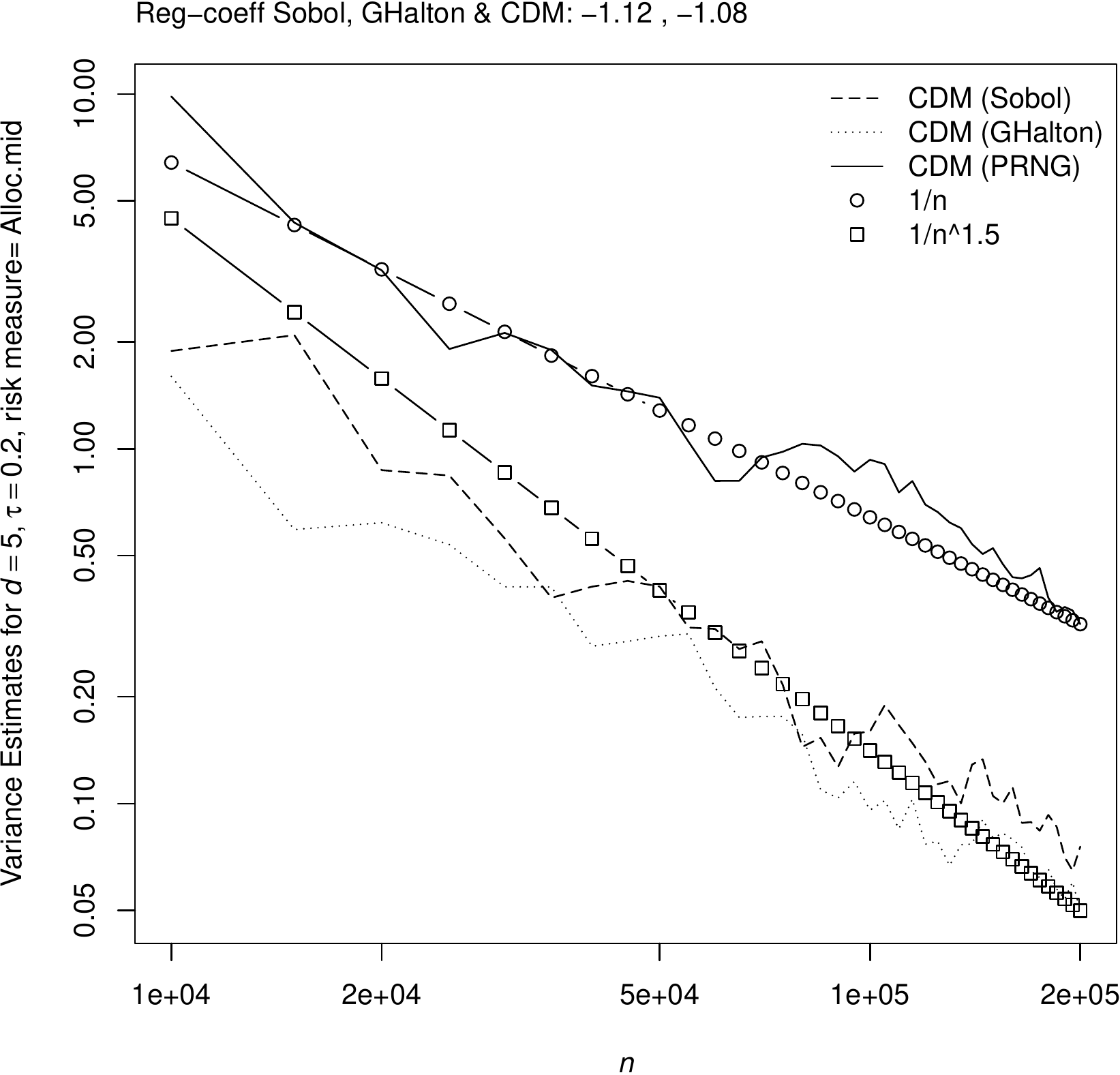}\\[2mm]
\includegraphics[width=0.425\textwidth]{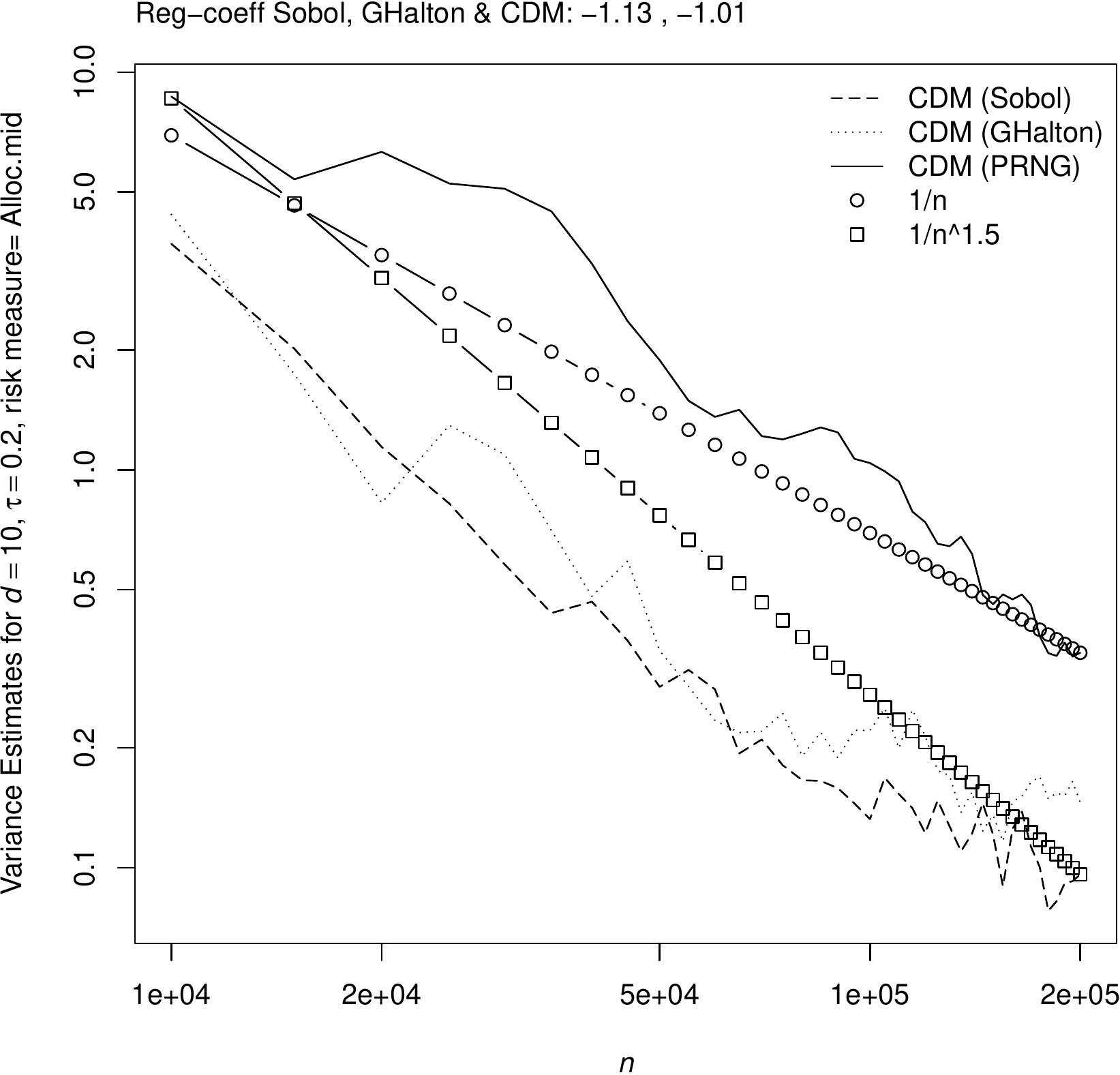}\\[2mm]
\includegraphics[width=0.425\textwidth]{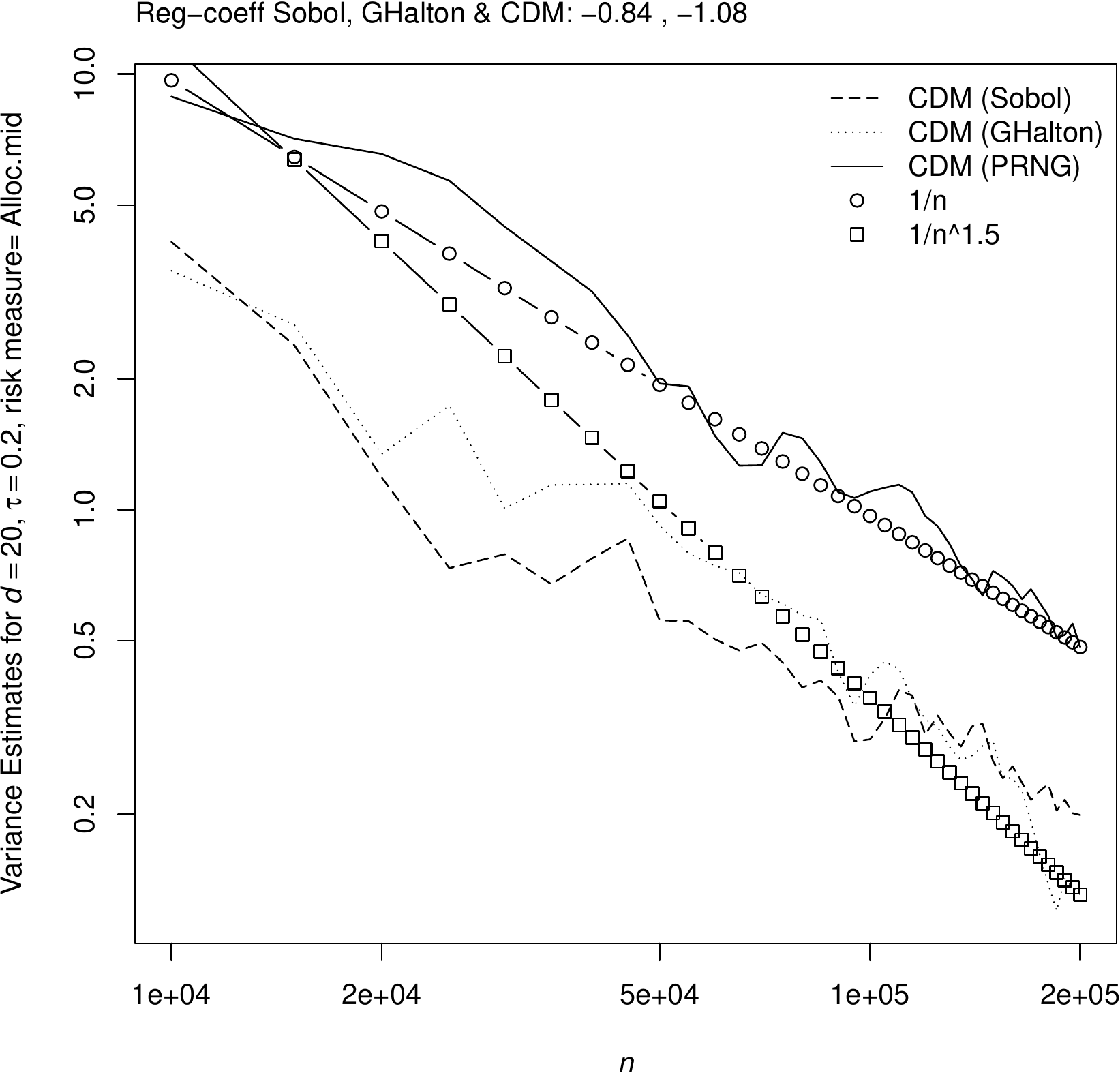}%
\caption{Variance estimates for the functionals Allocation First for lognormal margins and an
  exchangeable $t$ copula with three degrees of freedom such that Kendall's tau equals 0.2 based on $B = 25$ repetitions for $d=5$ (top), $d=10$ (middle) and
  $d=20$ (bottom).}
\label{fig:clayton-t:allocations:mid}
\end{figure}

\begin{figure}[htbp]
\centering
\includegraphics[width=0.425\textwidth]{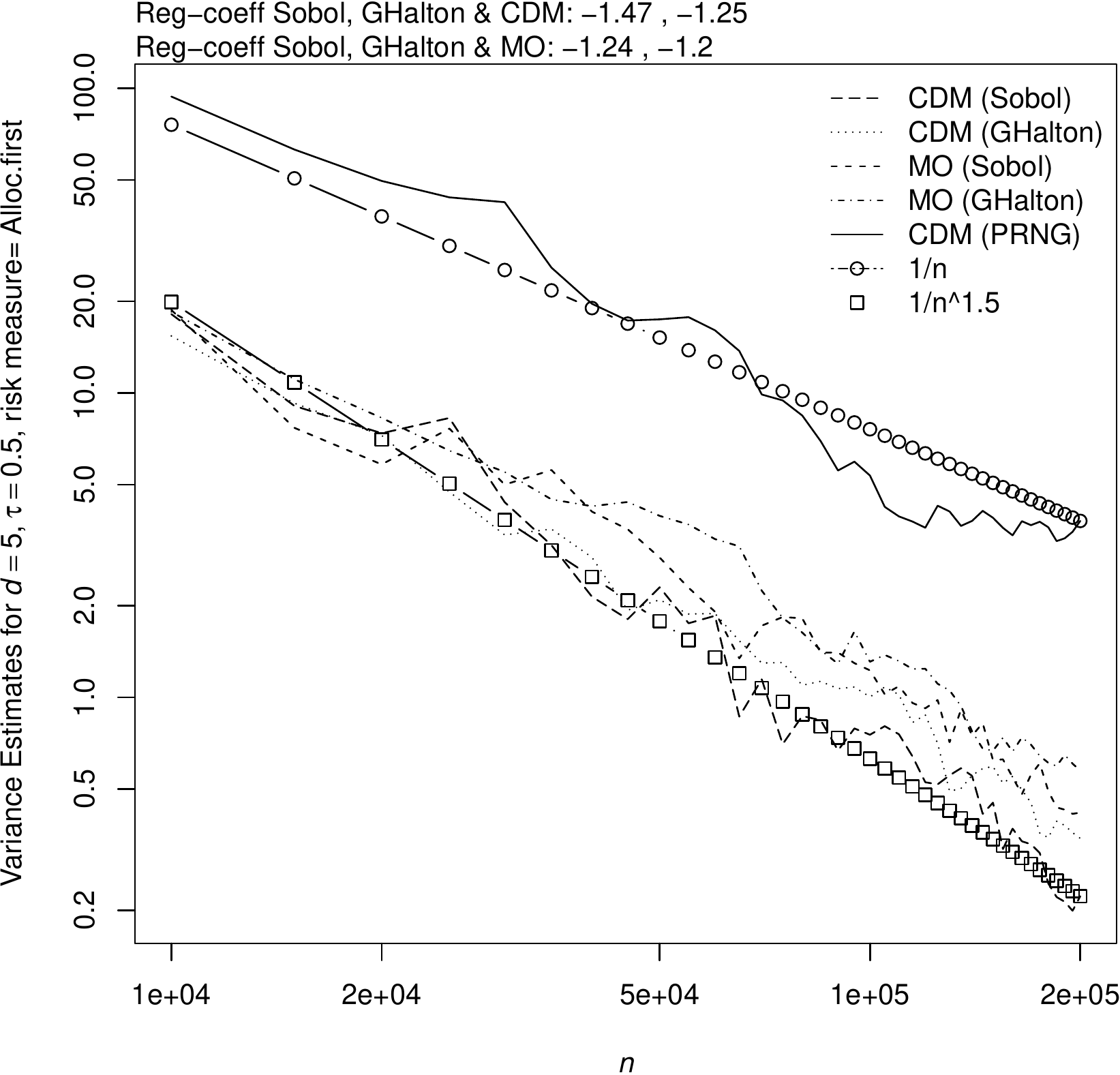}\\[2mm]
\includegraphics[width=0.425\textwidth]{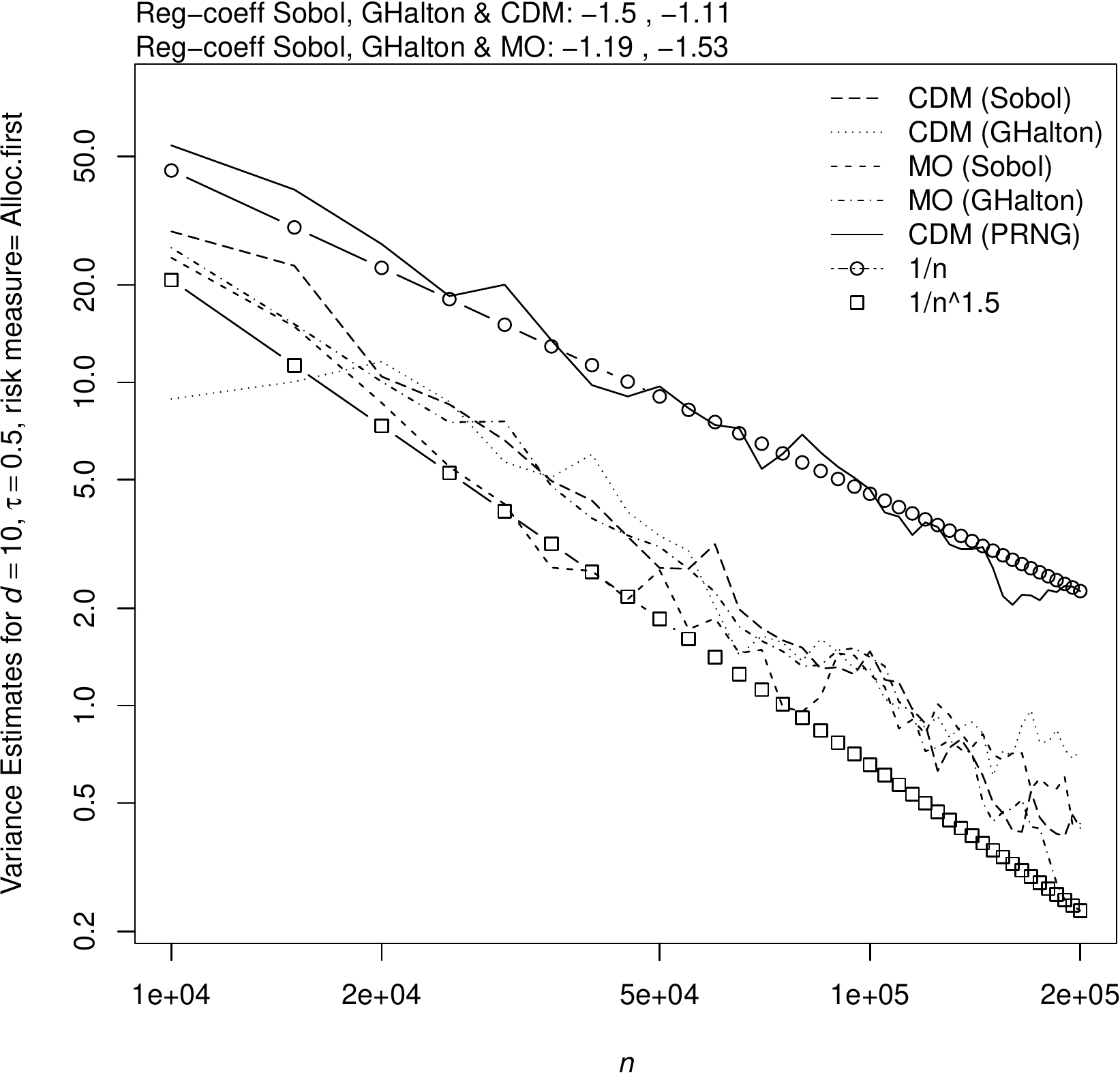}\\[-0.8mm]
\includegraphics[width=0.425\textwidth]{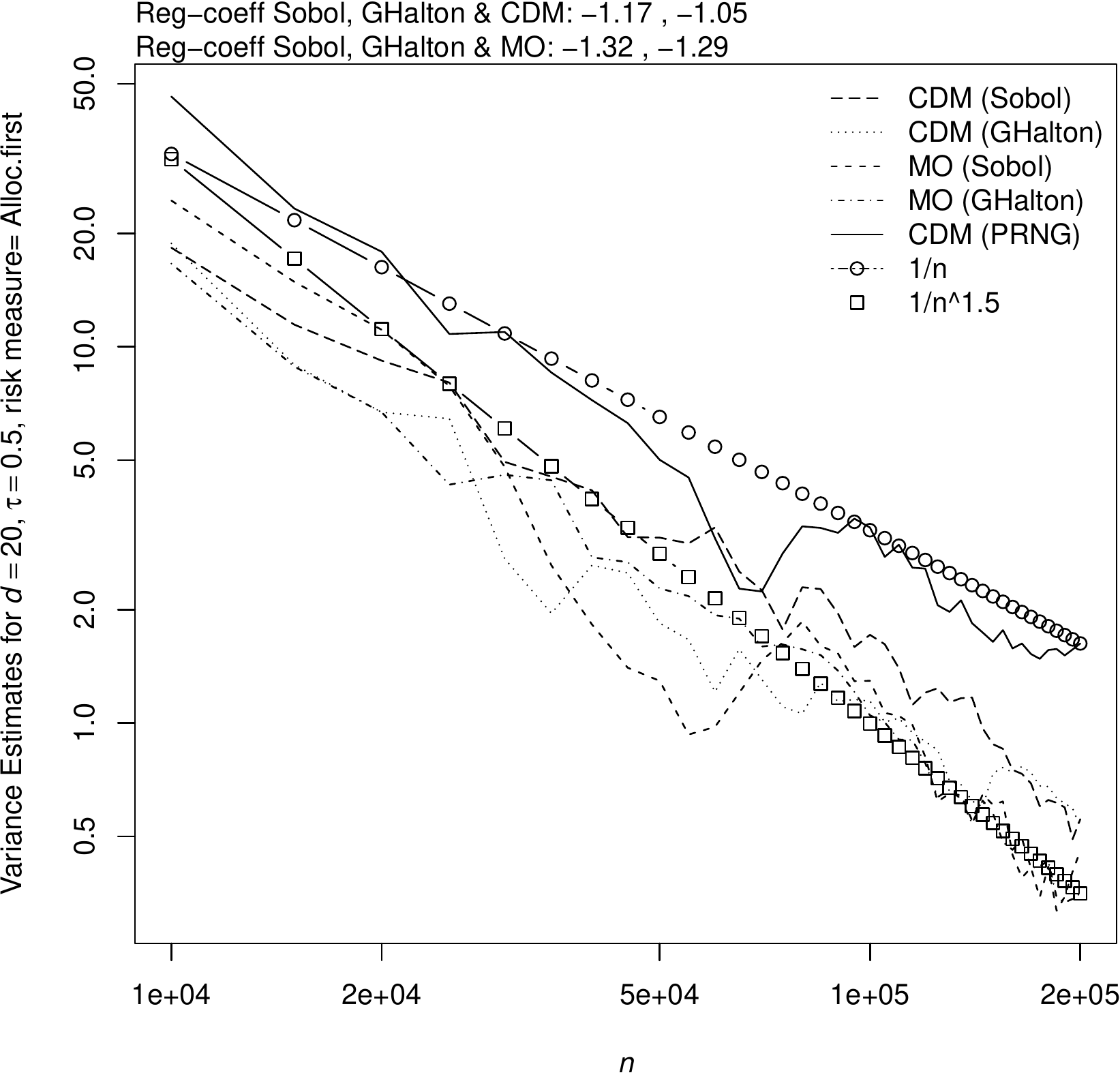}%
\caption{Variance estimates for the functionals Allocation Middle with Pareto margins and for a Clayton copula with parameter such that Kendall's tau
  equals 0.5 based on $B = 25$ repetitions  for $d=5$ (top), $d=10$ (middle) and
  $d=20$ (bottom).}
\label{fig:clayton-t:allocations:first}
\end{figure}

\begin{figure}[htbp]
\centering
\includegraphics[width=0.425\textwidth]{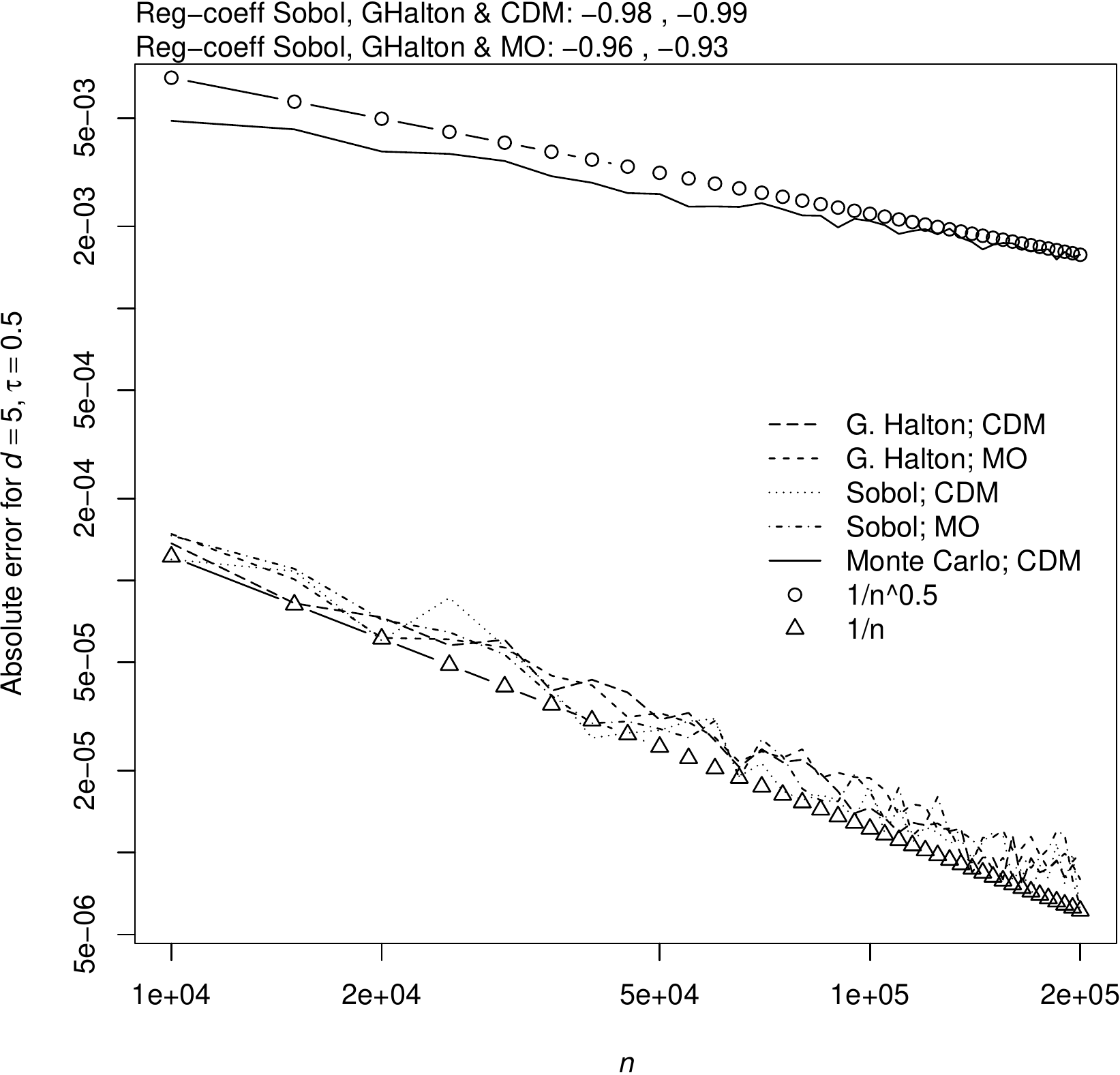}\\[3.3mm]
\includegraphics[width=0.425\textwidth]{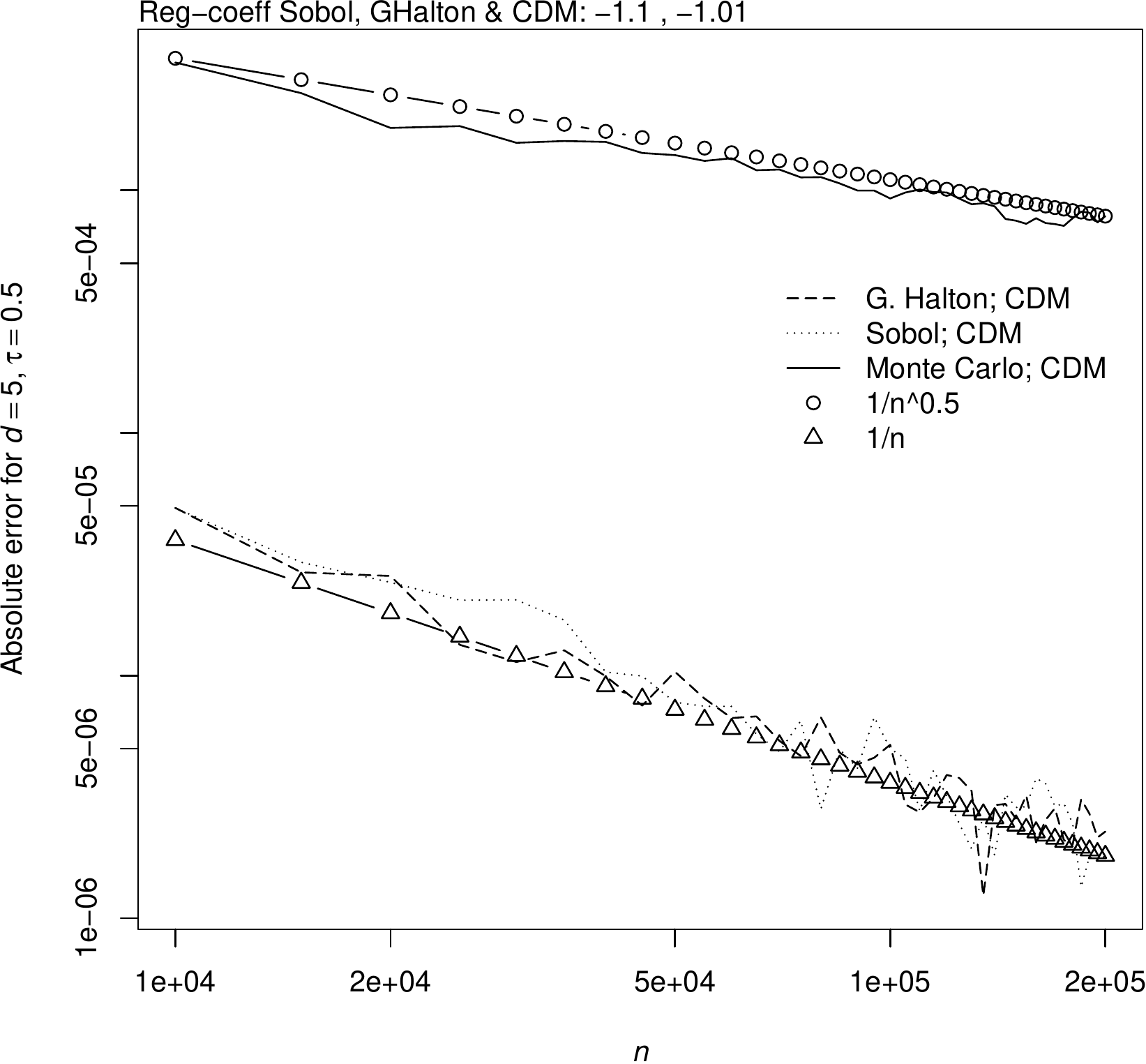}\\[-0.8mm]
\includegraphics[width=0.425\textwidth]{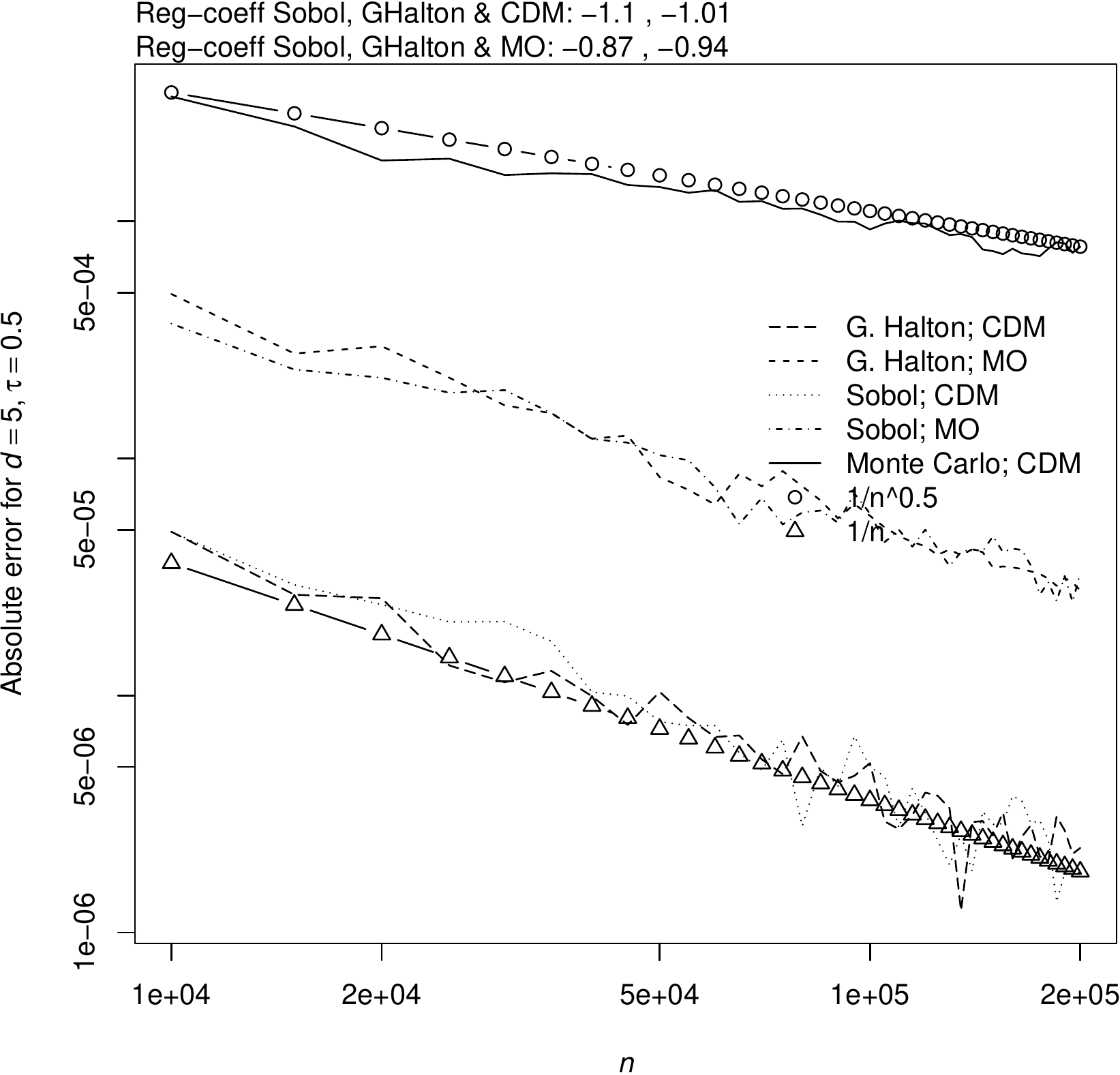}%
\caption{Average absolute errors for the test functions
    $\Psi_1(\bm{u})=3(u_1^2 + \ldots + u_d^2)/d$ (top) and $\Psi_2(\bm{u})=g_1((\phi^{\text{CDM}})^{-1}(\bm{u}))$ (bottom) for a Clayton copula with parameter such that Kendall's tau
  equals 0.5 based on $B = 25$ repetitions  for $d=5$: the middle plot shows results for $\Psi_1(\bm{u})$ and an exchangeable $t$ copula with three degrees of freedom and Kendall's tau of 0.2.}
\label{fig:test:tau:0.5:d:5}
\end{figure}

\begin{figure}[htbp]
\centering
\includegraphics[width=0.425\textwidth]{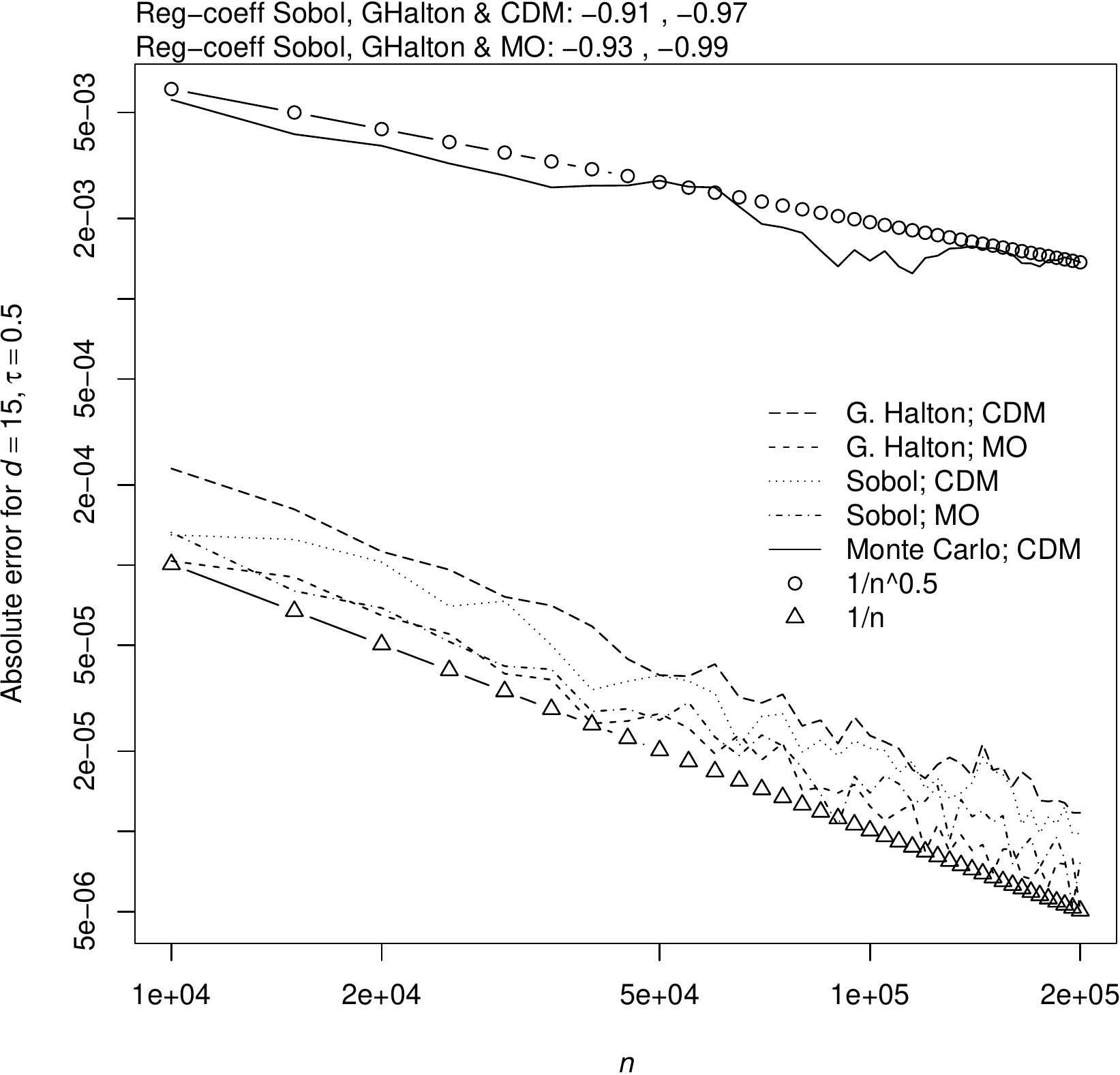}\\[2mm]
\includegraphics[width=0.425\textwidth]{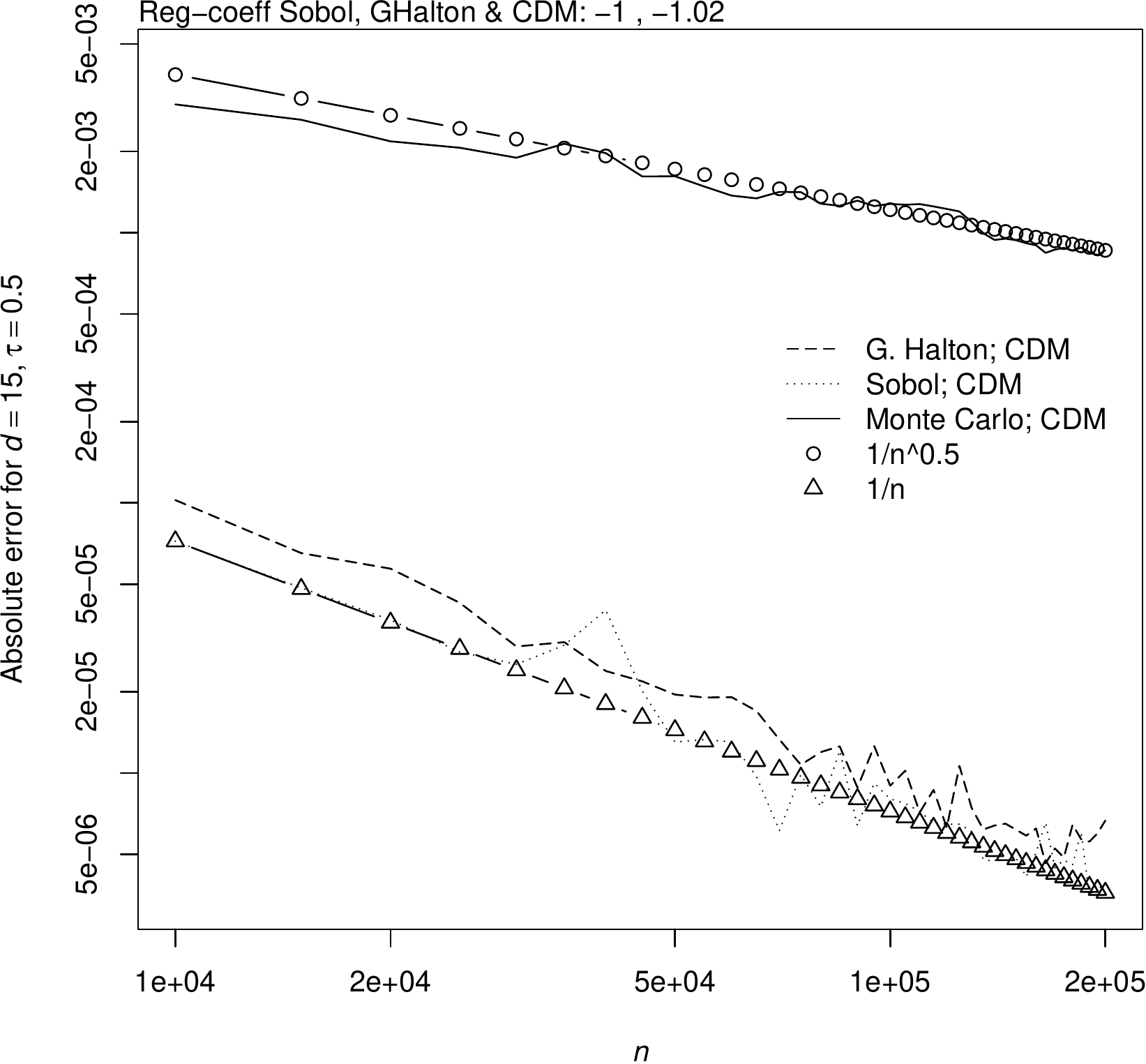}\\[-0.8mm]
\includegraphics[width=0.425\textwidth]{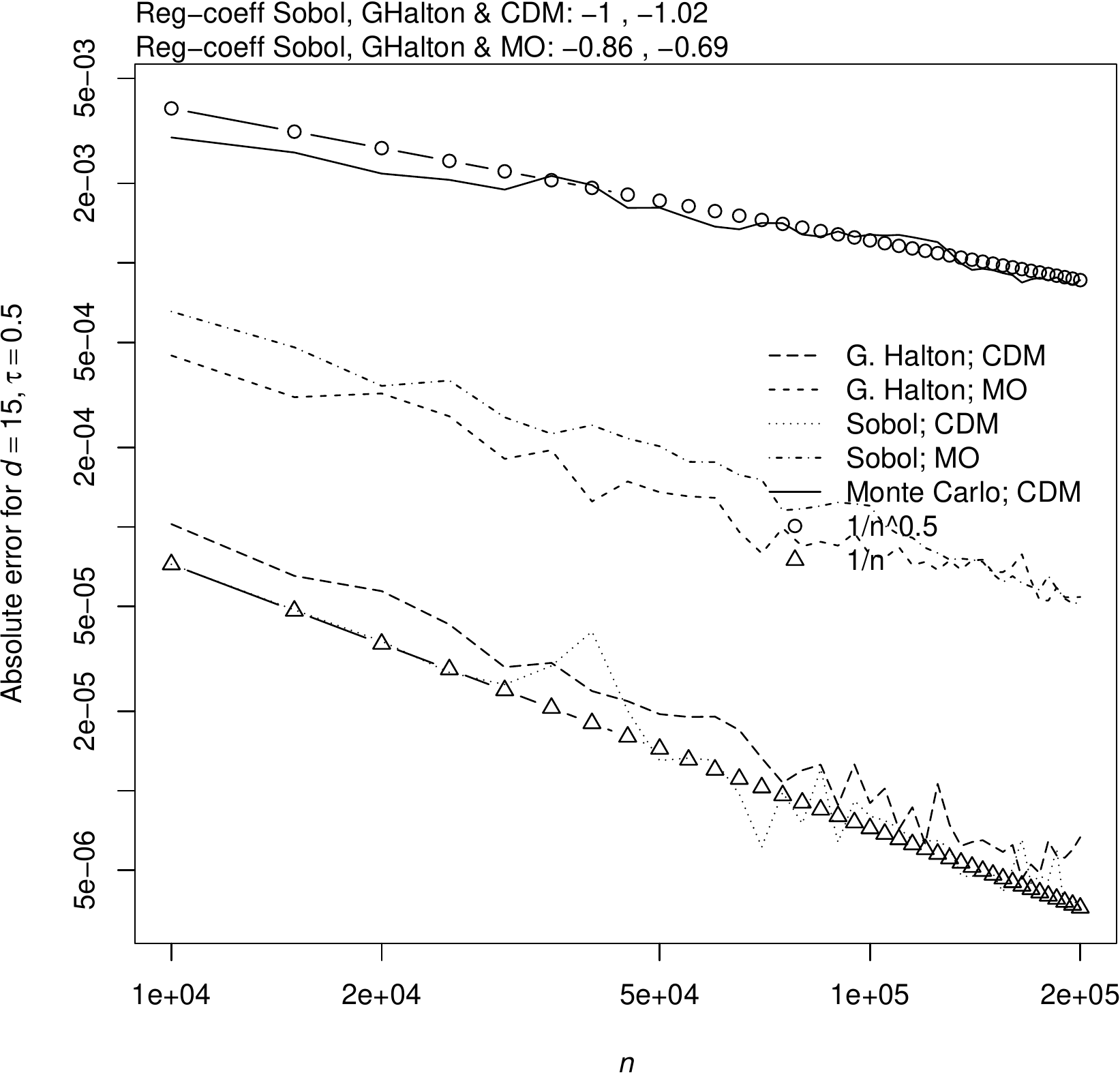}%
\caption{Average absolute errors for the test functions
    $\Psi_1(\bm{u})=3(u_1^2 + \ldots + u_d^2)/d$ (top), $\Psi_2(\bm{u}) =g_1((\phi^{\text{CDM}})^{-1}(\bm{u}))$ (bottom) for a Clayton copula with parameter such that Kendall's tau
  equals 0.5 based on $B = 25$ repetitions  for $d=15$: the middle plot shows results for $\Psi_1(\bm{u})$ and an exchangeable $t$ copula with three degrees of freedom and Kendall's tau of 0.2.}
\label{fig:test:tau:0.5:d:15}
\end{figure}

\end{document}

%%% Local Variables:
%%% mode: latex
%%% TeX-master: t
%%% End: